\documentclass[10pt,journal]{IEEEtran}

\ifCLASSOPTIONcompsoc
  \usepackage[nocompress]{cite}
\else
  \usepackage{cite}
\fi
\ifCLASSINFOpdf
\else
\fi
\usepackage{todonotes}
\usepackage{graphicx}

\graphicspath{{./}}
\usepackage[listofformat=parens, justification=centering, subrefformat=subparens]{subfig}

\usepackage{hyperref}

\newcommand\Windows[1]{\texttt{#1}}
\newcommand\Name[1]{\textsf{#1}}

\renewcommand\cap[3]{\caption[#2]{\label{#1}\textsc{#2}. \small\textit{#3}}}

\newcommand{\defeq}{\mathrel{\mathop:}=}

\usepackage[T1]{fontenc}

\usepackage{amsmath}
\usepackage{amsfonts}
\usepackage{amssymb}
\usepackage{color}
\usepackage{caption}
\usepackage{tabularx}
\usepackage{amsthm}

\newtheorem{mythm}{Theorem}
\newtheorem{mycorr}{Corollary}

\newcommand\sad[3]{\begin{itemize}\item[] \textit{Summary:} #1 \item[] \textit{Advantages:} #2\item[] \textit{Disadvantages:} #3\end{itemize}}

\DeclareMathOperator*{\argmax}{\mathrm{arg\,max}}

\begin{document}

\title{A Survey of Stealth Malware Attacks, Mitigation Measures, and Steps Toward Autonomous Open World Solutions}

\author{Ethan~M.~Rudd, 
        Andras~Rozsa, 
        Manuel~G\"unther, 
        and~Terrance~E.~Boult
\IEEEcompsocitemizethanks{\IEEEcompsocthanksitem E. Rudd, A. Rozsa, M. G\"unther, and T. Boult are with the Vision and Security Technology (VAST) Lab, Department
of Computer Science, University of Colorado at Colorado Springs.\protect\\
E-mail: see http://vast.uccs.edu/contact-us}
\thanks{Manuscript received February 19, 2016. Revised August 21, 2016. Revised September 16th, 2016 . Accepted December 1, 2016.}}

\markboth{PRE-PRINT OF MANUSCRIPT ACCEPTED TO IEEE COMMUNICATION SURVEYS \& TUTORIALS}%
{Shell \MakeLowercase{\textit{et al.}}: Bare Advanced Demo of IEEEtran.cls for IEEE Computer Society Journals}

\IEEEtitleabstractindextext{%
\begin{abstract}
As our professional, social, and financial existences become increasingly digitized and as our government, healthcare, and military infrastructures rely more on computer technologies, they present larger and more lucrative targets for malware.
Stealth malware in particular poses an increased threat because it is specifically designed to evade detection mechanisms, spreading dormant, in the wild for extended periods of time, gathering sensitive information or positioning itself for a high-impact zero-day attack.
Policing the growing attack surface requires the development of efficient anti-malware solutions with improved generalization to detect novel types of malware and resolve these occurrences with as little burden on human experts as possible.

In this paper, we survey malicious stealth technologies as well as existing solutions for detecting and categorizing these countermeasures autonomously.
While machine learning offers promising potential for increasingly autonomous solutions with improved generalization to new malware types, both at the network level and at the host level, our findings suggest that several flawed assumptions inherent to most recognition algorithms prevent a direct mapping between the stealth malware recognition problem and a machine learning solution. The most notable of these flawed assumptions is the \emph{closed world} assumption: that no sample belonging to a class outside of a static training set will appear at query time. We present a formalized \emph{adaptive open world framework} for stealth malware recognition and relate it mathematically to research from other machine learning domains.
\end{abstract}

}

\maketitle

\IEEEdisplaynontitleabstractindextext
\IEEEpeerreviewmaketitle

\ifCLASSOPTIONcompsoc
\IEEEraisesectionheading{\section{Introduction}\label{sec:introduction}}
\else
\section{Introduction}
\label{sec:introduction}
\fi

\IEEEPARstart{M}{alwares} have canonically been lumped into categories such as viruses, worms, Trojans, rootkits, etc. Today's advanced malwares, however, often include many components with different functionalities. For example, the same malware might behave as a virus when spreading over a host, behave as a worm when propagating through a network, exhibit \emph{botnet} behavior when communicating with command and control (C2) servers or synchronizing with other infected machines, and exhibit \emph{rootkit} behavior when concealing itself from an intrusion detection system (IDS).
A thorough study of all aspects of malware is important for developing security products and computer forensics solutions, but stealth components pose particularly difficult challenges. The ease or difficulty of repairative measures is irrelevant if the malware can evade detection in the first place.

While some authors refer to all stealth malwares as  \emph{rootkits}, the term rootkit properly refers to the modules that redirect code execution and subvert expected operating system functionalities for the purpose of maintaining stealth. With respect to this usage of the term, rootkits deviate from other stealth features such as elaborate code mutation engines that aim to change the appearance of malicious code so as to evade signature detection without changing the underlying functionality.

As malwares continue to increase in quantity and sophistication, solutions with improved generalization to previously unseen malware samples/types that also offer sufficient diagnostic information to resolve threats with as little human burden as possible are becoming increasingly desirable.
Machine learning offers tremendous potential to aid in stealth malware intrusion recognition, but there are still serious disconnects between many machine learning based intrusion detection ``solutions'' presented by the research community and those actually fielded in IDS software. Robin and Paxson\cite{sommer2010outside} discuss several factors that contribute to this disconnect and suggest useful guidelines for applying machine learning in practical IDS settings. Although their suggestions are a good start, we contend that refinements must be made to machine learning algorithms themselves in order to effectively apply such algorithms to the recognition of stealth malware. Specifically, there are several flawed assumptions inherent to many algorithms that distort their mappings to realistic stealth malware intrusion recognition problems. The chief among these is the \emph{closed-world} assumption -- that only a fixed number of known classes that are available in the training set will be present at classification time.

Our contributions are as follows:
\begin{itemize}
\item \textit{We present the first comprehensive academic survey of stealth malware technologies and countermeasures.} There have been several light and narrowly-scoped academic surveys on rootkits\cite{li2011survey,kim2012brief,shields2008survey}, and many broader surveys on the problem of intrusion detection, e.g. \cite{axelsson2000intrusion,vasilomanolakis2015taxonomy,zuech2015intrusion}, some specifically discussing machine learning intrusion detection techniques\cite{tsai2009intrusion,garcia2009anomaly,lee1999data,sommer2010outside}. However, none of these works come close to addressing the mechanics of stealth malware and countermeasures with the level of technical and mathematical detail that we provide. Our survey is \emph{broader in scope} and \emph{more rigorous in detail} than any existing academic rootkit survey and provides not only detailed discussion of the mechanics of stealth malwares that goes far beyond rootkits, but an overview of countermeasures, with rigorous mathematical detail and examples for applied machine learning countermeasures.

\item \textit{We analyze six flawed assumptions inherent to many machine learning algorithms} that hinder their application to stealth malware intrusion recognition and other IDS domains.

\item \textit{We propose an adaptive open world mathematical framework for stealth malware recognition} that obviates the six inappropriate assumptions. Mathematical proofs of relationships to other intrusion recognition algorithms/frameworks are included, and the formal treatment of open world recognition is mathematically generalized beyond previous work on the subject.
\end{itemize}

Throughout this work, we will mainly provide examples for the \Name{Microsoft Windows} family of operating systems, supplementing where appropriate with examples from other OS types.
Our primary rationale for this decision is that, according to numerous recent tech reports from anti-malware vendors and research groups \cite{kaspersky2015,hpe2015,hpe2016,ibm2016,symantecistr2015,symantecistr2016,mcaffee2016,microsoft_security,mandiant_consulting}, malware for the \Name{Windows} family of operating systems is still far more prevalent than for any other OS type (cf. Fig.~\ref{fig:malware_proportions}). Our secondary rationale is that within the academic literature that we examined, we found comparatively little research discussing \Name{Windows} security. We believe that this gap needs to be filled. 
Note that many of the stealth malware attacks and defenses that apply to \Name{Windows} have their respective analogs in other systems, but each system has its unique strengths and susceptibilities.
This can be seen by comparing our survey to parts of\cite{faruki2015android}, in which Faruki et al. provide a survey of generic Android security issues and defenses.
Nonetheless, since our survey is about stealth malware; not exclusively \Name{Windows} stealth malware, we shall occasionally highlight techniques specific to other systems and mention discrepancies between systems.
\Name{Unix/Linux} rootkits shall also be discussed because the development of \Name{Unix} rootkits pre-dates the development of \Name{Windows}.
Any system call will be marked in \Windows{a special font}, while proper nouns are \Name{highlighted differently}.
A complete list of \Name{Windows} system calls discussed in this paper is given in Tab.~\ref{tab:SysCalls} of the appendix.

\begin{figure}[ht]
  \centering
  \subfloat[\label{fig:malware:proportions}Proportion of Malware by Platform Type]{\includegraphics[width=\linewidth]{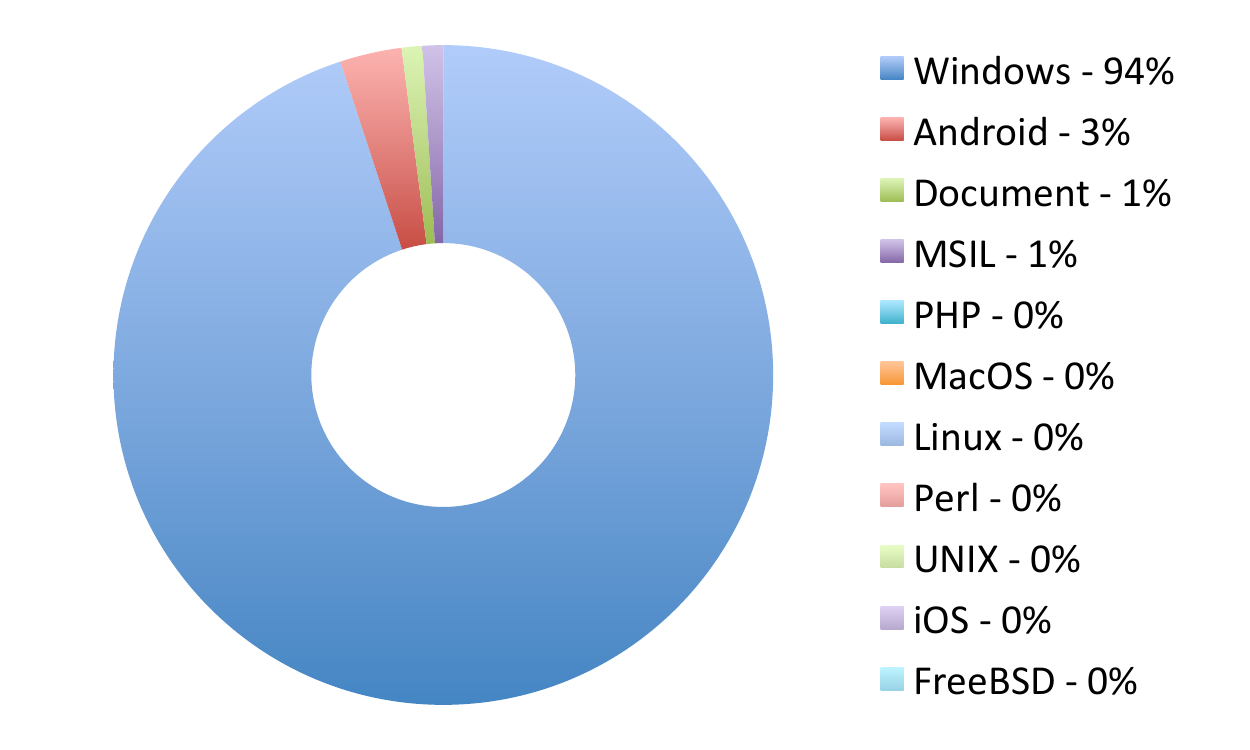}}\\
  \subfloat[\label{fig:malware:rates}Growth Rate in Malware Proportion by Platform Type]{\includegraphics[width=\linewidth]{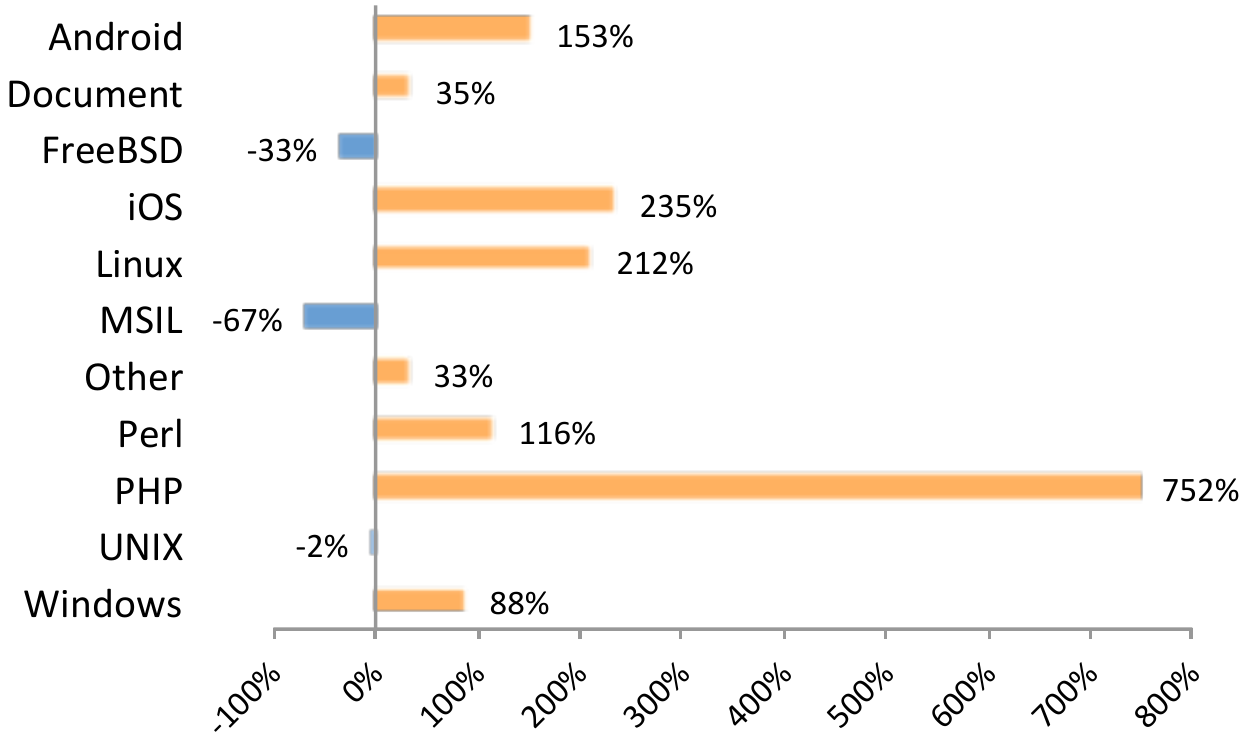}}
  \cap{fig:malware_proportions}{Malware proportions and rates by platform}{
As shown in \protect\subref*{fig:malware:proportions}, malware designed specifically for the \Name{Microsoft Windows} family of operating systems accounts for over 90\,\% of all malware. While malware for other platforms is growing rapidly, at current growth rates, shown in \protect\subref*{fig:malware:rates}, the quantity of malware designed for any other platform is unlikely to surpass the quantity of \Name{Windows} malware any time soon. Examining overall growth rates per platform, we see higher growth rates in non windows malware, but a high growth rate on a small base is still quite small in terms of overall impact. For \Name{Windows} the 88\% growth rate is on a base of 135 million malware samples, which translates into about 118 million new \Name{Windows} malwares. In comparison the high growth rate for \Name{Apple iOS}, with an increase of more than 230\%, is on a base of 30,400 samples, with the total number of discovered \Name{Apple} malware samples in 2015 just under 70,000; very small compared to the number of \Name{Windows} malware samples as well as the 4.5 milion \Name{Android} malware samples.
Numbers for these plots were obtained from a 2016 HP Enterprise threat report \cite{hpe2016}.
}
\end{figure}

The remainder of this paper is structured as follows:
In Sec.~\ref{sec:stealth_survey} we present the problems inherent to stealth malware by providing a comprehensive survey of stealth malware technologies, with an emphasis on rootkits and code obfuscation.
In Sec.~\ref{sec:countermeasure_survey}, we discuss stealth malware countermeasures, which aim to protect the integrity of areas of systems known to be vulnerable to attacks.
These include network intrusion recognition countermeasures as well as host intrusion recognition countermeasures.
Our discussion highlights the need for these methods to be combined with more generic recognition techniques.
In Sec.~\ref{sec:signatures_heuristics}, we discuss some of these more generic stealth malware countermeasures in the research literature, many of which are based on machine learning.
In Sec.~\ref{sec:open_world_ids}, we identify six critical flawed algorithmic assumptions that hinder the utility of machine learning approaches for malware recognition and more generic IDS domains.
We then formalize an \emph{adaptive open world framework} for stealth malware recognition, bringing together recent advances in several areas of machine learning literature including intrusion detection, novelty detection, and other recognition domains.
Finally, Sec.~\ref{sec:conclusion} concludes this survey.

\section{A Survey of Existing Stealth Malware}
\label{sec:stealth_survey}

We discuss four types of stealth technology: rootkits, code mutation, anti-emulation, and targeting mechanisms.
Before getting into the details of each, we summarize them at a high level.
Note that current malware usually uses a mixture of several or all concepts that are described in this section.
For example, a rootkit might maintain malicious files on disk that survive reboots, while using hooking techniques to hide these files and processes so that they cannot be easily detected and removed, and applying code mutation techniques to prevent anti-malware systems from detecting running code.

{\rm Rootkit technology} refers to software designed for two purposes: maintaining stealth presence and allowing continued access to a computer system. The stealth functionality includes hiding files, hiding directories, hiding processes, masking resource utilization, masking network connections, and hiding registry keys. Not all rootkit technology is malicious, for example, some anti-malware suites use their own rootkit technologies to evade detection by malware. \Name{Samhain} \cite{chuvakin2003ups,petroni2004copilot}, for example, was one of the first pieces of anti-malware (specifically anti-rootkit) software to hide its own presence from the system, such that a malware or hacker would not be able to detect and, thus, kill off the \Name{Samhain} process. Whether rootkit implementations are designed for malicious or benign applications, many of the underlying technologies are the same. In short, rootkits can be thought of as performing man-in-the-middle attacks between different components of the operating system. In doing so, different rootkit technologies employ radically different techniques. In this section, we review four different types of rootkits.

Unlike rootkit technologies, code mutation does not aim to change the dynamic functionality of the code. Instead it aims to alter the appearance of code with each generation, generally at the binary level, so that copies of the code cannot be recognized by simple pattern-matching algorithms.


Due in part to the difficulties of static code analysis, and in part to protect system resources, the behavior of suspicious executables is often analyzed by running these executables in virtual sandboxed environments. {\em Anti-emulation} technologies aim to detect these sandboxes; if a sandbox is detected, they alter the execution flow of malicious code in order to stay hidden.

Finally, {\em targeting mechanisms} seek to manage the spread of malware and therefore minimize risk of detection and collateral damage, allowing it to remain in the wild for a longer period of time.

\subsection{Type 1 Rootkits: Malicious System Files on Disk}
\sad{Mimic system process files.}{Easy to install, survives reboots.}{Easy to detect and remove.}
The first-generation of rootkits masqueraded as disk-resident system programs (e.g., \Windows{ls}, \Windows{top}) on early \Name{Unix} machines, pre-dating the development of \Name{Windows}.
These early implementations were predominantly designed to obtain elevated privileges, hence the name ``rootkit''.
Modern rootkit technologies are designed to maintain stealth, perform activity logging (e.g., key logging), and set up backdoors and covert channels for command and control (C2) server communication \cite{kim2012brief}.

Although modern rootkits (types 2, 3, and 4) rely on privilege escalation for their functionalities, their main objective is stealth (although privilege escalation is often assumed) \cite{butler2005windows1}. 
Since first-generation rootkits reside on disk, they are easily detectable via a comparison of their hashes or checksums to hashes or checksums of system files. Due to early file integrity checkers such as \Name{Tripwire} \cite{kim1994tripwire}, first-generation rootkits have greatly decreased in prevalence and modern rootkits have trended toward memory residency over disk residency\cite{hoglund2005rootkits,szor2005theart}. This should not be conflated with saying that modern malwares are trending away from disk presence -- e.g. \Name{Gapz}\cite{Gaptz} and \Name{Olmasco} \cite{Olmasco} are advanced bootkits with persistent disk data. As we shall see below, many modern rootkits are specifically designed to intercept calls that enumerate files associated with a specific malware and hide these files from the file listing.

\subsection{Type 2 Rootkits: Hooking and in-Memory Redirection of Code Execution}
\label{sec:hooking}
\sad{Code injection by modifying pointers to libraries/functions or by explicit insertion of code.}{Difficult to differentiate genuine and malicious hooking.}{Difficult to inject.}
Second-generation rootkits hijack process memory and divert the flow of code execution so that malicious code gets executed.
This rootkit technique is generally referred to as \emph{hooking}, and can be done in several ways \cite{butler2005windows1}; e.g., via modification of function pointers to point to malicious code or via inline function patching -- an approach involving overwriting of code; not just pointers\cite{butler2004vice}.
For readability, however, we use the term hooking to refer to any in-memory redirection of code execution.

Rootkits use hooking to alter memory so that malicious code gets executed, usually prior to or after the execution of a legitimate operating system call\cite{butler2004vice,szor2005theart}.  This allows the rootkit to filter return values or functionality requested from the operating system. There are three types of hooking\cite{hoglund2005rootkits}: user-mode hooking, kernel-mode hooking, and hybrid hooking.

Hooking in general is not an inherently malicious technique. Legitimate uses for hooking exist, including hot patching, monitoring, profiling, and debugging. Hooking is generally straightforward to detect, but distinguishing legitimate hooking instances from (malicious) rootkit hooking is a challenging task\cite{hoglund2005rootkits,szor2005theart}.

\subsubsection{User-Mode Hooking}
\label{sec:userhooking}
\sad{Injection of code into User DLLs.}{Difficult to classify as malicious.}{Easy to detect.}
\begin{figure*}[t]
  \centering
  \subfloat[\label{fig:Hooking:Normal}Normal Operation]{\begin{minipage}{.39\textwidth}\centering\includegraphics[scale=0.23]{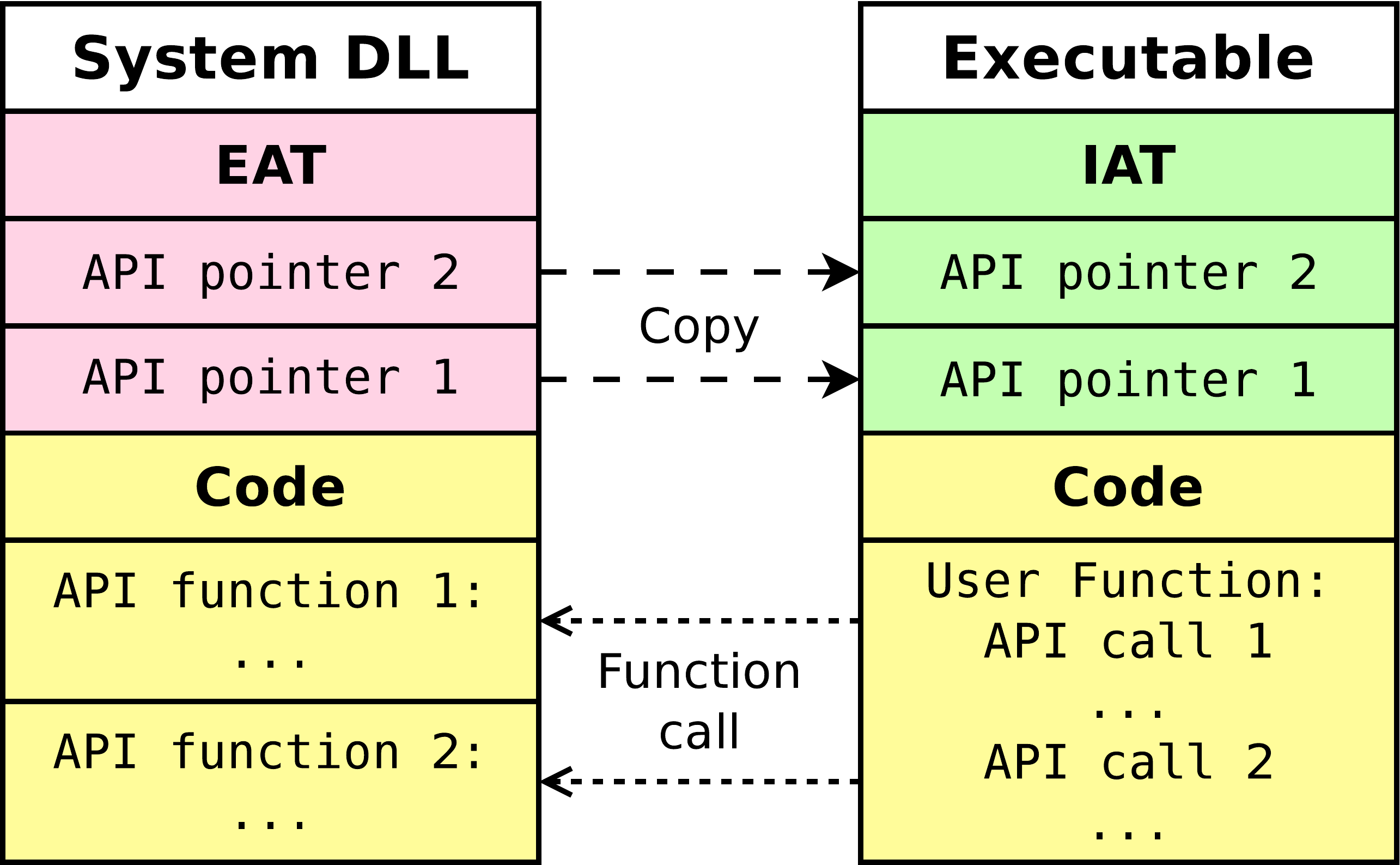}\end{minipage}}
  \subfloat[\label{fig:Hooking:Infected}Infected Operation]{\begin{minipage}{.59\textwidth}\centering\includegraphics[scale=0.23]{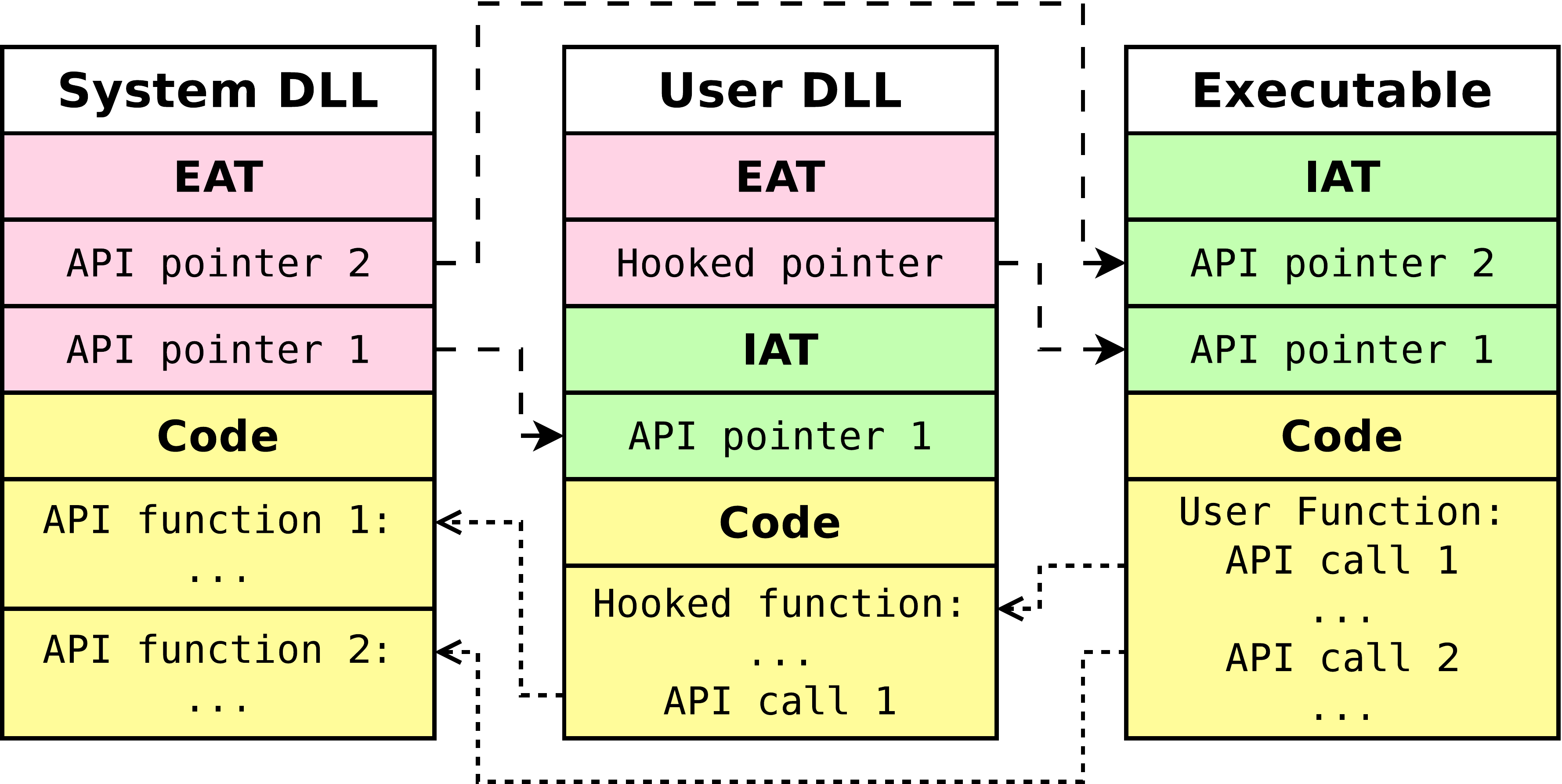}\end{minipage}}
  \cap{fig:Hooking}{Hooking}{This figure shows an example of code redirection on shared library imports. \protect\subref*{fig:Hooking:Normal} displays the normal operation of API calls, where API pointers of the DLL's \Windows{EAT} are copied into the executable's \Windows{IAT}. \protect\subref*{fig:Hooking:Infected} shows how \Windows{IAT} hooking injects (malicious) code from a user DLL before executing the original \texttt{API call 1}.}
\end{figure*}

To improve resource utilization and to provide an organized interface to request kernel resources from user space, much of the \Name{Win32} API is implemented as dynamically linked libraries (DLLs) whose callable functions are accessible via tables of function pointers.
DLLs allow multiple programs to share the same code in memory without requiring the code to be resident in each program's binary \cite{microsoft2007what}.
In and of themselves, DLLs are nothing more than special types of portable executable (PE) files.
Each DLL contains an \Windows{Export Address Table} (\Windows{EAT}) of pointers to functions that can be called outside of the DLL.
Other programs that wish to call these exported functions generally have an \Windows{Import Address Table} (\Windows{IAT}) containing pointers to the DLL functions in their PE images in memory.
This lays the ground for the popular user-mode rootkit exploit known as \emph{\Windows{IAT} hooking}\cite{kim2012brief}, in which the rootkit changes the function pointers within the \Windows{IAT} to point to malicious code.
Fig.~\ref{fig:Hooking} illustrates both malicious and legitimate usage of \Windows{IAT} hooks.
In the context of rootkit \Windows{IAT} hooking, the functions hooked are almost always operating system API functions and the malicious code pointed to by the overwritten \Windows{IAT} entry, in addition to its malicious behavior, almost always makes a call to the original API function in order to spoof its functionality\cite{hoglund2005rootkits,kim2012brief}.
Prior to or after the original API call, however, the malicious code causes the result of the library call to be changed or filtered.
By interposing the \Windows{FindFirstFile}  and \Windows{FindNextFile} \Name{Win32} API calls, for example, a rootkit can selectively filter files of specific unicode identifiers so that they will not be seen by the caller.
This particular exploit might involve calling \Windows{FindNextFile} multiple times within the malicious code to skip over malicious files and protect its stealth.

\Windows{IAT} hooking is nontrivial and has its limitations\cite{hoglund2005rootkits,leitch2011iat,microsoft2007what}, for example, it requires the PE header of the target binary to be parsed and the correct addresses of target functions to be identified. Practically, \Windows{IAT} hooking is restricted to OS API calls, unless specifically engineered for a particular non-API DLL\cite{hoglund2005rootkits,leitch2011iat}. An additional difficulty of \Windows{IAT} hooking is that DLLs can be loaded at different times with respect to the executable launch \cite{hoglund2005rootkits}. DLLs can be loaded either at load time or at runtime of the executable. In the latter case, the \Windows{IAT} does not contain function pointers until just before they are used, so hooking the \Windows{IAT} is considerably more difficult. Further, by loading DLLs with the \Name{Win32} API calls \Windows{LoadLibrary} and \Windows{GetProcAddress}, no entries will be created in the \Windows{IAT}, making the loaded DLLs impervious to \Windows{IAT} hooking\cite{hoglund2005rootkits,microsoft2007what}.

Inline function patching, a.k.a. \emph{detouring}, is another common second-generation technique, which avoids some of the shortcomings of \Windows{IAT} hooking \cite{hoglund2005rootkits}.
Unlike function pointer modification, detouring uses the direct modification of code in memory.
It involves overwriting a snippet of code with an unconditional jump to malicious code, saving the stub of code that was overwritten by the malicious code, executing the stub after the malicious code, and possibly jumping back to the point of departure from the original code so that the original code gets executed -- a technique known as \emph{trampolining} \cite{hoglund2005rootkits}.

In practice, overwriting generic code segments is difficult for several reasons. First, stub-saving without corrupting memory is inherently difficult\cite{szor2005theart}. Second, the most common instruction sets, including \Name{x86} and \Name{x86-64}, are variable-length instruction sets, meaning that disassembly is necessary to avoid splitting instructions in a generic solution \cite{x64_introduction}.
Not only is disassembly a high overhead task for stealth software, but even with an effective disassembly strategy, performing arbitrary jumps to malicious code can result in unexpected behavior that can compromise the stealth of the rootkit\cite{szor2005theart}.
Consider, for example, mistakenly placing a jump to shell code and back into a loop that executes for many iterations. One execution of the shell code might have negligible overhead, but a detour placed within an otherwise tight loop may have a noticeable effect on system behavior.

Almost all existing \Name{Windows} rootkits 
that rely on inline function patching consequently hook in the first few bytes of the target function \cite{hoglund2005rootkits}. In addition to the fact that an immediate detour limits the potential for causing strange behaviors, many \Name{Windows} compilers for \Name{x86} leave 5 \Windows{NOP} bytes at the beginning of each function. These bytes can easily be replaced by a single byte jump opcode and a 32 bit address. This is not an oversight. Rather, like hooking in general, detours are not inherently malicious and have a legitimate application, namely \emph{hot patching}\cite{microsoft1999detours}. Hot patching is a technique, which uses detours to perform updates to binary in memory.
During hot patching, an updated copy of the function is placed elsewhere in memory and a jump instruction with the address of the updated copy as an argument is placed at the beginning of the original function. The purpose of hot patching is to increase availability without the need for program suspension or reboot \cite{hunt1999detours}. \Name{Microsoft Research} even produced a software package called \Name{Detours} specifically designed for hot patching \cite{hunt1999detours,microsoft1999detours}. In addition, detours are also used in anti-malware \cite{hoglund2005rootkits}. Like \Windows{IAT} hooks, detours are relatively easy to detect. However, the legitimate applications of detours are difficult to distinguish from rootkit uses of detours\cite{hunt1999detours,hoglund2005rootkits}.

Detours are also not limited in scope to user mode API functions. They can also be used to hook kernel functions\cite{hoglund2005rootkits}.
Regardless of the hooking strategy, when working in user mode, a rootkit must place malicious code into the address space of the target process. This is usually orchestrated through \emph{DLL injection}, i.e., by making a target process load a DLL into its address space.
Having a process load a DLL is common, so common that DLL injection can simply be performed using \Name{Win32} API functionality. This makes DLL injections easy to detect. However, discerning benign DLL injections from malicious DLL injections is a more daunting task\cite{protean2013api, protean2013hookex, protean2013remotethread}.

Three of the most common DLL injection techniques are detailed in \cite{protean2013api, protean2013hookex, protean2013remotethread}. The simplest technique exploits the \Name{AppInit\_DLLs} registry key, which proceeds as follows: a DLL, containing a \Windows{DllMain} function is written, optionally with a payload to be executed. The DLL main  function takes three arguments: a DLL handle, the reason for the call (process attach/detach or thread attach/detach), and a third value which depends on whether the DLL is statically or dynamically loaded. By changing the value of the \Name{AppInit\_DLLs} registry key to include the path to the DLL to be executed, and changing the \Windows{LoadAppInit\_DLLs} registry key's value to 1, whenever any process loads \Windows{user32.dll}, the injected DLL will also be loaded and the \Windows{DllMain} functionality will be executed. Although the DLL gets injected only when a program loads \Windows{user32.dll}, \Windows{user32.dll} is prevalent in many applications, since it is responsible for key user interface functionality. Whether or not \Windows{DllMain} calls malicious functionality, the \Name{AppInit} technique can be used to inject a DLL into an arbitrary process' address space, as long as that process calls functionality from \Windows{user32.dll}. Note that, although the injection itself involves setting a registry key, which could indicate the presence of a rootkit, the rootkit can change the value of the registry key once resident in the target process' address space \cite{hoglund2005rootkits}.

\begin{figure*}[ht]
  \centering
  \includegraphics[width=\linewidth]{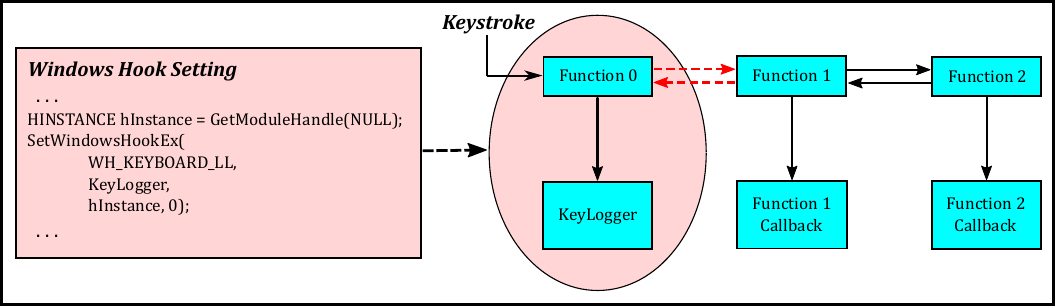}
  \cap{fig:keylogger}{Exploiting event hook chains}{A prototypical keylogger application gets the target process' context, injects a malicious DLL into its address space, and prepends a function pointer to code within this DLL to the keypress event hook chain. Whenever a key is pressed, the newly introduced callback is triggered, thereby allowing the malicious code to log every keystroke.}
\end{figure*}

A second method of DLL injection exploits \Name{Windows} event hook chains \cite{protean2013hookex,msdnSet,msdnHooks}. Event hook chains are linked lists containing function pointers to application-defined callbacks. Different hook chains exist for different types of events, including key presses, mouse motions, mouse clicks, and messages\cite{msdnHooks}. Additional procedures can be inserted into hook chains using the \Windows{SetWindowsHookEx} \Name{Win32} API call.
By default, inserted hook procedures are placed at the front of a hook chain, e.g.,  prominent rootkits/bootkits ~\cite{Olmasco,TDL4-reboot} overwrite the pointer to the handlers in the \Windows{DRIVER\_OBJECT} structure, but some rootkits, e.g.,  \Name{Win32/Gapz}\cite{Gaptz},  use splicing, patching the handlers' code themselves.
Upon an event associated with a particular hook chain, the operating system sequentially calls functions within the hook chain. Each hook function determines whether it is designed to handle the event. If not, it calls the \Windows{CallNextHookEx} API function. This invokes the next procedure within the hook chain. There are two specific types of hook chains: global and thread specific. Global hooks monitor events for all threads within the calling thread's desktop, whereas thread-specific hooks monitor events for individual threads. Global hook procedures must reside in a DLL disjoint from the calling thread's code, whereas local hook procedures may reside within the code of the calling thread or within a DLL \cite{msdnHooks}.

For hook chain DLL injection a support program is required, as well as a DLL exporting the functionality to be hooked. The attack proceeds as follows\cite{protean2013hookex}: the support program gets a handle to the DLL and obtains the address of one of the exported functions through \Name{Win32} API calls. The support program then calls the \Windows{SetWindowsHookEx} API function passing it the action to be hooked and the address of the exported function from the DLL. \Windows{SetWindowsHookEx} places the hook routine into the hook chain of the victim process, so that the DLL functionality is invoked whenever a specified event is triggered. When the event first occurs, the OS injects the specified DLL into the process' address space, which automatically causes the \Windows{DllMain} function to be called. Subsequent events do not require the DLL to be reloaded since the DLL's exports get called from within the victim process' address space. An example keylogger rootkit is shown in Fig.~\ref{fig:keylogger}, which will log the pressed key and call \Windows{CallNextHookEx} to trigger the default handling of the keystroke. Again, benign addition of hook chain functions via \Windows{SetWindowsHookEx} is common, e.g., to decide which window/process should get the keystroke; the difficult task is determining if any of the added functions are malicious.

A third common DLL injection strategy involves creating a remote thread inside the virtual address space of a target process using the \Windows{CreateRemoteThread} Win32 API call \cite{protean2013remotethread}.
The injection proceeds as follows\cite{protean2013remotethread}: a support program controlled by the malware calls \Windows{OpenProcess}, which returns a handle to the target process.
The support program then calls \Windows{GetProcAddress} for the API function \Windows{LoadLibrary}. \Windows{LoadLibrary} will be accessible from the target process because this API function is part of \Windows{kernel32.dll}, a user space DLL from which every \Name{Win32} user space process imports functionality.
To insert the exported function name into the target process' address space, the malicious process must call the \Windows{VirtualAllocEx} API function.
This API function allocates a virtual memory range within the target process' address space.
The allocation is required in order to store the name of the rootkit DLL function. The \Windows{WriteProcessMemory} API call is then used to place the name of the malicious DLL into the target process' address space.
Finally, the \Windows{CreateRemoteThread} API function calls the \Windows{LoadLibrary} function to inject the rootkit DLL.
Like event chains, the \Windows{CreateRemoteThread} API call has legitimate uses.
For example, a debugger might fire off a remote thread in a target process' address space for profiling and state inspection.
An anti-malware module might perform similar behavior.
Finally, IO might be handled through pointers to callbacks exchanged by several processes, where the callback method is intended to execute in another process' address space.
The fact that the API call has so many potentially legitimate uses makes malicious exploits particularly difficult to detect.

\subsubsection{Kernel-Mode Hooking}
\sad{Injection of code into the Kernel via device drivers.}{Difficult to detect by user-mode IDSs.}{Intricate to implement correctly.}

Rootkits implementing kernel  hooks are more difficult to detect than those implementing user space hooks. In addition to the extended functionality afforded to the rootkit, user space anti-malwares cannot detect kernel hooks because they do not have the requisite permissions to access kernel memory\cite{hoglund2005rootkits,szor2005theart}. Kernel memory resides in the top part of a process' address space. For 32-bit \Name{Windows}, this usually corresponds to addresses above \Name{0x80000000}, but can correspond to addresses above \Name{0xC0000000}, if processes are configured for 3GB rather than 2GB of user space memory allocation. All kernel memory across all processes maps to the same physical location, and without special permissions, processes cannot directly access this memory.

Kernel hooks are most commonly implemented as device drivers\cite{hoglund2005rootkits}. Popular places to hook into the kernel include the \Windows{System Service Descriptor Table} (\Windows{SSDT}), the \Windows{Interrupt Descriptor Table} (\Windows{IDT}), and \Windows{I/O Request Packet} (\Windows{IRP}) function tables of device drivers\cite{hoglund2005rootkits}. The \Windows{SSDT} was the  hooking mechanism used in  the classic \Name{Sony DRM Rootkit} \cite{sonydrm} and the more recent \Name{Necurs} malware \cite{necurs}, and is often a used as part of more complex multi-exploitation kits such as the \Name{RTF} zero-day (\Name{CVE-2014-1761}) attack\cite{rtf-necurs} which was detected in the wild.

The \Windows{SSDT} is a \Name{Windows} kernel memory data structure of function pointers to system calls. Upon a system call, the operating system indexes into the table by the function call ID number, left-shifted 2 bits to find the appropriate function in memory. The \Windows{System Service Parameter Table} (\Windows{SSPT}) stores the number of bytes that arguments require for each system call. Since \Windows{SSPT} entries are one byte each, system calls can take up to 255 arguments.  The \Windows{KeServiceDescriptorTable} contains pointers to both the \Windows{SSDT} and the \Windows{SSPT}.

When a user space program performs a system call, it invokes functionality within \Windows{ntdll.dll}, which is the main interface library between user space and kernel space. The \Name{EAX} register is filled with the system function call ID and the \Name{EDX} register is filled with the memory address of the arguments. After performing a range check, the value of \Name{EAX} is used by the OS to index into the \Windows{SSDT}. The program counter register \Name{IP} is then filled with the appropriate address from the \Windows{SSDT}, executing the dispatcher call. Dispatches are triggered by the \Windows{SYSENTER} processor instruction or the more dated \Windows{INT 2E} interrupt.

\Windows{SSDT} hooks are particularly dangerous because they can supplement any system call with their own functionality.
Hoglund and Butler \cite{hoglund2005rootkits} provide an example of process hiding via \Windows{SSDT} hook, in which the \Windows{NTQuerySystemInformation} \Name{NTOS} system call is hooked to point to shell code, which filters \Windows{ZwQuerySystemInformation} structures  corresponding to processes by their unicode string identifiers. Selective processes can be hidden by changing pointers in this data structure. \Name{Windows} provides some protection to prevent \Windows{SSDT} hooks by making \Windows{SSDT} memory read-only. Although this protection makes the attacker's job more difficult, there are ways around it. One method is to change the memory descriptor list (MDL)\cite{ssdt_hooking,hoglund2005rootkits} for the requisite area in memory. This involves casting the \Windows{KeServiceDescriptorTable} to the appropriate data structure, and using it to build an MDL from non-paged memory. By locking the memory pages and changing the flags on the MDL one can change the permissions on memory pages. Another method of disabling memory protections is by zeroing the write protection bit in control register \Name{CR0}.

The \Windows{Interrupt Descriptor Table} (\Windows{IDT}) is another popular hook target\cite{idt_hooking}. The interrupt table contains pointers to callbacks that occur upon an interrupt.
Interrupts can be triggered by both hardware and software.
Because interrupts have no return values, \Windows{IDT} hooks are limited in functionality to denying interrupt requests.
They cannot perform data filtration. Multiprocessing systems have made \Windows{IDT} hooking more difficult\cite{hoglund2005rootkits}. Since each CPU has its own \Windows{IDT}, an attacker must usually hook \Windows{IDT}s of all CPUs to be successful.
Hooking only one of multiple \Windows{IDT}s causes an attack to have only limited impact.

A final popular kernel hook target discussed by Hoglund and Butler \cite{hoglund2005rootkits} is the \Windows{IRP} dispatch table of a device driver. Since many devices access kernel memory directly, \Name{Windows} abstracts devices as device objects. Device objects may represent either physical hardware devices such as buses or they may represent software ``devices'' such as network protocol stacks or updates to the \Name{Windows} kernel. Device objects may even correspond to anti-virus components that monitor the kernel. Each device object has at least one device driver. Communication between kernel and device driver is performed via the I/O manager. The I/O manager calls the driver, passing pointers to the device object and the I/O request. The I/O request is passed in a standardized data structure called an \Windows{I/O Request Packet} (\Windows{IRP}).
Within the device object is a pointer to the driver object. Drivers themselves are nothing more than special types of DLLs. The I/O manager passes the device object and the \Windows{IRP} to the driver. How the driver behaves depends on the contents and flags of the \Windows{IRP}. The function pointers for various \Windows{IRP} options are stored in a dispatch table within the driver. A rootkit can subvert the kernel by changing these function pointers to point to shell code \cite{hoglund2005rootkits}. An anti-virus implemented as a filter driver, for example, may be subverted by rewriting its dispatch table. Hoglund and Butler \cite{hoglund2005rootkits} provide an in-depth example of using driver function table hooking to hide TCP connections. Essentially any kernel service that uses a device driver can be subverted by hooking the \Windows{IRP} dispatch table in a similar manner.

Note that while hooking device driver dispatch tables sounds relatively simple, the technique requires sophistication to be implemented correctly \cite{hoglund2005rootkits}. First, implementing bug-free kernel driver code is an involved task to begin with. Since drivers share the same memory address space as the kernel, a small implementation error can corrupt kernel memory and result in a kernel crash. This is one of the reasons that \Name{Microsoft} has gravitated to user-mode drivers when possible \cite{tanenbaum2007modern}. Second, in many applications drivers are stacked. When dealing with physical devices, the lowest level driver on the stack serves the purpose of abstracting bus-specific behavior to an intermediate interface for the upper level driver. Even in software, drivers may be stacked, for example, anti-virus I/O filtering or file system encryption/decryption can be performed by a filter driver residing in the mid-level of the stack\cite{hoglund2005rootkits}. A successful rootkit author must therefore understand how the device stack behaves, where in the device stack to hook, and how to perform I/O completion such that the hook does not result in noticeably different behavior.  

\subsubsection{Hybrid Hooking}
\label{sec:hybrid}
\sad{Hook user functions into kernel DLLs.}{Even more difficult to detect than kernel hooking.}{More difficult to implement than kernel hooking.}
Hybrid hooks aim to circumvent anti-malwares by attacking user space and kernel space simultaneously. They involve implementing a user space hook to kernel memory. Hoglund and Butler \cite{hoglund2005rootkits} discuss a technique to hook the user space \Windows{IAT} of a process from the kernel. The motivation behind this technique is based on the observation that user space \Windows{IAT} hooks are detectable because one needs to allocate memory within the process' context or inject a DLL for the same effect. But is there some means to hook the \Windows{IAT} through the kernel, without the need to allocate user space memory for \Windows{IAT} hooks? The answer is yes: the attack in \cite{hoglund2005rootkits} leverages two aspects of the Windows architecture. First, it uses the \Windows{PsSetLoadImageNotifyRoutine}, a kernel mode support routine that registers driver callback functions to be called whenever a PE image gets loaded into memory \cite{msdnPs}. The callback is called within the target process' context after loading but before execution of the PE. By parsing the PE image in memory, an attacker can change the \Windows{IAT}. The question then becomes, how to run malicious code without overt memory allocation or DLL injection into the process' address space? One solution uses a specific page of memory \cite{jack2005step}: in \Name{Windows} there exists a physical memory address shared by both kernel space and user space, which the kernel can write to, but the user cannot. The user and kernel mode addresses are \Name{0x7FFE0000} and \Name{0xFFDF0000}, respectively. The reason for this virtual $\leftrightarrow$ physical mapping convention stems from the introduction of \Windows{SYSENTER} and \Windows{SYSEXIT} processor instructions, for fast switches between user mode and kernel mode. Approximately 1kB of this page is used by the kernel \cite{hoglund2005rootkits}, but the remaining 3kB are blank. Writing malicious code to addresses in the page starting at \Name{0xFFDF0000} in kernel space and placing a function pointer to the beginning of the code at the corresponding address in user space allows the rootkit to hook the \Windows{IAT} without allocating memory or performing DLL injection.

Another hybrid attack is discussed in \cite{srivastava2011operating}. The attack is called \Name{Illusion}, and involves both kernel space and user space components. The motivation behind the attack is to circumvent intrusion detection systems that rely on system call analysis (cf. Sec.~\ref{sec:countermeasure_survey}).
To understand the \Name{Illusion} attack, we review steps involved in performing a system call:
first, a user space application issues an \Windows{INT 3} interrupt or a \Windows{SYSENTER} processor instruction, which causes the processor to switch to kernel mode and execute the dispatcher. The dispatcher indexes into the \Windows{SSDT} to find the handler for the system call in question. The handler performs its functionality and returns to the dispatcher. The dispatcher then passes return values to the user space application and returns the processor to user space.
These steps should be familiar from the prior discussion of hooking the \Windows{SSDT}.
\Name{Illusion} works by creating a one-to-all mapping of potential execution paths between system calls, which take array buffer arguments and the function pointers of the \Windows{SSDT}. Although the same effect could be obtained by making changes directly to the \Windows{SSDT}, the \Name{Illusion} approach, unlike \Windows{SSDT} hooking, cannot be detected using the techniques discussed in Sec.~\ref{sec:countermeasure_survey}. \Name{Illusion} exploits system calls such as \Windows{DeviceIoControl}, which is used to exchange data buffers between application and kernel driver. Parts of the rootkit reside in both in kernel space and in user space. Messages are passed between user space rootkit and kernel space rootkit by filling the buffer. Communication is managed via a dedicated protocol.
This allows the user space rootkit to make system calls on its behalf without changing the \Windows{SSDT}. Further, metamorphic code as described in Sec.~\ref{sec:code_mutation} can be leveraged to change the communication protocol at each execution.

\subsection{Type 3 Rootkits: Direct Kernel Object Manipulation}
\sad{Modify dynamic kernel data structures.}{Extremely difficult to detect.}{Has limited applications.}

\begin{figure}[t]
  \centering
  \includegraphics[width=.48\textwidth]{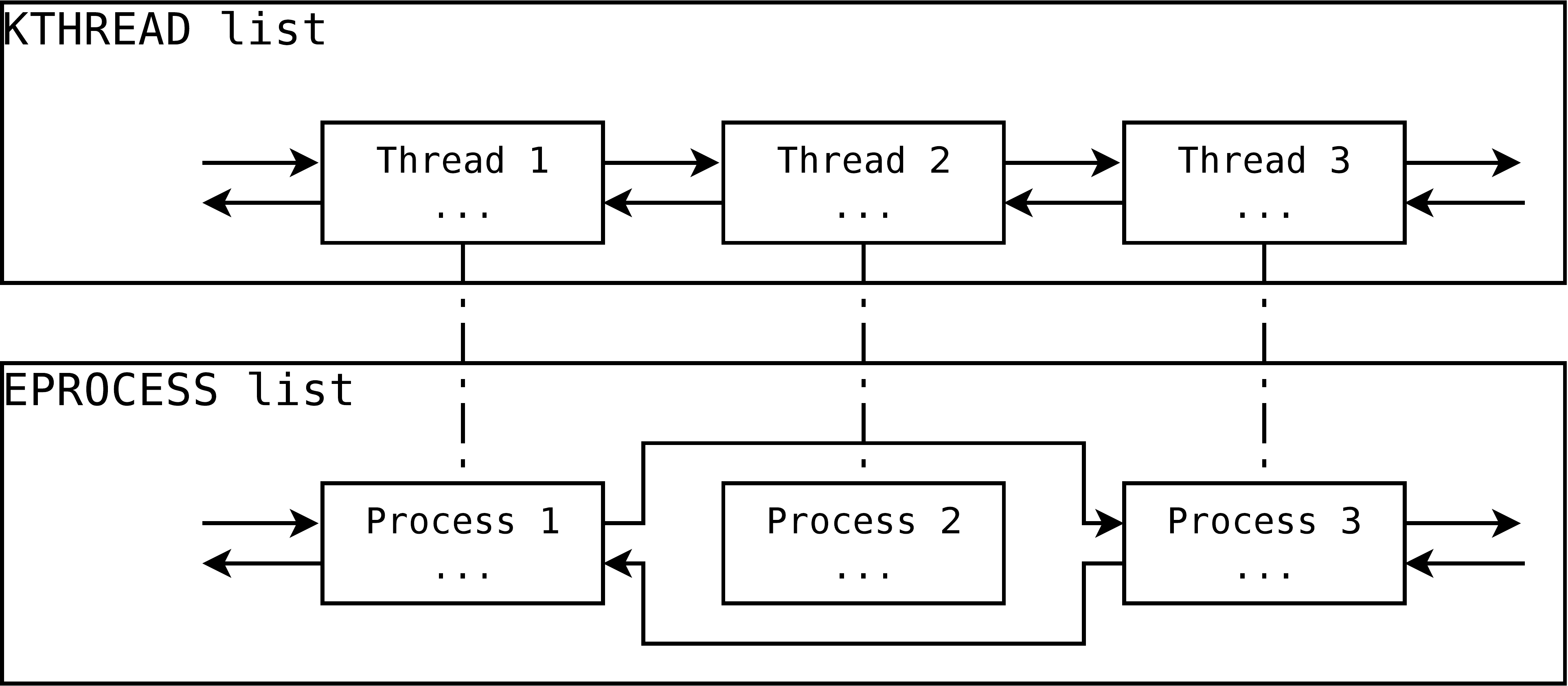}
  \cap{fig:DKOM}{DKOM attack}{This figure displays a successful DKOM attack, where the malicious code of \texttt{Process 2} is hidden from the system yet continues to execute.}
\end{figure}

Although second-generation rootkits remain ubiquitous, they are not without their limitations.
Their change of overt function behavior inherently leaves a detectable footprint because it introduces malicious code -- either in user space or in kernel space -- which can be detected and analyzed\cite{butler2005windows1}.
Third-generation \emph{direct kernel object manipulation} (DKOM) attacks take a different approach.
DKOM aims to subvert the integrity of the kernel by targeting dynamic kernel data structures responsible for bookkeeping operations \cite{butler2005windows1}.
Like kernel space hooks, DKOM attacks are immune to user space anti-malware, which assumes a trusted kernel.
DKOM attacks are also much harder to detect than kernel hooks because they target dynamic data structures whose values change during normal runtime operation. By contrast, hooking static areas of the kernel like the \Windows{SSDT} can be detected with relative ease because these areas should remain constant during normal operation\cite{butler2005windows1}.

The canonical example of DKOM is process hiding. The attack can be carried out on most operating systems, and relies on the fact that schedulers use different data structures to track processes than the data structures used for resource bookkeeping operations\cite{beck2005detecting}. In the \Name{Windows NTOS} kernel, for example, the kernel layer\footnote{The kernel itself has three layers, one of which is called the \textit{kernel layer}. The other two layers of the kernel are the \textit{executive layer} and the \textit{hardware abstraction layer}.} is responsible for managing thread scheduling, whereas the executive layer, which contains the memory manager, the object manager, and the I/O manager is responsible for resource management \cite{tanenbaum2007modern}. Since the executive layer allocates resources (e.g., memory) on a per-process basis, it views processes as \Windows{EPROCESS} (executive process) data structures, maintained in double circularly linked lists. The scheduler, however, operates on a per-thread instance, and consequently maintains threads in its own double circularly linked list of \Windows{KTHREAD} (kernel thread) 
data structures. By modifying pointers, a rootkit with control over kernel memory can decouple an \Windows{EPROCESS} node from the linked list, re-coupling the next and previous \Windows{EPROCESS} structures' pointers. Consequently, the process will no longer be visible to the executive layer and calls by the \Name{Win32} API will, therefore, not display the process. However, the thread scheduler will continue CPU quantum allocation to the threads corresponding to the hidden \Windows{EPROCESS} node. The process will, thus, be effectively invisible to both user and kernel mode programs -- yet it will still continue to run. This attack is depicted in Fig.~\ref{fig:DKOM}.

While process hiding is the canonical DKOM example, it is just one of several DKOM attack possibilities. Baliga et al.~\cite{baliga2011data} discuss several known DKOM attack variants, including zeroing entropy pools for pseudorandom number generator seeds, disabling pseudorandom number generators, resource waste and intrinsic DOS, adding new binary formats, disabling firewalls, and spoofing in-memory signature scans by providing a false view of memory.

Proper DKOM implementations are extremely difficult to detect. Fortunately, DKOM is not without its shortcomings and difficulties from the rootkit developer's perspective. Changing OS kernel data structures is no easy task and incorrect implementations can easily result in kernel crashes, thereby causing an overt indication of a malware's presence. Also, DKOM introduces no new code to the kernel apart from the code to modify kernel data structures to begin with. Therefore, inherent limitations on the scope of a DKOM attack are imposed by the manner in which the kernel uses its data structures. For example, one usually cannot hide disk resident files via DKOM because most modern operating systems
do not have kernel level data structures corresponding to lists of files.

\subsection{Type 4 Rootkits: Cross Platform Rootkits and Rootkits in Hardware}
\sad{Attack systems using very low-level rootkits.}{Undetectable by conventional software countermeasures.}{Requires custom low-level hypervisor, BIOS, hardware or physical/supply chain compromise to be effective.}
Fourth-generation rootkit technologies operate at the virtualization layer, in the BIOS, and in hardware \cite{seifried2008fourth}. To our knowledge, fourth-generation rootkits have been developed only in proof-of-concept settings, as we could not find any documentation of fourth-generation rootkits in the wild.
Because they reside at a lower level than the operating system, they cannot be detected through the operating system and are, therefore, OS independent.
However, they still dependent on the type of BIOS version, instruction set, and hardware\cite{seifried2008fourth}.
Since fourth-generation rootkits are theoretical in nature -- at least as of now -- we consider them outside the scope of this survey. We mention them in this section for completeness and because they may become relevant after the publication of this survey.

\subsection{Code Mutation}
\label{sec:code_mutation}
\sad{Self-modifying malicious code.}{Avoids simple signature matching.}{Greater runtime overhead and detectable via emulation.}

\begin{figure*}[!htbp]
  \centering
  \subfloat[\label{fig:metamorphic1}]{\includegraphics[width=0.48\textwidth]{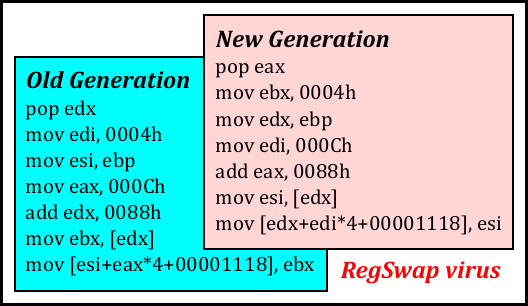}}\hspace*{.02\textwidth}
  \subfloat[\label{fig:metamorphic2}]{\includegraphics[width=0.48\textwidth]{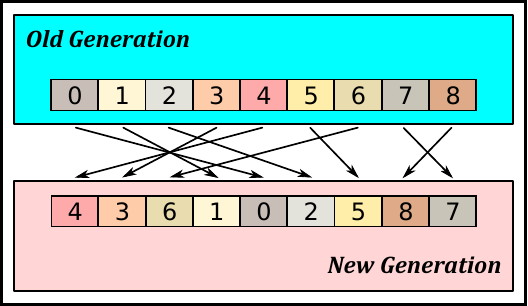}}\\
  \subfloat[\label{fig:metamorphic3}]{\includegraphics[width=0.32\textwidth]{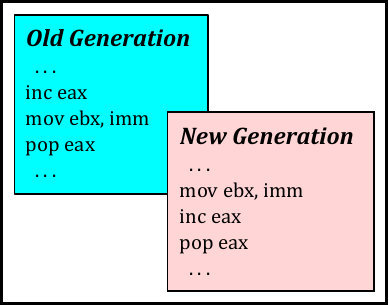}}\hspace*{.01\textwidth}
  \subfloat[\label{fig:metamorphic4}]{\includegraphics[width=0.32\textwidth]{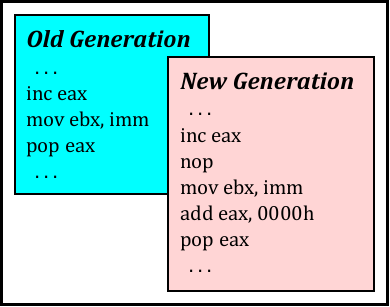}}\hspace*{.01\textwidth}
  \subfloat[\label{fig:metamorphic5}]{\includegraphics[width=0.32\textwidth]{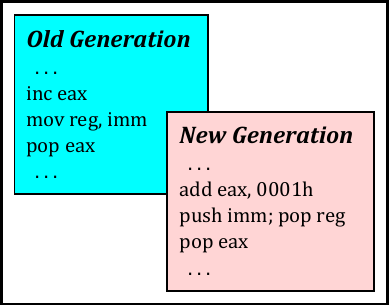}}
  \cap{fig:metamorphic}{Metamorphic code obfuscation}{Five techniques employed by metamorphic engines to evade signature scans across malware generations. \protect\subref*{fig:metamorphic1} Register swap: exchanging registers as demonstrated by code fragments from the \Name{RegSwap} virus \cite{szor2001hunting}. \protect\subref*{fig:metamorphic2} Subroutine permutation: reordering subroutines of the virus code. \protect\subref*{fig:metamorphic3} Transposition: modifying the execution order of independent instructions. \protect\subref*{fig:metamorphic4} Semantic NOP-insertion: injecting NOPs or instructions that are semantically identical to NOPs. \protect\subref*{fig:metamorphic5} Code mutation: replacing instructions with semantically equivalent code.
}
\end{figure*}

In early viruses, the viral code was often appended to the end of an executable file, with the entry point changed to jump to the viral code before running the original executable\cite{szor2005theart}. Once executed, the virus code in turn would jump to the beginning of the body of the executable so that the executable was run post-replication. The user would be none the wiser until the host system had been thoroughly infected. Anti-malware companies soon got wise and started checking hashes of code blocks -- generally at the end of files. To counter, malware authors began to encrypt the text of the viruses. This required a decryption routine to be called at the beginning of execution. The virus was then re-encrypted with a different key upon each replication\cite{szor2005theart}.

These encrypted viruses had a fatal flaw: the decryption routine was jumped to somewhere in the executable. Anti-malware solutions merely had to look for the decrypter. Thus, \emph{polymorphic} engines were created, in which the decryption engine mutated itself at each generation, no longer matching a fixed signature.
However, polymorphic viruses were still susceptible to detection\cite{szor2005theart}: although the detector mutated, the size of the malicious code did not change, was still placed at the end of the file, and was susceptible to entropy analysis, depending on the encryption technique.

To this end, entry-point obscuring (EPO) viruses were created, where the body of the viral code is placed arbitrarily in the executable, sometimes in a distributed fashion\cite{szor2005theart}. To counter the threats from polymorphic viruses, Kaspersky (of \Name{Kaspersky Lab} fame) and others \cite{beaucamps2007advanced} created emulation engines, which run potentially malicious code in a virtual machine. In order to run, the body of the viral code must decrypt itself in memory in some form or another, and when it does, the body of the malicious code is laid bare for hashed signature comparison as well as behavioral/heuristic analysis.

To combat emulation, \emph{metamorphic} engines were developed.
Just as polymorphic malwares mutate their decryption engines at each generation, metamorphic engines mutate the full body of their code and, unlike polymorphics, change the size of the code body from one generation to another\cite{szor2005theart}. Some malwares still encrypt metamorphic code, or parts of metamorphic code, while others do not -- as encryption and run time packing techniques can reveal the existence of malicious code\cite{szor2005theart}.

Metamorphic code mutation techniques, as shown in Fig.~\ref{fig:metamorphic}, include register swaps, subroutine permutations, transpositions of independent instructions, insertion of \Name{NOP}s or instruction sequences that behave as \Name{NOP}s, and parser-like mutations by context-free grammars (or other grammar types) \cite{sridhara2013metamorphic,filiol2007metamorphism,zbitskiy2009code,beaucamps2007advanced}. Many metamorphic techniques are similar to compilation techniques, but for a much different purpose. The metamorphic engine in the \Name{MetaPHOR} worm, for example, disassembles executable code into its own intermediate representation and uses its own formal grammar to carry out this mutation \cite{beaucamps2007advanced}.

Code transformation techniques are not particular to native code either: Faruki et al.~\cite{faruki2014evaluation} presented several \Name{Dalvik} bytecode obfuscation techniques and tested them against several \Name{Android} security suites, which often failed at recognizing transformed malicious code. While some of the transformation targets are unique to obfuscation on \Name{Android} devices -- e.g., renaming packages and encrypting resource files -- the control, data, and layout transformations in \cite{faruki2014evaluation} follow the same principles of code obfuscation at the native level.

\subsection{Anti-Emulation}
\label{sec:antiemulation}
\sad{Malware behaves differently when running in an emulated environment.}{Malware evades detection during emulation.}{Needs to detect the presence of the emulator reliably. May not run in certain virtualized environments.}

Mutation engines, including metamorphics and polymorphics, change the instructions in the target code itself and naturally its runtime\cite{szor2001hunting}.
However, they do not change the underlying functionality.
Therefore, during emulation (cf.~Sec.~\ref{sec:virtualization}), behavioral and heuristic techniques can be used to fingerprint malicious code, for example, if the malware conducts a strange series of system calls, or if it attempts to establish a connection with a C2 server at a known malicious address.
Hence, malware can be spotted regardless of the degree of obfuscation present in the code \cite{szor2005theart}.

The success of early emulation techniques led to the usage of malicious anti-emulation tactics, which include attempts to detect the emulator by examining machine configurations -- e.g., volume identifiers and network interface -- and use of difficult to emulate functionality, e.g., invoking the GPU \cite{faruki2016droidanalyst,szor2005theart}.
In turn, emulation strategies have become more advanced, for example, in their \Name{DroidAnalyst} framework \cite{faruki2016droidanalyst} for \Name{Android}, Faruki et al. implement a realistic emulation platform by overloading default serial numbers, phone numbers, geolocations, system time, email accounts, and multimedia files to make their emulator more difficult to detect.

A realistic emulation environment is a good start to avoid emulator detection based on hardware characteristics, but it alone is insufficient to defeat all types of anti-emulation, for example, \Name{Duqu} only executes certain components after 10 minutes idle when certain requirements are met \cite{bencsath2012cousins}. Similarly the \Name{Kelihos} botnet \cite{kelihos}, and the \Name{Nap} Trojan \cite{nap} use the \Windows{SleepEx} and \Windows{NtDelayExecution} API calls to delay malicious execution until longer times than a typical emulator will devote to analysis.
\Name{PoisonIvy} \cite{poisonivy} and similarly \Name{UpClicker} \cite{upclicker} establish malicious connections only when the left mouse button is released.
\Name{PushDo} takes a more offensive approach, using \Windows{PspCreateProcessNotify} to de-register sandbox monitoring routines\cite{singh2013hot}.
Other malwares take advantage of dialog boxes and scrolling\cite{singh2013hot}.
Even mouse movements are taken into consideration and malware can differentiate between human and simulated mouse movements by assessing speed, curvature, and other features\cite{singh2013hot}.
Thus, emulated environments for stealth malware detection face the tradeoff between realistic emulation and implementation cost.
Anti-emulation in turn faces a different problem: with the explosion of virtualization technology, thanks largely to the heavy drive toward cloud computing, virtualized (emulated) environments are seeing increased general-purpose use.
This draws into question the effectiveness of anti-emulation as a stealth technique: if malicious code will not run in a virtual environment, then it might not be an effective attack if the targeted machine is virtualized.

\subsection{Targeting Mechanisms}
\sad{Malware runs on or spreads to only chosen systems.}{Decreases risk of detection.}{Malware spreads at a lower rate. Motivation for the attack is given away if detected.}

Stealth targeted attacks -- which aim to compromise specific targets -- are becoming more advanced and more widespread\cite{istr2016symantec}.
While targeting mechanisms are not necessarily designed for stealth purposes, they have the effect of prolonging the amount of time that malware can remain undetected in the wild. This is done by allowing the malware to spread/execute only on certain high-value systems, thus minimizing the likelihood of detection while maximizing the impact of the attack. For example, recent point of sale (POS) compromises \cite{pos} targeted only specific corporations. The \Name{DarkHotel} \cite{darkhotel} \textit{advanced persistent threat} (APT) targets only certain individuals (e.g., business executives). The notorious \Name{Stuxnet} worm and its relatives \Name{Duqu}, \Name{Flame}, and \Name{Gauss} employed sophisticated targeting mechanisms \cite{bencsath2012cousins}, preventing the malwares from executing on un-targeted systems. 
\Name{Stuxnet} checks system configuration prior to execution; malicious components simply will not execute if the detected environment is not correct rather attempting to execute and failing\cite{falliere2011w32}.
\Name{Gauss}'s \Name{G\"odel} module is even encrypted with an RC4 cipher, with a key derived from system-specific data; thus, much of the functionality of the malware remains unknown, since a large part of the body of the code can only be decrypted with knowledge of the targeted machines\cite{gauss_abnormal}. Hence, IDS developers and anti-malware researchers cannot get the malicious code running in un-targeted machines.
Targeting mechanisms may also change the behavior of the malware depending on the configuration of the machine so as to evade detection. For example, \Name{Flame} dynamically changes file extensions depending on the type of anti-malware that it detects on the machine\cite{bencsath2012skywiper}. Other malwares may simply not run or choose to uninstall themselves to evade detection, while others will execute only under certain conditions on time, date, and geolocation\cite{szor2005theart,bencsath2012cousins}.

\section{Component-Based Stealth Malware Countermeasures}
\label{sec:countermeasure_survey}

In this section, we discuss anti-stealth malware techniques that aim to protect the integrity of areas of systems, which are known to be vulnerable to attacks.
These techniques include hook detection, cross-view detection, invariant specification, and hardware and virtualization solutions.

When assessing the effectiveness of any malware recognition system, it is important to consider the system's respective precision/recall tradeoff.
Recall refers to the proportion of malicious samples of a given type that were correctly detected as malicious samples of that type, while precision refers to the proportion of the samples that the system marked as a malicious type that are actually of that malicious type. 
Increased recall tends to decrease precision, whereas increased precision tends to decrease recall.
The ``optimal'' tradeoff between precision and recall for a given system depends on the application at hand.
The integrity based solutions discussed in this section tend to offer higher precision rates than the pattern recognition techniques discussed in Sec.~\ref{sec:signatures_heuristics}, but they are difficult to update because custom changes to hardware and software are required, making scalability an issue.

It is important to realize that the \emph{component protection} techniques presented in this section are in practice often combined with more generic pattern recognition techniques discussed in Sec.~\ref{sec:signatures_heuristics}\cite{szor2005theart,sommer2010outside,singh2013hot}, for example, hardware and virtualization solutions might be used to achieve a clean view of memory, on which a signature scan can be run \cite{petroni2004copilot,garfinkel2003virtual}.



\subsection{Detecting Hooks}
\sad{Detect malwares that use hooking.}{Easy to implement.}{High false positive rates from legitimate benign hooks.}

If a stealth malware uses in-memory hooks as described in Sec.~\ref{sec:hooking}, IDSs can detect the malware by detecting its hooks.
Unfortunately, methods that simply detect hooks trigger high false alarm rates since hooks are not inherently malicious.
This makes weeding out false positives a challenging task.
Also, since DKOM is not a form of hooking, hook detection techniques cannot detect DKOM attacks.

Ironically, an effective approach to detect hooks is to hook common attack points. By doing so, an anti-malware may not only be able to detect a rootkit loading into memory, but may also be able to preempt the attack. This might be accomplished by hooking the API functions used to inject DLLs into a target process' context (cf. Sec.~\ref{sec:userhooking}) \cite{hoglund2005rootkits}. However, one must know what functions to hook and where to look for malicious attacks. Pinpointing attack vectors is not easy. For example, symbolic links are often not resolved to a common name until system call hooks have been executed\cite{hoglund2005rootkits}. Therefore, if the anti-malware relies on hooking the \Windows{SSDT} alone and matching the name of the target in the hook routine, an attacker can simply use an alias. Once hooks are observed, some tradeoff between precision and recall must be made: One can easily catch all rootkits loading into memory, and in doing so, create a completely unusable system (i.e., very high recall rates but extremely low precision rates).

Hook detection can be combined with signature and heuristic scans (discussed in Sec.~\ref{sec:signatures_heuristics}) for ingress point monitoring. Based on a signature of the hooked code, the ingress point monitoring system can determine whether or not to raise an alarm. In contrast to trying to detect rootkit hooks as a rootkit loads, \Name{VICE} \cite{hoglund2005rootkits,butler2004vice} uses memory scanning techniques that periodically inspect likely target locations of hooks such as the \Windows{IAT}, \Windows{SSDT}, or \Windows{IDT}. \Name{VICE} detects hooks based on the presence of unconditional jumps to memory values outside of acceptable address ranges. Acceptable ranges can be defined by \Windows{IAT} module ranges, driver address ranges, and kernel process address ranges. For example, a system call in the \Windows{SSDT} should not point to an address outside \Windows{ntoskrnl.exe}.

Generic inline hooks cannot feasibly be detected via this method. Fortunately, as we discussed in Sec.~\ref{sec:hooking}, hooks beyond the first few bytes of a function are rare, since they can result in strange behaviors, including noticeable slow down and outright program failure. For \Windows{SSDT} functions, unconditional jumps within the first few bytes outside of \Windows{ntoskrnl.exe} are indicators of hooks. \Windows{IAT} range checks require context switching into the process in question, enumerating the address ranges of the loaded DLLs, checking whether the function pointers in the \Windows{IAT} fall outside of their corresponding ranges, and recursively repeating this for all loaded DLLs.

A similar approach to \Name{VICE} was taken in the implementation of \Name{System Virginity Verifier} \cite{rutkowska2005system}, which attempts to separate malicious hooking from benign hooking by comparing the in-memory code sections of drivers and DLLs to their disk images. Since these sections are supposed to be read-only, they should match in most cases, with the exception of a few lines of self-modifying kernel code in the \Name{NTOS} kernel and hardware abstraction layers. Malicious hooks distinguish themselves from benign hooks when they exhibit discrepancies between in-memory and on-disk PE images, which will not occur under benign hooking \cite{rutkowska2005system}. Additionally, if the disk image is hidden then the hook likely corresponds to a rootkit.
One must be careful in this case to distinguish missing files, which can occur in legitimate hooking applications, from hidden files.
Other examples of image discrepancies associated with malicious hooks include failure of module attribution and code obfuscation.

An indirect approach to detecting hooks was implemented in \Name{Patchfinder 2} \cite{rutkowska2004detecting}, in the form of API call tracing.
This approach counts the number of instructions executed during an API call and compares the count to the number of instructions executed in a clean state. The intuition is based on the observation that in the context of rootkits, hooks almost always add instructions\cite{rutkowska2004detecting}. The technique requires proper baselining, which presents two challenges: first, deducing that the system is in a non-hooked state to begin with is difficult to establish, unless the system is fresh out of the box. Second, the \Name{Win32} API has many functions, which take many different arguments. Since enumerating all argument combination possibilities while acquiring the baseline is infeasible, API calls can vary substantially in instruction count even when unhooked.

\subsection{Cross-View Detection and Specification Based Methods}
\sad{Compare the output of API calls with that of low-level calls that are designed to do the same thing.}{Detects malware that hijacks API calls.}{Requires meticulous low-level code for to replicate functionality of most of the system API.}

Cross-view detection is a technique aimed to reveal the presence of rootkits. The idea behind cross-view detection \cite{rutkowska2005thoughts} is to observe the same aspect of a system in multiple ways, analogous to interviewing witnesses at a crime scene: just as conflicting stories from multiple witnesses likely indicate the presence of a liar, if different observations of a system return different results, the presence of a rootkit is likely. First, OS objects -- processes, files, etc. -- are enumerated via system API calls. This count is compared to that obtained using a different approach not reliant on the system API. For example, when traversing the file system, if the results returned by \Windows{FindFirstFile} and \Windows{FindNextFile} are inconsistent with direct queries to the disk controller, then a rootkit that hides files from the system is likely present.

One of the advantages of cross-view detection is that -- if implemented correctly -- maliciously hooked API calls can be detected with very few false positives because legitimate applications of API hooking rarely change the outputs of the API calls.
Depending on the implementation, cross-view detection may or may not assume an intact kernel, and therefore may even be applied to detect DKOM.
The main disadvantage of cross-view detection is that it is difficult to implement, especially for a commercial OS\cite{butler2005windows3}.
API calls are provided for a reason: to simplify the interface to kernel and hardware resources.
Cross-view detection must circumvent the API, in many cases providing its own implementation.
Theoretically, in most cases combinations of other API calls could be used in place of a from-scratch implementation.
However, API call combinations are susceptible to the risk of other hooked API calls or duplicate calls to the same underlying code for multiple API functions, a common feature of the \Name{Win32} API \cite{tanenbaum2007modern}.

Several cross-view detection tools have been developed over the years. \Name{Rootkitrevealer} \cite{cogswell2006rootkitrevealer} by \Name{Windows SysInternals} applies a cross-view detection strategy for the purposes of detecting persistent rootkits, i.e., disk-resident rootkits that survive across reboots. \Name{Rootkitrevealer} uses the \Name{Windows} API to scan the file system and registry, and compares the results to a manual parsing of the file system volume and registry hive. \Name{Klister} \cite{rutkowska2004detecting} detects hidden processes in \Name{Windows 2000} by finding contradictions between executive process entries and kernel process entries used by the scheduler. \Name{Blacklight} \cite{butler2006raide} combines  both hidden file detection and hidden process detection. \Name{Microsoft}'s \Name{Strider Ghostbuster} \cite{beck2005detecting} is similar to \Name{Rootkitrevealer}, except that it also detects hidden processes and it has the ability to compare an ``inside the box'' infected scan with an ``outside the box'' scan, in which the operating system is booted from a clean version.

If properly applied, cross-view detection offers high precision rootkit detection\cite{butler2005windows3}. However, cross-view detection alone provides little insight on the \emph{type of the rootkit} and must be combined with recognition methods (e.g., signature/behavioral) to attain this information \cite{szor2005theart,butler2005windows3}. Cross-view detection methods are also cumbersome to update because they require new code, often interfacing with the kernel. Determining, which areas to cross-view, is also a challenging task\cite{butler2005windows3}.

%
\subsection{Invariant Specification}
\sad{Define constraints of an uninfected system.}{Detects DKOM attacks reliably.}{Constraints need to be well-specified, often by hand, and are highly platform-dependent.}

A related approach to cross-view detection, especially applied to detecting DKOM, involves pinpointing kernel invariants -- aspects of the kernel that should not change under normal OS behavior -- and periodically monitoring these invariants. One example of a kernel invariant is that the length of the executive and kernel process linked lists should be equal, which is violated in the case of process hiding (cf. Fig.~\ref{fig:DKOM}). Petroni et al.~\cite{petroni2006architecture} introduce a framework for writing security specifications for dynamic kernel data structures. Their framework consists of five components: a low-level monitor used to access kernel memory, a model builder to synthesize the raw kernel memory binary into objects defined by the specification, a constraint verifier that checks the objects constructed by the model builder against the security specifications, response mechanisms that define the actions to take upon violation of a constraint, and finally, a specification compiler, which compiles specification constraints written in a high-level language into a form readily understood by the model builder.

Compelling arguments can be made in favor of the kernel-invariant based security specification approaches described above\cite{petroni2006architecture}: first, they allow a decoupling of site-specific constraints from system-specific constraints. An organization may have a security policy that forbids behavior not in direct violation of proper kernel function (e.g., no shell processes running with root UUID in \Name{Linux}). Via a layered framework, specifications can be added without changing low-level implementations. Unlike signature-based approaches relying on rootkits having overlapping code fragments with other malwares, kernel-invariant specifications catch all DKOM attacks that violate particular specification constraints with only few false positives. The specification approach can even be extended beyond DKOM.  However, using kernel invariant specification is not without its own difficulties. Proper and correct framework implementation is a tremendous programming effort in itself\cite{petroni2006architecture}. For closed-source operating systems like \Name{Windows}, full information about kernel data structures and their implementation is seldom available, unless the specification framework tool is being developed as part of or in cooperation with the operating system vendor. Specification approaches can also exhibit false positives, for example, if kernel memory is accessed asynchronously via an external PCI interface like \Name{Copilot} \cite{petroni2004copilot}, a legitimate kernel update to a data structure may trigger a false positive detection simply because the update has not completed. Finally, the degree to which the invariant-specification approach works depends on the quality of the specification\cite{petroni2006architecture}. Correct specifications require in-depth domain specific knowledge about the kernel and/or about the organization's security policy. Due to the massive sizes and heterogeneities of operating systems, even those of similar distribution, discerning a complete list of specifications is implausible without incorrect specifications that result in false positives. While a similar approach may have applications to other types of stealth malwares, Petroni et al.~\cite{petroni2006architecture} introduced invariant specification as specific solution tailored to DKOM rootkits . Although invariant specification provides more readily available diagnostic information than cross-view detection because it tells which invariants are violated, invariant specification cannot discern the type of DKOM rootkit. Hence, more generic signature/behavioral techniques are required.

\subsection{Hardware Solutions}
\sad{Via hardware interface, use a clean machine to monitor another machine for the presence of rootkits/stealth malware.}{Does not require an intact kernel on the monitored machine.}{Cannot interpose execution of malicious code.}

The key motivation behind hardware based integrity checking is quite simple: a well-designed rootkit that has successfully subverted the OS kernel, or theoretically even the virtual layer and BIOS of a host machine, can return a spurious view of memory to a host based intrusion detection system (HIDS) such that the HIDS has no way of detecting the attack because its correct operation requires an intact kernel.
Rather than relying on the kernel to provide a correct view of kernel memory, hardware solutions have been developed. For example, \Name{Copilot} \cite{petroni2004copilot} uses direct memory access (DMA) via the PCI bus to access kernel memory from the hardware of the host machine itself and displays that view of memory to another machine. This in turn subverts any rootkit's ability to change the view of memory, barring a rootkit implemented in hardware itself. Depending on the hardware integrity checker in question, further analysis of kernel memory on the host machine may be performed via a supervisory machine alone, or alternatively with the aid of additional hardware. \Name{Copilot} uses a coprocessor to perform fast hashes over static kernel memory and reports violations to a supervisory machine. Analysis mechanisms similar to those in \cite{srivastava2011operating,siddiqui2008survey} are employed on the supervisory machine in conjunction with DMA in order to properly parse kernel memory.

%
Using DMA to observe the memory layout of the host system from a supervisory system is appealing since a correct view of host memory is practically guaranteed. However, like all of the techniques that we have discussed, hardware based integrity checking is no silver bullet. In addition to the added expense and annoyance of requiring a supervisory machine, DMA based rootkit detection techniques can only detect rootkits, but they cannot intervene in the hosts execution. They have no access to the CPU and, therefore, cannot prevent or respond to attacks directly.
This CPU access limitation not only means that CPU registers are invisible to DMA, it also means that the contents of the CPU cache cannot be inspected, leaving the theoretical possibility of a rootkit hiding malicious code in the cache. However, a more pressing concern is that because DMA approaches operate at a lower level than the kernel they do not have a clear view of dynamic kernel data structures, which requires that these structures need to be located in memory, a problem discussed in \cite{dolangavitt2009robust}. Even after locating the kernel data structures, there remains a synchronization issue between DMA operations and the host kernel: DMA cannot be used to acquire kernel locks on data structures. Consequently, race conditions result when the kernel is updating a data structure contemporaneous with a DMA read. False positives were observed by Baliga et al.~\cite{baliga2011data} for precisely this reason. An inelegant solution \cite{petroni2004copilot} is to simply re-read memory locations containing suspicious values. Another consideration when implementing DMA approaches is the timing of DMA scans. Both \cite{petroni2004copilot} and \cite{baliga2011data} employed synchronous DMA scans, which are theoretically susceptible to timing attacks. Petroni et al.~\cite{petroni2004copilot} suggested introducing randomness to the scan interval timings to overcome this susceptibility.

\subsection{Virtualization Techniques}
\label{sec:virtualization}
\sad{Use virtual environments to detect malware.}{Can be used to detect kernel-level rootkits and interpose state.}{Vulnerable to anti-emulation.}
Virtualization, though technologically quite different from DMA, aims to satisfy the same goal of inspecting resources of the host machine without relying on the integrity of the operating system. Several techniques for rootkit detection, mitigation, and profiling that leverage virtualization have been developed, including \cite{srivastava2011operating, rutkowska2005system, garfinkel2003virtual, seshadri2007secvisor, riley2008guest}. The idea behind virtualization approaches is to involve a virtual machine monitor, a.k.a. the \emph{hypervisor}, in the inspection of system resources. Since the hypervisor resides at a higher level of privilege than the guest OS, either on the hardware itself or simulated in software, and the hypervisor controls the access of the guest OS to hardware resources, the hypervisor can be used to inspect these resources even if the guest OS is entirely compromised. Unlike \Name{Copilot}'s approach, in which kernel writes and DMA reads are unsynchronized, the hypervisor and the guest OS kernel are synchronous since the guest OS relies on the hypervisor for resources. Moreover, the hypervisor has access to state information in the CPU, meaning that it can interpose state, a valuable ability not only for rootkit detection, prevention and mitigation, but also for computer forensics. Additionally, the hypervisor can be used to enforce site specific hardware policies, for example the hypervisor can prevent promiscuous mode network interface operation \cite{garfinkel2003virtual}. Hypervisors themselves may be vulnerable to attack, but the threat surface is much smaller than for an operating system: hypervisors have been written in as little as 30,000 lines of C code as opposed to the tens of millions of lines of code in modern \Name{Windows} and \Name{Linux} distributions. Significant security validations on hypervisors have also been conducted by academia, private security firms, the open source community, and intelligence organizations (e.g., CIA, NSA) \cite{garfinkel2003virtual}.

%
Garfinkel and Rosenblum \cite{garfinkel2003virtual} created \Name{Livewire}, a proof of concept intrusion detection system residing at the hypervisor layer. The authors refer to their approach as virtual machine introspection since the design utilizes an OS interface to translate raw hardware state into guest OS semantics and inspect guest OS objects via a policy engine, which interfaces with the view presented by the translation engine. The policy engine effectively is the intrusion detection system, which performs introspection on the virtual machine. The policy engine can monitor the machine and can also take mitigation steps such as pausing the state of the VM upon certain events or denying access to hardware resources.

A particular advantage of virtualization is that it can be leveraged to prevent rootkits from executing code in kernel memory -- a task that all kernel rootkits must perform to load themselves into memory in the first place\cite{seshadri2007secvisor}. This includes DKOM rootkits: although the changes to kernel objects themselves cannot be detected as code changes to the kernel, code must be introduced at some point to make these changes.  To this end, Seshadri et al.~\cite{seshadri2007secvisor} formulated \Name{SecVisor}. In contrast to the software-centric approach of \Name{Livewire}, \Name{SecVisor} leverages hardware support for virtualization of the \Name{x86} instruction set architecture as well as \Name{AMD}'s secure virtual machine technologies. \Name{SecVisor} intercepts code via modifications to the CPU's memory management unit (MMU) and the I/O memory management unit (IOMMU), so that only code conforming to a user supplied policy will be executable. As such, kernel code violating the policy will not run on the hardware. In fact, \Name{SecVisor}'s modification to the IOMMU even protects the kernel from malicious writes via a DMA device. \Name{SecVisor} works by allowing transfer of control to kernel mode only at entry points designated in kernel data structures, then performing comparisons to shadow copies of entry point pointers. This approach is analogous to that used in memory integrity checking modules of heavyweight dynamic binary instrumentation (DBI) frameworks like \Name{Valgrind} \cite{nethercote2007valgrind}.

Unfortunately, \Name{SecVisor} has several drawbacks. First, modern \Name{Linux} and \Name{Windows} distributions mix code and data pages \cite{riley2008guest}, while \Name{SecVisor}'s approach -- enforcing \textit{write XOR execute} ($W \oplus X$) permissions for kernel code pages through hardware virtualization -- assumes that kernel code and data are not mixed within memory pages. The approach also fails for pages that contain self-modifying kernel code. Second, \Name{SecVisor} requires modifications to the kernel itself -- a difficult proposition for adoption on closed-source operating systems like \Name{Windows}.

Riley et al.~\cite{riley2008guest} formulated \Name{NICKLE} (No Instruction Creeping into Kernel Level Executed), which, like \Name{SecVisor}, leverages virtualization to prevent execution of malicious code in kernel memory. \Name{NICKLE} approaches the problem via software virtualization and overcomes some of the limitations of \Name{SecVisor}. \Name{NICKLE} works by shadowing every byte of kernel code in a separate shadow memory writable only by the hypervisor. Because the hypervisor resides in a higher privilege domain than the kernel, even the kernel cannot modify the shadowed code. The shadowed code gets authenticated either during bootstrapping, when the kernel is loaded into memory, or when drivers are mounted or unmounted. Authentication consists of cryptographic hash comparisons of code segments with known good values taken by OS vendors or distribution maintainers. When the operating system requires access to kernel-level code an indirection mechanism in the hypervisor reroutes this request to shadow values. To maintain transparency to the guest OS, this guest memory address indirection is implemented after the ``virtual to physical'' address translation in the hypervisors MMU. When the guest VM attempts to execute kernel code, a comparison is made to shadow memory. If the code is the same, then the shadow memory copy is executed. If the kernel memory and shadow memory code differ then one of several responses can be taken including logging and observing -- an approach extended by Riley et al.~\cite{riley2009multi} for rootkit profiling -- rewriting the malicious kernel code with shadow values and continuing execution, or breaking execution.
\Name{NICKLE}'s approach has two key advantages over \Name{SecVisor}: first, it does not assume homogeneous code and data pages. Second, it does not require any modifications to kernel code. These benefits, however, incur hits in speed due to software virtualization and memory indirection costs and require a two-fold increase in memory for kernel code\cite{riley2008guest}.  An additional complication arises from code relocation: when driver code is relocated in kernel memory, cryptographic hashes change. Riley et al.~\cite{riley2008guest} handle this problem by tracking and ignoring relocated segments. Also, the \Name{NICKLE} implementation in \cite{riley2008guest} does not support kernel page swapping, which  would need to ensure that swapped in pages had the same cryptographic hash as when they were swapped out. Finally, \Name{NICKLE} is ineffective in protecting self-modifying kernel code, a phenomenon present in both \Name{Linux} and \Name{Windows} kernels.

Srivastava et al.~\cite{srivastava2011operating} leverage virtualization in their implementation of \Name{Sherlock} -- a defense system against the \Name{Illusion} attack mentioned in Sec.~\ref{sec:hybrid}. \Name{Sherlock} uses the \Name{Xen} hypervisor to monitor system call execution paths. Specifically, the guest OS is assumed to run on a virtual machine controlled by the \Name{Xen} hypervisor. Monitoring of memory is conducted by the hypervisor itself with the aid of a separate security VM for system call reconstruction, analysis, and notification of other intrusion detections systems. Watchpoints are manually and strategically placed in kernel memory off-line, and a \Name{B\"uchi} automaton \cite{srivastava2011operating} is constructed, which efficiently describes the expected and unexpected behavior of every system call in terms of watch points. Each watch point contains the \Windows{VMCALL} instruction, so that when it is hit, it notifies the hypervisor. Watch point identifiers are passed to the automaton as they are executed. During normal execution, the automaton remains in benign states and watch points are discarded. When a malicious state is reached, the hypervisor logs watch points and suspends the state of the guest VM. The function specific parameters at each watch point corresponding to a malicious state are then passed to the security VM for further analysis. An important consideration of this implementation is where to place watchpoints to balance effectiveness and efficiency. Srivastava et al.~\cite{srivastava2011operating} manually chose watch point locations based on a reachability analysis of a kernel control flow graph, but suggest that an autonomous approach \cite{ganapathy2005automatic} could be implemented.

%

\section{Pattern-Based Stealth Malware Countermeasures}
\label{sec:signatures_heuristics}

Pattern-based approaches aim to achieve more generic recognition of heterogeneous malwares.
While these approaches offer potential for efficient updates and scalability beyond most component protection techniques, their increased generalization causes them to tend to exhibit higher recall rates but lower precision rates.
Pattern-based approaches can be applied on static code fragments or on dynamic memory snapshots or behavioral data (e.g., system/API calls, network connections, CPU usage, etc.) and may be coupled with component protection approaches  \cite{szor2005theart,jang2016detecting,faruki2016droidanalyst}.
Static analysis has the advantage that it is fast\cite{szor2005theart}, since the raw code is inspected but not executed; there is no need for an emulated environment.
However, dynamic code mutation mechanisms outlined in Sec.~\ref{sec:code_mutation}, are often able to hide functionalities from static code analyzers \cite{jang2016detecting}, and obfuscated code recognition techniques discussed in Sec.~\ref{sec:feature_state} often rely on an emulated environment for decryption.
Dynamic analysis tools that leverage emulated environments (cf. Sec.~\ref{sec:antiemulation}) are not fooled so easily, since much of the underlying code, data, and behavior of the malware is revealed.
However, dynamic analysis is potentially vulnerable to anti-emulation techniques (cf. Sec.~\ref{sec:antiemulation}).
While dynamic analysis techniques generally suffer lower false-positive rates than static analysis techniques \cite{jang2016detecting} dynamic techniques are far slower than static approaches \cite{szor2005theart} due to the need for an emulated environment, and can only be feasibly executed on a small number of code samples for short amounts of time. Consequently, hybrid approaches \cite{faruki2016droidanalyst} are often employed, in which static methods are used to select suspicious samples, which are then passed to a dynamic analysis framework for further categorization.

\subsection{Signature Analysis}
\sad{Compare code signatures with database of malicious signatures via exact-matching or machine-learnt techniques.}{Detects known malwares reliably.}{Difficult to detect novel malware types.}

Code-signature-based malware defenses are techniques that compare malware signatures -- fragments of code or hashes of fragments of code -- to databases of signatures associated with known attacks.
Although signatures cannot be directly used to discover new exploits \cite{butler2005windows3}, they can do so indirectly due to component overlap between malwares \cite{abouassaleh2004ngram,reddy2005new}. 
Ironically, shared stealth components have sometimes given away the presence of malwares that would have otherwise gone unnoticed \cite{szor2005theart}. Moreover, some byte sequences of length $n$ ($n$-grams) specific to a common type of exploit are often present even under metamorphism of the code. Machine learning approaches to malware classification via n-gram and sequence analysis have been widely studied and deployed as integral components of anti-malware systems for more than ten years \cite{szor2005theart,siddiqui2008survey}.

While most in-memory rootkit signature recognition strategies behave much like on-disk signature strategies for detecting and classifying malicious code by matching raw bytes against samples from known malware, DKOM rootkit detection requires a different approach.
Since DKOM involves changing existing data fields within OS data structures to hide them from view of certain parts of the OS, DKOM signature scanning techniques instead perform memory scans using signatures designed to pinpoint hidden data structures in kernel memory.
Surprisingly, memory signature scans are useful both in live and forensics contexts. Chow et al.~\cite{chow2005shredding} demonstrated that structure data in kernel memory can survive up to 14 days after de-allocation, provided that the machine has not been rebooted.
Schuster \cite{schuster2006searching} formulated a series of signature rules for detecting processes and threads in memory, for the general purpose of computer forensics.
Several spinoffs of this approach have been implemented.
Unfortunately, many of these signature approaches can be subverted by rootkits that change structure header information.
Dolan-Gavitt et al.~\cite{dolangavitt2009robust} employed an approach to automatically obtain signatures for kernel data structures based on values in the structures that, if modified, cause the OS to crash.
The approach includes data structure profiling and fuzzing stages.
In the profiling stage, a clean version of the operating system is run, while a variety of tasks are performed.
Kernel data structure fields commonly accessed by the OS are logged.
The goal of the profiling stage is to determine fields that the OS often accesses and weed out fields that are not widely used for consideration as signatures. 
The fuzzing stage consists of running the OS on a virtual machine, pausing execution, and modifying the values in the candidate structure.
After resuming, candidate structure values are added to the signature list if they cause the kernel to crash.
The approach in \cite{dolangavitt2009robust} is in many ways the complement of the kernel invariant approach in \cite{baliga2011data}.
Instead of traversing kernel data structures and examining which invariants are violated, Dolan-Gavitt et al. scan all of kernel memory for plausible data structures.
If certain byte offsets within the detected structures do not contain signatures consistent with certain values, then the detections cannot correspond to actual data structures used by the kernel because otherwise they would crash the operating system.
A limitation of the approach in \cite{dolangavitt2009robust} is that it is susceptible to attack by scattering copies or near copies of data structures throughout kernel memory.

\subsection{Behavioral/Heuristic Analysis}
\sad{Derived from system behavior rather than code fragments.}{Not affected by attempts to hide malicious code.}{Cannot detect malware prior to execution.}

On the host level, signatures are not the only heuristic used for intrusion detection. System call sequences for intrusion and anomaly detection \cite{feng2003anomaly,forrest1996sense,giffin2002detecting,hofmeyr1998intrusion,krohn2007information,mutz2007exploiting,sekar2001fast} are an especially popular alternative for rootkit analysis since hooked \Windows{IAT} or \Windows{SSDT} entries often make repetitive patterns of system calls.
Interestingly,  rootkits can also be detected by network intrusion detection systems (NIDSs), because rootkits in the wild are almost always small components of a larger malware. The larger malware often performs some sort of network activity such as C2 server communication and synchronization across infected machines, or infection propagation.
This is even true for some of the most sophisticated stealth malwares that leverage rootkit technologies to hide network connections from the host \cite{bencsath2012cousins}, while a rootkit cannot hide connections from a network.
Therefore, signature scans at the network level as well as traffic flow analysis techniques can give away the presence of the larger malware as well as the underlying rootkit.
Botnets with rootkits that effectively hide the behavior of an individual host may be easier to detect when analyzing macro, network-level traffic~\cite{yu2015fool,yu2015modeling}.
NIDSs also have the advantage that they provide isolation between the malware and the intrusion detection system, reducing a malware's capacity to spoof or compromise the IDS.
However, NIDSs have no way of inspecting the state of a host or interposing a host's execution at the network level.
A hybrid approach, which extends the concept of cross-view detection is to compare network connections from a host query with those detected at the network level \cite{fink2005visual}.
A discrepancy indicates the presence of a rootkit.

\subsection{Feature Space Models vs. State Space Models}
\label{sec:feature_state}
\sad{Classify code sequences.}{Can detect similar malicious code patterns.}{Cannot detect unseen malicious code.}
\begin{figure*}[!ht]
  \centering
  \subfloat[Binary OPCODE n-gram histogram classification\label{fig:FeatureSpace:SVM}]{\includegraphics[height=.15\textheight]{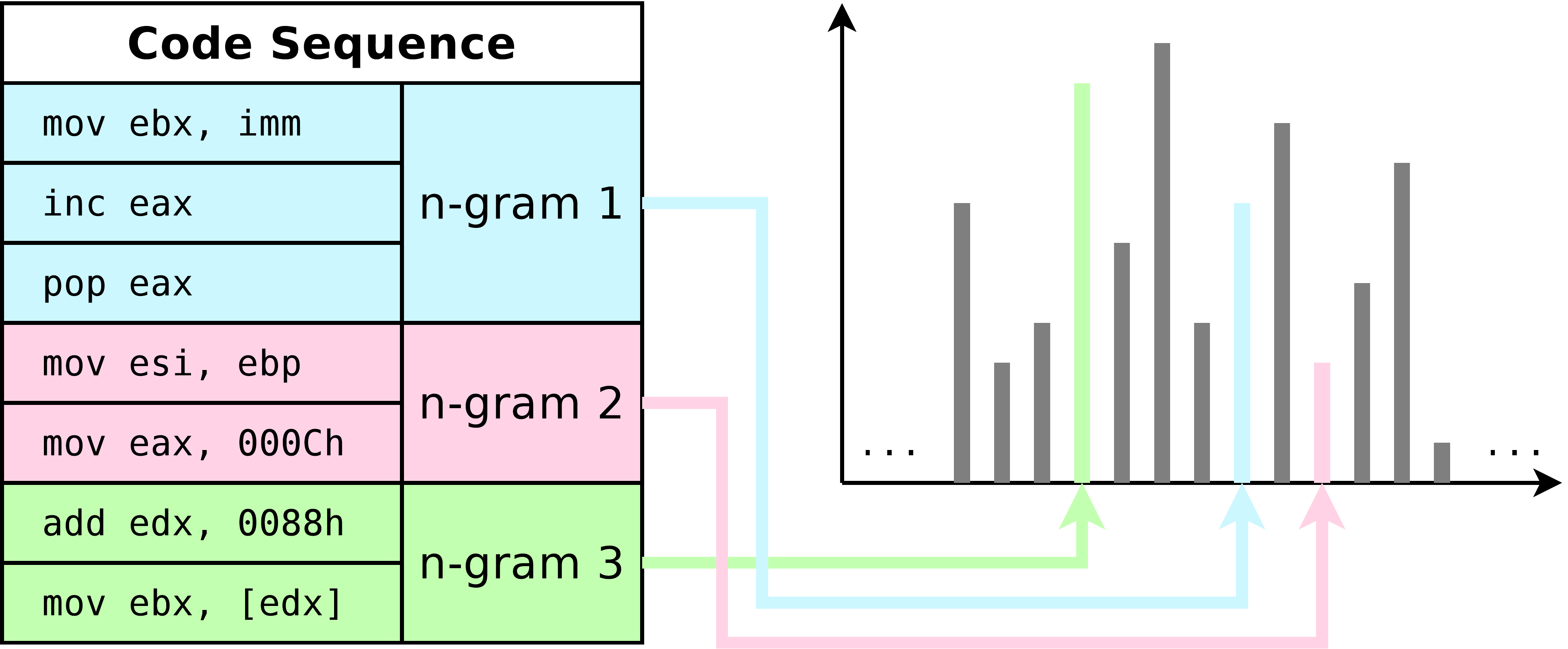}\hspace*{.05\textwidth}\includegraphics[height=.15\textheight]{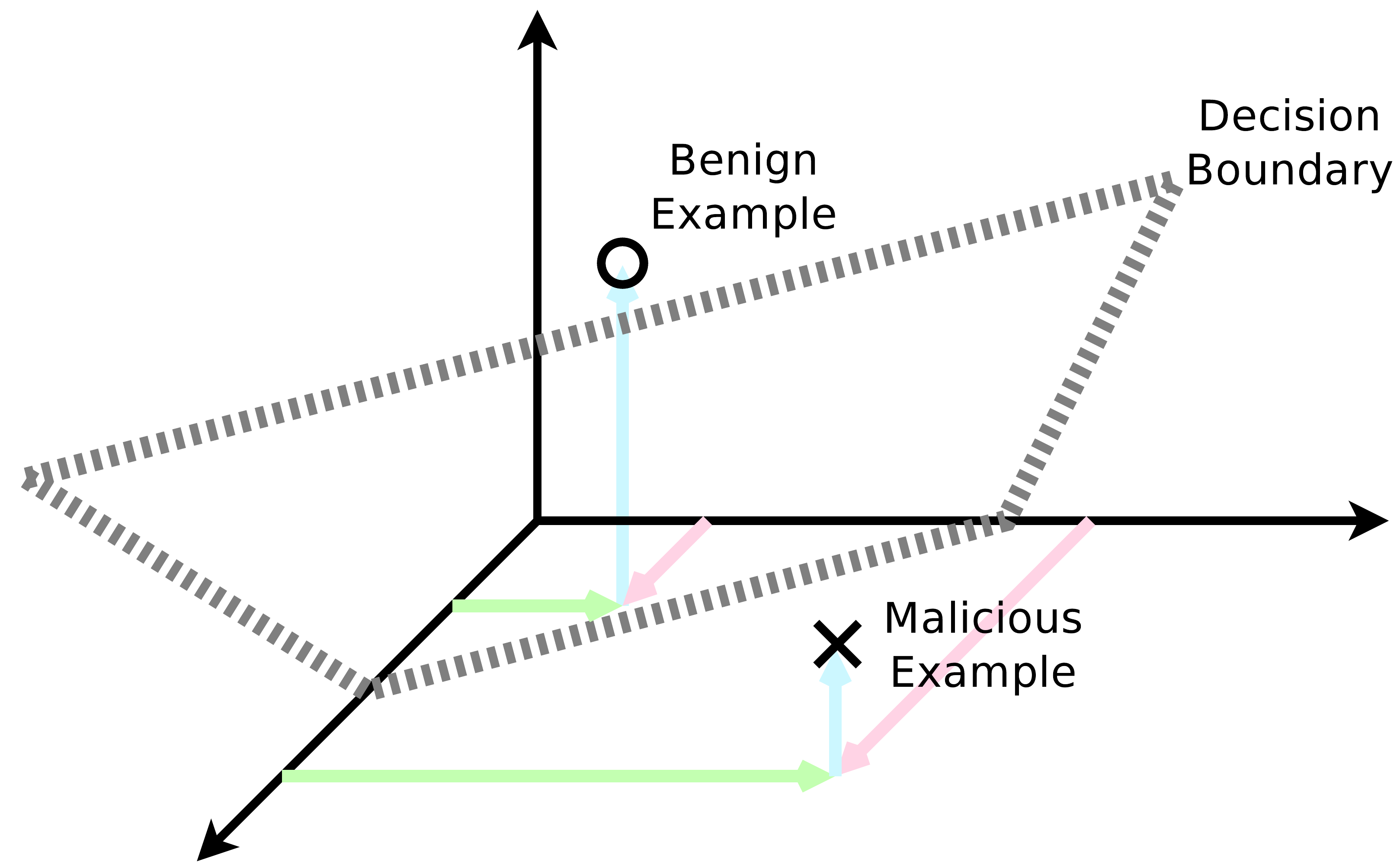}}\\
  \subfloat[OPCODE sequence analysis\label{fig:FeatureSpace:HMM}]{%
    \begin{minipage}[c]{.35\textwidth}\centering\includegraphics[height=.17\textheight]{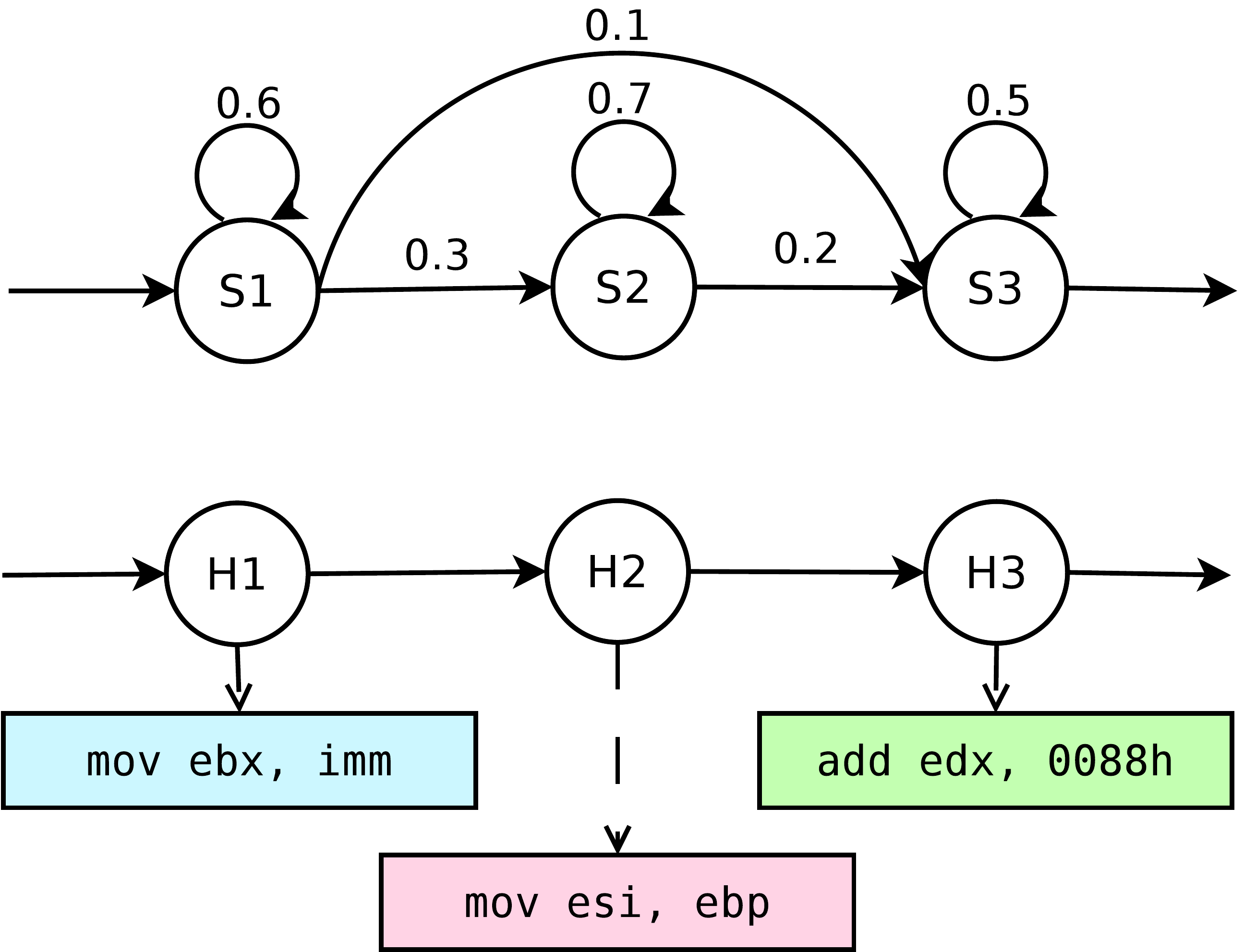}\end{minipage}%
    \begin{minipage}[c]{.28\textwidth}\small\begin{center}
\renewcommand\arraystretch{1.3}
\setlength{\tabcolsep}{.4em}
State transition matrix\\[2ex]
\begin{tabular}{r||c|c|c|c}
  State & S1  & S2  & S3  & \dots{}\\\hline\hline
  S1    & 0.6 & 0.3 & 0.1 & \dots{}\\\hline
  S2    & 0   & 0.7 & 0.2 & \dots{}\\\hline
  S3    & 0   & 0   & 0.5 & \dots{}\\\hline
  \dots &\dots&\dots&\dots
\end{tabular}
\end{center}
\end{minipage}%
    \begin{minipage}[c]{.32\textwidth}\small\begin{center}
\renewcommand\arraystretch{1.3}
\setlength{\tabcolsep}{.4em}
Observation probabilities\\[2ex]
\begin{tabular}{c||c|c|c|c}
  OPCODE & S1  & S2  & S3  & \dots{}\\\hline\hline
  \scriptsize\tt mov ebx, imm & 0.6  & 0.1  & 0.01 & \dots{}\\\hline
  \scriptsize\tt mov esi, epb & 0.07 & 0.5  & 0.06 & \dots{}\\\hline
  \scriptsize\tt add edx, 0088h & 0    & 0.04 & 0.81 & \dots{}\\\hline
  \dots &\dots&\dots&\dots
\end{tabular}
\end{center}
\end{minipage}%
  }
  \cap{fig:FeatureSpace}{Feature Space vs. State Space OPCODE classification}{This diagram depicts \protect\subref*{fig:FeatureSpace:SVM} a schematic interpretation of a linear classifier that separates benign and malicious OPCODE $n$-grams in feature space and \protect\subref*{fig:FeatureSpace:HMM} a sequential OPCODE analysis using a hidden Markov model. The feature space model must explicitly treat histograms of $n$-grams as independent dimensions for varying values of $n$ in order to capture sequential relationships. This approach is only scalable to a few sequence lengths. HMMs, on the other hand, impose a Markov assumption on a sequence of hidden variables which emit observations. State transition and observation probability matrices are inferred via expectation maximization on training sequences. An HMM factor graph is shown on the bottom left.}
\end{figure*}

As discussed above, code signatures and application behaviors/heuristics can be used in a variety of ways to detect and classify intrusions, and they operate across many levels of the intrusion detection hierarchy. For example, encrypted viruses are particularly robust against code signatures -- until they are decrypted -- but during emulation, once the virus is in memory, it might be particularly susceptible to signature analysis. This analysis may range from a simple frequency count of OPCODES to more sophisticated machine-learning techniques.

Machine learning models can be divided into feature space and state space models.
Examples of both are shown in Fig.~\ref{fig:FeatureSpace}, in which code fragments are classified as malicious or benign based on their OPCODE $n$-grams.
Feature space models aim to treat signature/behavioral features as a spatial dimension and parameterize a manifold within this high-dimensional feature space for each class.
Feature space models can be further broken down into generative and discriminative models.
Generative models aim to model the joint distribution $P(x,y)$ of target variable $y$ and spatial dimension $x$, and perform classification via the product rule of probability: $P(y|x) = \frac{P(x,y)}{P(x)}$.
Discriminative classifiers aim to model $P(y|x)$ directly \cite{bishop2006pattern}.
By treating the frequencies of distinct n-gram hashes as elements of a high-dimensional feature vector, for example, the \emph{input feature space} becomes the domain of these vectors.
Support vector machines (SVMs), which are discriminative feature space classifiers, aim to separate classes by fracturing the input feature space (or some transformation thereof) by a hyperplane that maximizes soft class margins.
An advantage of feature space models is that in high dimensions, even if only a few of the dimensions are relevant, different classes of data tend to separate\cite{bishop2006pattern}.
However, feature space models do not explicitly account for probabilistic dependencies, and a good feature space from a classification accuracy perspective is not necessarily intuitive.

State space models are used to infer probabilities about \emph{sequences}.
They leverage the fact that certain sequences of instructions exist within malicious binary due to functional overlap as well as general lack of creativity and laziness of malware authors. State space models can also be applied to functional sequences (e.g., sequences of system calls or network communications). The intuition is that we can use certain types of functional behaviors to describe classes of malware in terms of what they do, for example, ransomwares like \Name{CryptoLocker} typically generate a key that they use to encrypt files on disk and subsequently attempt to send that key to a C2 server. After a certain amount of time, they remove the local copy of the key and generate a ransom screen demanding money for the key \cite{ogorman2012ransomware}. State space models for intrusion recognition aim to recognize these sorts of malicious sequences.


The most common type of state space models are based in some form on the Markov assumption -- that recent events will be independent of events that happened in the far past.
While the Markov assumption is not always valid, it makes sequential inference tractable and is often reasonable.
For example, if the last fifty assembly instructions were devoted to adding elements from two arrays together and incrementing respective pointers, with no other knowledge, it is a reasonable assumption that the next few instructions will add array elements.
On the other hand, knowing that ``hello world'' was printed to the screen a million instructions ago provides little information about the probability of the next instruction.
Hidden Markov models (HMMs) are perhaps the most widely used type of Markov models \cite{bishop2006pattern} and have been particularly useful in code analysis including recognition of  metamorphic viruses\cite{sridhara2013metamorphic,venkatesan2008code,venkatachalam2010detecting,runwal2012opcode,lin2011hunting,desai2008towards,wong2006analysis,wong2006hunting,attaluri2009profile}. HMMs assume that latent variables, which take on \emph{states}, are linked in a Markov chain with conditional dependencies on the previous states. The \textit{order} of the HMM corresponds to the number of previous states on which the current state depends, for example, in an $n$-th order HMM the current state depends only on the previous $n$ states.

In HMMs, previous states are fused with current states via a transition probability matrix $A$ governing the Markov chain, and an observation probability matrix $B$ -- the probability of observing the data in a given state -- as well as an initial state vector $\pi$. $A$, $B$, and $\pi$ can be estimated via expectation maximization (EM) inference on observation sequences $O$, which aims to find the maximum likelihood estimate (MLE) of a sequence of observations, i.e., $\argmax_{\lambda}P(O|\lambda)$, where $\lambda=(A,B,\pi)$. Although EM is guaranteed to converge to a local likelihood maximum, it is not guaranteed to converge to the global optimum. In the context of HMMs, this inference is usually carried out via the Baum-Welch algorithm \cite{bishop2006pattern} (aka. the Forward-Backward algorithm), which iterates between forward and backward passes and an update step until the likelihood of the observed sequence $O$ is maximized with respect to the model.

The usage of HMMs for metamorphic virus detection has been documented in \cite{sridhara2013metamorphic,venkatesan2008code,venkatachalam2010detecting,runwal2012opcode,lin2011hunting,desai2008towards,wong2006analysis,wong2006hunting,attaluri2009profile}.\footnote{HMMs are used for many sequential learning problems and have several different notations. Here, we borrow notation from \cite{sridhara2013metamorphic}}
These works assume a predominantly decrypted virus body, i.e., little to no encryption within the body to begin with, or that a previously encrypted metamorphic has been decrypted inside an emulator.
The number of hidden states and therewith the state transition matrix is generally chosen to be small (2-3), while observation matrix is larger, with rows consisting of conditional probabilities of OPCODES for given states.
For metamorphic detection, the semantic \emph{meaning} of the states themselves is unclear as is the optimal number of hidden states -- they only reflect some latent structure within the code.
This contrasts with other applications of HMMs, for example, in handwriting sequence recognition, the latent structure behind a noisy scrawl of an ``X'' is the letter ``X'' itself; thus with proper training data there should be 26 latent variables (for the English alphabet) with transition probabilities corresponding to what one might expect from an English dictionary, e.g., a ``T'' $\rightarrow$ ``H'' transition is much more likely than a ``T'' $\rightarrow$ ``X'' transition.

A common metamorphic virus recognition measure is the thresholded negative log-likelihood probability per \Name{OPCODE} \cite{wong2006analysis,runwal2012opcode,sridhara2013metamorphic} obtained from a \textit{forward} pass on an HMM, i.e.:
$$-\frac{\log(p(O_1,..,O_N,z_1,..,z_N))}{N},$$
where $O_1, \hdots, O_N$ are \Name{OPCODE}s in an $N$-length program and $z_1, \hdots, z_N$ are the hidden variables. The per-opcode normalization is required because different programs have different lengths. Most of the HMMs used in these works are first-order HMMs, in which the state space probability distribution of hidden variable $z_n$ is conditioned only on the value of $z_{n-1}$ and the current observation. For a $k-th$ order HMM, the probability of $z_{n}$ is conditioned on $z_{n-1} \hdots z_{n-k}$. However, the time complexity of HMMs increases exponentially with their order.
Although in their works \cite{sridhara2013metamorphic,venkatesan2008code,venkatachalam2010detecting,runwal2012opcode,lin2011hunting,desai2008towards,wong2006analysis,wong2006hunting,attaluri2009profile} the authors claim that the number of hidden variables did not seem to make a difference, they might if higher-order Markov chains were used.
As Lin and Stamp \cite{lin2011hunting} discuss, one problem with HMMs is that it ultimately measures similarity between code sequences; if the inter-class to intra-class sequential variation is large enough due to some exogenous factor such as very similar non-viral code in train/test, then HMM readout may be error-prone.

\section{Toward Adaptive Models for Stealth Malware Recognition}
\label{sec:open_world_ids}

A large portion of the malware detected by both component protection and generic recognition techniques is previously observed malware with known signatures, deployed by \emph{script kiddies} -- attackers with little technical expertise that predominantly use pre-written scripts to propagate existing attacks \cite{zanero2004unsupervised}.
Systems with up-to-date security profiles are not vulnerable to such attacks.

Sophisticated stealth malwares, on the other hand, have propagated undetected for long periods of time because they do not match known signatures, do not attack protected system components with previously seen patterns, and mask harmful behaviors as benign.
To reduce the amount of time that these previously unseen stealth malwares spend propagating in the wild, component protection and generic recognition techniques alike must be able to quickly recognize and adapt to new types of attacks.
Typically, it is slower to adapt component techniques than it is to adapt generic recognition techniques because new hardware and software are required. However, even more generic algorithmic techniques may take time to update and this must be factored into the design of an intrusion recognition system.

The choice of the algorithm for efficient updates is only one of several considerations that must be addressed in an intrusion recognition system. More elementary is how to autonomously make a decision that additional training data is needed and that the classifier \textit{needs to be updated} in the first place. In short, an intrusion recognition system must be \textit{adaptive} in order to efficiently mitigate the threat of stealth malware. It must also be \textit{interpretable} to yield actionable information to human operators and incident response modules. Unfortunately, many systems proposed in the literature are neither adaptive nor interpretable. We have isolated six flawed modeling assumptions, which we believe must be addressed at the algorithmic level. We discuss these flawed assumptions in Sec.~\ref{sec:flawed_assumptions}, and propose an algorithmic framework for attenuating them Sec.~\ref{sec:adaptive}.

\subsection{Six Flawed Assumptions}
\label{sec:flawed_assumptions}

\subsubsection{Intrusions are Closed Set}

\begin{figure*}[!htbp]
  \centering
  \subfloat[\label{fig:openset1}]{\includegraphics[width=0.32\textwidth]{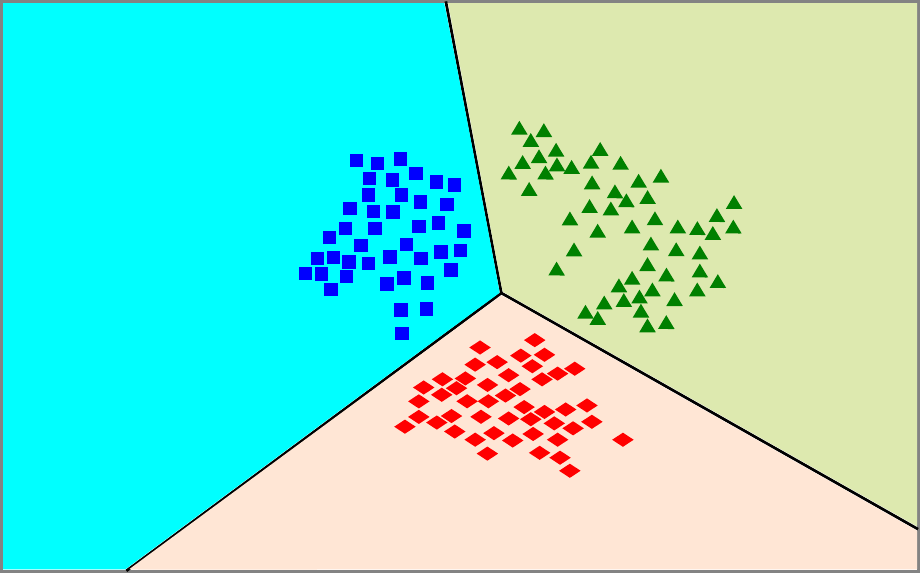}}\hspace*{.01\textwidth}
  \subfloat[\label{fig:openset2}]{\includegraphics[width=0.32\textwidth]{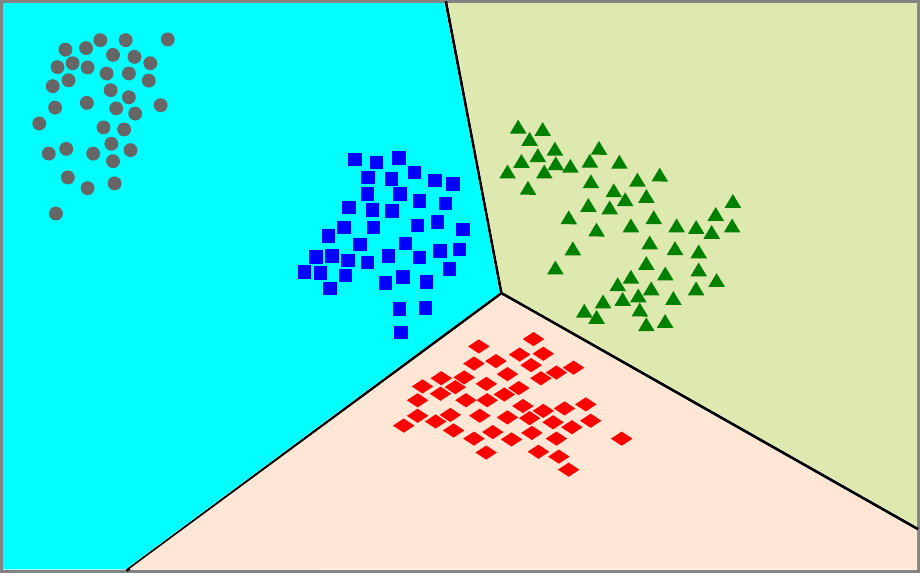}}\hspace*{.01\textwidth}
  \subfloat[\label{fig:openset3}]{\includegraphics[width=0.32\textwidth]{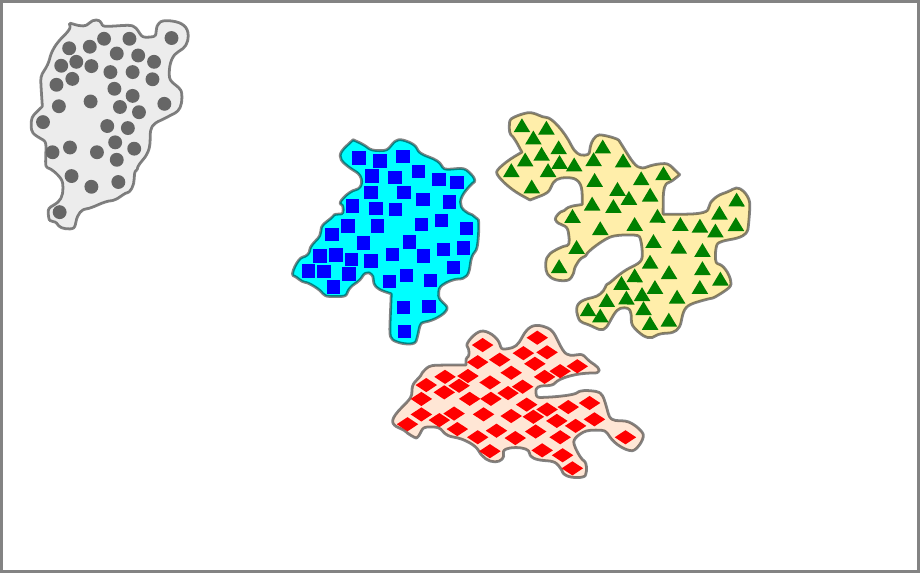}}
  \cap{fig:openset}{Problems with the closed world assumption}{\protect\subref*{fig:openset1} Red, green, and blue points correspond to a training set of different classes of malicious or benign samples in feature space.
The intersecting lines depict a decision boundary learnt from training a linear classifier on this data. \protect\subref*{fig:openset2} The classifier categorizes points from a novel class (gray) as a training class (blue) with high confidence since the gray samples lie far on the blue side of the decision boundary and the classifier labels span infinitely in feature space. \protect\subref*{fig:openset3} An idealized open world classifier bounds the amount of space ascribed to each class' label by the support of the training data, labeling unlabeled (white) space as ``unknown''. With manually or automatically supplied labels, novel classes (gray) can be added to the classifier without retraining on the vast majority of data.
}
\end{figure*}

Real intrusion recognition tasks have unseen classes at classification time. Neither all variations of malicious code nor all variations of benign behaviors can be known apriori.
However, the majority of the intrusion recognition techniques cited in this paper implicitly assume that all classes seen at classification time are also present in the training set, evaluating recognition accuracy only for a fixed \emph{closed set} of classes.


Consequently, good performance on IDS benchmarks does not necessarily translate into an effective classifier in a real application.
In real \emph{open set} scenarios, a classifier that is trained on $M$ classes of interest, at classification time is confronted with instances of classes that are sampled from a distribution of nearly infinitely-many categories.
Conventional classifiers are designed to separate classes from one another by dividing a hypothesis space into regions and assigning labels respectively.
Effective classifiers roughly seek to approximate the Bayesian optimal classifier on the posterior probability $P(y_i |x;{\cal C}_1, {\cal C}_2, \ldots, {\cal C}_M), i\in \{1,\ldots, M\}$, where $x$ is a feature vector, $y_i$ is a class label, and ${\cal C}_i$ is a particular known class.
However, in the presence of $\Omega$ unknown classes $U_n$ the optimal posterior model would become $P(y_i |x;{\cal C}_1, {\cal C}_2, \ldots, {\cal C}_M, U_{1}, \ldots, U_{\Omega})$. 
 %
Unfortunately, our ability to model this posterior distribution is limited because $U_{1}, \hdots, U_{\Omega}$ are unknown.
Mining negatives during training may help to define known but uninteresting classes (e.g., ${\cal C}_{M+1}$
), but it is impossible to span all negative space, and the costs of negative training become infeasible with increasing numbers of feature dimensions. 
Consequently, a classifier may label space belonging to ${\cal C}_i$ far beyond the support of the training data for ${\cal C}_i$.
This fundamental machine learning problem has been termed \emph{open space risk} \cite{scheirer2013towards}. Worse yet, if probability calibration is used, $x$ may be ascribed to ${\cal C}_i$ with high confidence as distance from the positive side of the decision boundary increases.
Therefore, the optimal closed set classifier operating in an open set intrusion recognition regime is not only wrong, it can be wrong while being very confident that it is correct.
An open set intrusion recognition system needs to separate $M$ known classes from one another, but must also manage open space risk by labeling a decision as ``unknown'' when far from known class data. Problems with the closed set assumption as well as desirable open set behavior are shown in Fig.~\ref{fig:openset}.

The binary intrusion recognition task, i.e., \textit{intrusion detection} appears to be a two-class closed set problem. However, each label -- \textit{intrusion} or \textit{no intrusion} -- is respectively a meta-label applied to a collection of many subclasses. While some of the subclasses will naturally be known, others will not, and the problem of open space risk still applies.


\subsubsection{Anomalies Imply Class Labels}

The incorrect assumption that anomalies imply class labels is largely related to the closed set assumption, and it is implicit to \emph{all} binary malicious/benign classification systems.
Anomalies constitute data points that deviate from statistical support of the model in question.
In the classification regime, anomalies are data points that are far from the class to which they belong.
In the open set scenario, anomalies should be resolved by an operator, protocol, or other recognition modalities.
Effective anomaly detection is necessary for open set intrusion recognition.
Without it, the implicit assumption of an overly closed set can lead to undesirable classifications because it forces a decision to be made without support of statistical evidence.
The conflation between anomaly and intrusion assumes that anomalous behavior constitutes intrusive behavior and that intrusive behavior constitutes anomalous behavior.
Often, neither of these assumptions hold.
Especially in large networks, previously unseen benign behavior is common: new software installations are routine, new users with different usage patterns come and go, new servers and switches are added, thereby changing the network topology, etc.
Stealth malwares, on the other hand, are specifically designed to exhibit normal behavior profiles and are less likely to be registered as anomalies than many previously unseen benign behaviors.

\subsubsection{Static Models are Sufficient}
\label{sec:incremental}

In the anti-malware domain, the assumption of a static model, which is implicit to the closed set modeling assumption, is particularly insufficient because of the need to update new nominal behavior profiles and malicious code signatures.
The attacks that a system sees \emph{will} change over time.
This problem is often referred to as \emph{concept drift} in the incremental learning literature \cite{masud2011classification}.
Depending on the model, the time required for a full batch retrain may not be feasible. A $k$th-order HMM with $k>>1$, for example, may perform quite well for some intrusion recognition tasks, but at the same time may be expensive to retrain in terms of both time and hardware and may require enormous amounts of training data in order to generalize well. There is a temporal risk associated with the amount of time that it takes to update a classifier to recognize new malicious samples. Therefore, even if that classifier exhibits high accuracy, it may be vulnerable to temporal risk unless it possesses an efficient update mechanism.

\subsubsection{No Feature Space Transformation is Required}
\label{sec:meaningul_features}
A key reason why machine learning algorithms are not overwhelmed by the curse of dimensionality is that, due to statistical correlations, classes of data tend to lie on manifolds that are highly non-linear, but effectively much smaller in dimension than the input space. Obtaining a good manifold representation via a feature transformation obtained from either hand-tuned or machine-learnt optimization is often critical to effective and discriminative classification. Many approaches in the intrusion detection literature simply pass raw log data or aggregated log data directly to a decision machine \cite{mukkamala2002intrusion,catania2012automatic,lee1998data,portnoy2001intrusion,lazarevic2003comparative,ertoz2004minds,rehak2009adaptive}. The inputs often possess heterogeneous scale and nominal and continuous-valued features with aggregations, which ignore temporal scale and varying spatial bandwidths. We contend that, like any other machine learning task, fine-grained discriminative intrusion recognition requires a meaningful feature space transformation, whether learnt explicitly by the classifier or carried out as a pre-processing task. Feature spaces for intrusion recognition have been explored \cite{aggarwal2007data,helmer1998intelligent,ahmad2011feature,nguyen2010improving,lakhina2010feature,middlemiss2003feature,yu2010feature,mukkamala2003feature,stein2005decision}. While this research is a good start, we believe that much additional work is needed.

\subsubsection{Model Interpretation is Optional}

Effective feature space transformations must be balanced with semantically meaningful interpretation. Unfortunately, these two objectives are sometimes conflicting.
Neural networks, which have been successfully applied to intrusion recognition tasks \cite{mukkamala2002intrusion,sung2003identifying,wang2010new,amini2006rt,liu2007letters,srinivasan2006self,shun2008network}, are appealing because they provide the ability to adapt a fixed set of basis functions to input data, thus optimizing the feature space in which the readout layer operates.
However, these basis functions correspond to a composition of aggregations of non-linear projections/liftings onto a locally optimal manifold prior to final readout, and neither the semantic meaning of the space, nor the semantic meaning of the final readout is well understood.
Recent work has demonstrated that neural networks can be vulnerable to adversarial examples \cite{nguyen2015deep,szegedy2014intriguing,goodfellow2015explaining}, which are misclassified with high confidence, yet appear very similar to known class data.
The lack of interpretability of such models means that not only could intrusion recognition systems be vulnerable to such adversarial models, but more critically, machine learning techniques are not yet ``smart'' enough to resolve most potential intrusions. 
Instead, their role is to alert specialized anti-malware modules and human operators to take swift action.
Fast response and resolution times are critical.
Even if an intrusion detection system offers nearly perfect detection performance, if it cannot provide meaningful diagnostics to the operator, a temporal risk is induced, in which the operator or anti-malware wastes valuable time trying to diagnose the problem \cite{sommer2010outside}.
Also, as we have seen from previous sections, many potential malware signals (e.g., hooking) may stem from legitimate uses. It is important to know why an alarm was triggered and which features triggered it to determine and refine the system's response to both malicious and benign behaviors.

The interpretation and temporal risk problems are not unique to intrusion detection. They are a key reason why many diagnosis and troubleshooting systems rely on directed acyclic probabilistic graphical models such as Bayesian networks as well as rule mining instead of neural networks or support vector machines (SVMs) \cite{jensen2007bayesian}. 
To better resolve the model interpretation problem, intrusion detection should move to a more generic recognition framework, ideally providing additional diagnostic information.

\subsubsection{Class Distributions are Gaussian}

The majority of probabilistic models cited in this paper assume that class distributions are single or multi-modal Gaussian mixtures in feature space. Although Gaussian mixtures often appear to capture class distributions, barring special cases, they generally fail to capture distribution tails \cite{kotz2000extreme}.

There are several different types of anomalies. \emph{Empirical anomalies} are anomalous with respect to a probabilistic model of training data, whereas \emph{idealized anomalies} are anomalous with respect to the joint distribution of training data. Provided good modeling, these two anomaly types are equivalent.
However, from an anomaly detection perspective, na\"ive Gaussian assumptions do not provide a good match between empirical and idealized anomalies because an anomaly is defined with respect to the tail of a joint distribution and tails tend to deviate from Gaussian \cite{kotz2000extreme}. Theorems from statistical extreme value theory (EVT) provide theoretically grounded functional forms for the classes of distributions that these class-tails can assume, provided that positive class outliers are \emph{process anomalies} -- rare occurrences from an underlying generating stochastic process -- and not noise exogenous to the process, e.g., previously unseen classes.

\subsection{An Open Set Recognition Framework}
\label{sec:open_set}
Accommodating new attack profiles and normative behavior models requires a method for diagnosing when query data are unsupported by previously seen training samples.
This diagnosis is commonly referred to as \emph{novelty detection} in the literature \cite{markou2003novelty}.
Specifically in IDS literature, novelty detection is often hailed as a means of detecting malicious samples with no prior knowledge.
The intuition is that by spanning the space of normal behavior during training, any novel behavior will be either an attack or a serious system error.
In practice however, it is infeasible to span the space of benign behavior.
Even on an individual host, ``normal'' benign behavior can change dramatically depending on configurations, software installations, and user turnover. The network situation is even more complicated.
Even for a medium size network, services, protocols, switches, routers, and topologies vary routinely.

We contend that novelty detection has a particularly useful role in the recognition of stealth malware, but the premise that we can span the entire benign input space apriori is as unrealistic as the premise that signatures of all known attacks solve the intrusion detection problem.
Instead, novelty detection should be treated in terms of what it does mathematically -- as a tool to recognize samples that are unsupported by the training data and to quantify the degree of confidence to ascribe a model's decision. Specifically, we propose treating the intrusion recognition task as an \emph{open set recognition} problem, performing discriminative multi-class recognition under the assumption of unknown classes at classification time.
Scheirer et al.~\cite{scheirer2013towards,scheirer2014probability} formalized the open set recognition problem as tradeoff between minimizing \emph{empirical risk} and \emph{open space risk} -- the risk of labeling unknown space -- or mathematically, the ratio of positively labeled space that should have been labeled ``unknown'' to the total extent of positively labeled space. A classifier that can arbitrarily control this ratio via an adjustable threshold is said to \textit{manage open space risk}.


Scheirer et al.~\cite{scheirer2013towards} extended the linear SVM objective to bound data points belonging to each known class by two parallel hyperplanes; one corresponding to a discriminative decision boundary, managing empirical risk, and the other limiting the extent of the classification, managing open space risk. Unfortunately, this ``slab'' model is not easily extensible to a non-linear classifier. In later work \cite{scheirer2014probability,jain2014multiclass}, they extended their solution to multi-class open set recognition problems using non-linear kernels, via posterior EVT calibration and thresholding of nonlinear SVM decision scores. EVT-calibrated one-class SVMs are used in conjunction with multi-class SVMs to simultaneously bound open-space risk and provide strong discriminative capability \cite{scheirer2014probability}. The authors refer to this combination as the \emph{W-SVM}. For our discussion, however, the theorems of Scheirer et al.~\cite{scheirer2014probability} are more interesting than the W-SVM itself. They prove that sum, product, min, and max fusions of compact abating probability (CAP) models again generate CAP models. Bendale and Boult \cite{bendale2015towards} extended this work to show that CAP models in linearly transformed spaces manage open space risk in the original input space. While these works are interesting, the formulations limit their application to probability distributions. Due to the need for efficient model updates in an intrusion recognition setting, enforcing probabilistic constraints on the recognition problem might be non-trivial, due to the need to re-normalize at each increment. We therefore generalize the theorems of Scheirer et al.~\cite{scheirer2014probability} as follows.

\begin{mythm}{{\bf Abating Bounds for Open Space Risk: }}
\label{thm:abating_bounds}
Assume a set of non-negative continuous bounded functions $\{g_1,\hdots,g_n\}$ where $g_k(x,x')$ decreases monotonically with $||x-x'||$. Then thresholding any positively weighted sum, product, min, or max fusion of a finite set of non-negative discrete or continuous functions $\{f_1,\hdots,f_n\}$ that satisfy $f_k(x,x') \leq g_k(x,x')\ \forall k$  manages open space risk, i.e., it allows us to constrain open space risk below any given $\epsilon$.
\end{mythm}
\begin{proof}

Given $\tau > 0$, define
\[ g'_k(x,x',\tau) \defeq \left\{\def\arraystretch{1.2}%
  \begin{array}{@{}c@{\quad}l@{}}
    g_k(x,x') & \text{if $g_k(x,x') > \tau$}\\
    0 & \text{otherwise.}\\
  \end{array}\right.
  \]
\[ f'_k(x,x',\tau) \defeq \left\{\def\arraystretch{1.2}%
  \begin{array}{@{}c@{\quad}l@{}}
    f_k(x,x') & \text{if $f_k(x,x') > \tau$}\\
    0 & \text{otherwise.}\\
  \end{array}\right.
  \]
This yields $f'_k(x,x') \leq g'_k(x,x')\ \forall k$.
Because of the monotonicity of $g_k$, for any fixed constant $\delta$, $\exists \tau_{\delta} \colon \int g'_k(x,x',\tau)\ dx \le \delta$.
Combining that with  $f'_k(x,x') \leq g'_k(x,x')$, yields $\int f'_k(x,x',\tau_{\delta})\ dx \le \delta$, thus limiting positively labeled area to $f_k(x,x')>\tau$, which manages open space risk.
Without loss of generality on $k$, it is easy to see that max and min fusion also manage open space risk.
Because summation is a linear operator: $$\int \sum_k f'_k(x,x',\tau)\ dx = \sum_k \int f'_k(x,x',\tau)\ dx.$$
Since a finite sum of finite values is finite, and  $\sum_k \int f'_k(x,x',\tau)\ dx \le k\delta$ it follows that thresholded positively weighted sums of $f'_k$ manage open space risk.
In addition, $\prod_k f'_k$ is bounded because $$g_k \Rightarrow \exists \eta \colon g'_k(x,x',\tau) < \eta \Rightarrow \int \prod_k g'_k(x,x',\tau) dx \le \eta^k\delta.$$
This latter bound may not be tight, but is sufficient to show that $\prod_k g'_k(x,x',\tau)$ manages open space risk.
We have proven Theorem~\ref{thm:abating_bounds} without weights in the sums and products, but without loss of generality, non-negative weights can be incorporated directly into $g_k$.
\end{proof}

From Theorem~\ref{thm:abating_bounds}, it directly follows that many \emph{novelty detection} algorithms already in use by the IDS community provably manage open space risk and fit nicely into the open set recognition framework.
For example, Scheirer et al.~\cite{scheirer2014probability} prove that thresholding neighbor methods by distance manages open space risk.
Via such thresholding, clustering methods can be extended to an online regime, in which unknown classes $U_1, \hdots, U_\Omega$ are isolated \cite{markou2003novelty}.
Similarly, thresholded kernel density estimation (KDE) of ``normal'' data distributions has been successfully applied to the IDS domain.
Yeung and Chow \cite{yeung2002parzen} used kernel density estimates, in which they centered an isotropic Gaussian kernel on every data point $x_k$.
It is easy to prove that such estimates also manage open space risk.

\begin{mycorr}{{\bf Gaussian KDE Bounds for Open Space Risk.}}
\label{cor:kde}
Assume a Gaussian kernel density estimator where: $$p(x) = \frac{1}{N}\sum_{k=1}^{N} \frac{1}{(2\pi \sigma^2)^{D/2}}exp\left( \frac{||x-x_k||^2}{2\sigma^2}\right).$$ Thresholding $p(x)$ by $0 < \tau \leq 1$ manages open space risk.
\end{mycorr}
\begin{proof}
When $N$ 
is the total number of points, each kernel is given by: $$f_k(x,x_k) = \frac{1}{N} \frac{1}{(2\pi \sigma^2)^{D/2}}exp\left(\frac{||x-x_k||^2}{2\sigma^2}\right).$$
By Theorem~\ref{thm:abating_bounds}, we can treat $f_k(x,x_k)$ as its own bound. When thresholded, $f_k(x,x_k)$ will finitely bound open space risk. The kernel density estimate: $$p(x) = \sum_{k=1}^{N}f_k(x,x_k) = \frac{1}{N}\sum_{k=1}^{N} \frac{1}{(2\pi \sigma^2)^{D/2}}exp\left( \frac{||x-x_k||^2}{2\sigma^2}\right)$$ also bounds open space risk because it is a positively weighted sum of functions that satisfy the bounding criteria in Theorem~\ref{thm:abating_bounds}.
\end{proof}

Thresholded nearest neighbor approaches and KDE require selection of a meaningful $\sigma$, and distance/probability threshold.
They also implicitly assume local isotropy in the feature space, which highlights the need for a meaningful feature space representation.

Neighbor and kernel density estimators are nonparametric models, but several parametric novelty detectors in use by the IDS community also provably manage open space risk.
Thresholding density estimates from Gaussian mixture models (GMMs) is a popular parametric approach to novelty detection with a similar functional form to KDE \cite{alizadeh2015traffic,fan2013anomaly,gruhl2015building,lam2015outlier,yamanishi2004line}.
For GMMs, however, the input data $x$ is assumed distributed as a superposition of a \textit{fixed} number of Gaussians: $$p(x) = \sum_k c_k N(x|\mu_k,\Sigma_k)$$ such that $\sum_k c_k = 1$.
Unlike nonparametric Gaussian KDE, in which $\mu$ and $\sigma$ are selected apriori, the Gaussians in a GMM are fit via an expectation maximization technique similar to that used by HMMs.
By generalizing Corollary~\ref{cor:kde}, we can prove that thresholding GMMs probabilities manages open space risk.
When GMMs integrate to one, 
they are also CAP models. Although this constraint often holds, it is not required.

\begin{mycorr}{{\bf GMM Bounds for Open Space Risk.}}
\label{cor:GMM}
Assume a Gaussian mixture model. The thresholded density estimate from this model bounds open space risk.
\end{mycorr}
\begin{proof}
By Theorem~\ref{thm:abating_bounds}, the $k$th mode of a GMM: $$f_k(x,\mu) = c_k e^{-\frac{1}{2}(x-\mu)^T(x-\mu)}$$ is its own abating bound.
Because the superposition of all modes is a sum of non-negatively weighted functions, each with an abating bound, GMMs have an abating bound.
Thresholding GMM density estimates therefore manages open space risk.
\end{proof}

Note that Corollary~\ref{cor:GMM} only holds for the density estimate from an individual GMM, and not necessarily for all  recognition functions that leverage multiple GMMs. For example, the log ratio of probabilities of two GMM estimates, $\log\frac{p_1(x)}{p_2(x)}$ does not bound open space risk when $\frac{p_1(x)}{p_2(x)}$ diverges as either $p_1(x)$ or $p_2(x) \rightarrow 0$. Similarly a recognition function $p_1(x) > p_2(x)$ does not provably manage open space risk because $p_1(x) > p_2(x)$ can hold over unbounded $x$.

There is a strong connection between GMMs and the aforementioned HMMs.
Similarly to HMMs, GMMs can also be viewed as discrete latent variable models.
Given input data $x$ and  multinomial random variable $z$, whose value corresponds to the generating Gaussian, the joint distribution factors according to the product rule: $p(x,z) = p(z)p(x|z)$.
$p(z)$ is determined by the Gaussian mixture coefficients $c_k$.
Therefore, the factorization of GMMs can be viewed as a simplification of HMMs, with a factor graph, in which latent variables are not connected and, therefore, treated independent of sequence.

This raises two questions: first, can HMMs can be used for novelty detection? And second, do HMMs manage open space risk?
Indeed, HMMs \textit{can} be used for novelty detection on sequential data by running inference on sequences in the training set and thresholding the estimated joint probability (or log of the estimated joint probability) outputs.
This approach was taken by Yeung et al.~\cite{yeung2003host} for host-based intrusion detection using system call sequences.
To assess, whether HMMs manage open space risk, we need to consider the form of an HMM's estimated joint distribution. For an $N$-length sequence, an HMM factors as: $$p(x_1,..,x_N,z_1,..,z_N) = p(z_1) \prod_{n=2}^{N} p(z_n | z_{n-1}) \prod_{n=1}^N p(x_n | z_n)$$
where $x_1,\hdots,x_N$ are observations and $z_1,\hdots,z_N$ are latent variables. This leads to Corollary~\ref{cor:HMM}.

\begin{mycorr}{{\bf HMM Bounds for Open Space Risk.}}
\label{cor:HMM}
Assume HMM factors $p(z_1)$, $p(z_n | z_{n-1})$, and $p(x_n | z_n)$ satisfy the bounding constraints in Theorem~\ref{thm:abating_bounds}. Then thresholding the output of a forward pass of an HMM bounds open space risk.
\end{mycorr}
\begin{proof}
Under the assumption that $p(z_1)$, $p(z_n | z_{n-1})$, and $p(x_n | z_n)$ satisfy the bounding constraints in Theorem~\ref{thm:abating_bounds}, then the HMM factorization above is a product of these functions, which by Theorem~\ref{thm:abating_bounds} manages open space risk.
\end{proof}

Corollary~\ref{cor:HMM} states that under certain assumptions on the form of the factors in HMMs, an HMM will provably manage open space risk. Unfortunately, it is not immediately clear how to enforce such a form, so many HMMs, including those in \cite{yeung2003host} are not proven to manage open space risk and may ascribe known labels to infinite open space. Formulating HMMs that manage open space risk and provide adequate modeling of data is an important topic, which we leave for future research.

GMMs and HMMs are linear models. One-class SVMs are popular nonlinear models, which have been successfully applied to detecting novel intrusions \cite{amer2013enhancing,yang2015adaptive,heller2003one,wang2004anomaly,li2003improving,perdisci2006using}. 
In their Theorem 2, Scheirer et al.~\cite{scheirer2014probability} prove that one-class SVM density estimators \cite{scholkopf2001estimating} with a Gaussian radial-basis function (RBF) kernel manage open space risk. The decision functions for these machines are given by  $\sum_k \alpha_k K(x,x_k)$, where $K(x,x_k)$ is the kernel function and $\alpha_i$ are the Lagrange multipliers. 
It is important to note that non-negative $\alpha_k$ are required to satisfy Theorem 1 in \cite{scheirer2014probability}, and that multi-class RBF SVMs and one-class SVMs under different objective functions are not proven to manage open space risk.

\subsection{Open World Archetypes for Stealth Malware Intrusion Recognition}
\label{sec:adaptive}

The open set recognition framework introduced in Sec.~\ref{sec:open_set} can be \emph{incorporated} into existing intrusion recognition algorithms.
This means that there is no need to abandon closed set algorithms in order to manage open space risk, provided that they are fused with open set recognition algorithms.
Closed set techniques may be excellent solutions when they are well supported by training data, but open set algorithms are required in order to ascertain whether the closed set decisions are meaningful.
Therefore, the open set problem can be addressed by using an algorithm that is inherently open set for novelty detection and rejecting any closed set decision as unknown if its support is below the open set threshold.
A model with easily interpreted diagnostic information, e.g., a decision tree or Bayesian network, can be fused with the open set algorithm as well, in order to decrease response/mitigation times and to compensate for other discriminative algorithms that are not so readily interpretable.
Note that many of the algorithms proposed by Scheirer et al. are discriminative classifiers themselves, but underperform the state of the art in a purely closed setting.

The interpretation of a thresholded open set decision is trivial, assuming that the recognition function represents some sort of density estimation. For a query point, if the maximum density with respect to each class is below the open set threshold, $\tau$, then the class is labeled as ``unknown''. Otherwise, the query sample is ascribed the label corresponding to the class of maximum density.
Under the open set formulation, the degree of \emph{openness} can be controlled by the value of $\tau$.
The desired amount of openness will vary depending on the algorithm and the application's optimal precision/recall requirements.
For example, a high security non-latency sensitive virtualized environment that is administered by many security experts can label many examples as unknown and interpose state frequently for an expert opinion.
Systems that are latency sensitive, but for which potential harm of intrusion is relatively low, might have much looser open space bounds.

Note that an open set density estimator can be applied with or without normalization to a probability distribution.
However, we can only prove that it manages open space risk if the estimator's decision function satisfies Theorem~\ref{thm:abating_bounds}.

Open set algorithms can also be applied under many different feature space transformations.
When open set algorithms are fused with closed set algorithms, the two need not necessarily operate in the same feature space.
Research has demonstrated\cite{rudd2015extreme,bendale2015towards} the effectiveness of the open set classification framework in machine-learnt feature spaces.
Bendale and Boult \cite{bendale2015towards} bounded a nearest class mean (NCM) classifier in a metric-learnt transformed space, an algorithm they dubbed nearest non-outlier (NNO).
They also proved that under a linear transformation, open space risk management in the transformed feature space will manage open space risk in the original input space.
Rudd et al.~\cite{rudd2015extreme} formulated extreme value machine (EVM) classifiers to perform open set classification in a feature space learnt from a convolutional neural network.
The EVM formulation performs a kernel-like transformation, which supports variable data bandwidths, that implicitly transforms \emph{any} feature space to a probabilistically meaningful representation.
This research indicates that open set algorithms can support meaningful feature space transformations, although what constitutes a ``good'' feature space depends on the problem and classifier in question.

Bendale and Boult\cite{bendale2015towards} and Rudd et al.~\cite{rudd2015extreme} also extended open set recognition algorithms to an online regime, which supports incremental model updates.
They dubbed this recognition scenario \emph{open world recognition}, referring to online open set recognition.
The incremental aspects of this work are in a similar vein to other online intrusion recognition techniques \cite{lane1998approaches,karthick2012adaptive,wang2004anomalous,zhong2007clustering,cannady2000applying,hu2014online,wang2013concept},
which, given a batch of training points $X_{t}$ at time $t$, aim to update the prior for time $t+1$ in terms of the posterior for time $t$, so that $P_{t+1}(\theta_{t+1}) \leftarrow  P_{t}(\theta_{t}|X_{t},T_{t})$, where $T$ is the target variable, $P$ is a recognition function, and $\theta$ is a parameter vector. If $P$ is a probability, a Bayesian treatment can be adopted, where:
{\small
$$
P_{t + 1}(\theta_{t + 1}|X_{t + 1},T_{t + 1}) = \frac{P_{t+1}(T_{t+1}|\theta_{t+1},X_{t+1})P_{t}(\theta_{t}|X_{t},T_{t})}{P_{t+1}(T_{t+1})}
$$%
}%
With a few exceptions, however, recognition functions in the incremental learning intrusion recognition literature generally do not satisfy Theorem~\ref{thm:abating_bounds}, and are not proven to manage open space risk. This means that they are not necessarily true \textit{open world} classifiers.

Moreover, none of the work in \cite{lane1998approaches,karthick2012adaptive,wang2004anomalous,zhong2007clustering,cannady2000applying,hu2014online,wang2013concept} addresses the pressing need to \emph{prioritize labeling of detected novel} data for incremental training.
This is problematic, because the objective of online learning is to adapt a model to recognize new attack variations and benign patterns -- insights that would otherwise be perishable within a useful time horizon.
When intrusion recognition subsystems exhibit high recall rates, however, updating the model with new attack signatures is \emph{much more vital} than updating the model with novel benign profiles.
Since labeling capacity is often limited by the number of knowledgeable security practitioners, we contend that the ``optimal'' labeling approach is to greedily rank the unknown samples in terms of their likelihood of being associated with known malicious classes. Given bounded radially abating functions from Theorem~\ref{thm:abating_bounds}, i.e., open set decision functions, we can do just that, prioritizing labeling by some malicious likelihood criterion (MLC).

The intuition behind the MLC ranking is as follows: from the discussion in Sec.~\ref{sec:signatures_heuristics}, malwares often share components: even for vastly different malwares, similar components yield similar patterns in feature space. Although minor code overlap will not necessarily cause (mis)categorizations of malware classes, it may cause novel malware classes to be close enough to known ones in feature space that they are ranked higher by MLC criterion than most novel benign samples. Label prioritization by  MLC ranking could, therefore, improve resource allocation of security professionals and  dramatically reduce the amount of time that stealth malwares are able to propagate unnoticed. Of course, other considerations besides MLC are relevant to a truly ``optimal'' ranking, including difficulty of diagnosis and likely degree of harm, but these properties are difficult to ascertain autonomously.

A final useful aspect of the open world intrusion recognition framework is that it is not confined to naive Gaussian assumptions.
Mixtures of Gaussians can work well for modeling densities, but tend to deteriorate at the distribution tails, because the tails of the models tend toward tails from unimodal Gaussians, whereas the tails of the data distributions generally do not.
For recognition problems, however, accurate modeling of tail behavior is important, in fact, more important than accurate modeling of class centers \cite{scheirer2010robust,scheirer2012multi}.
To this end, researchers have turned to statistical extreme value theory techniques for density estimation, and open world recognition readily accommodates them.
Both \cite{jain2014multiclass} and \cite{scheirer2014probability} apply EVT modeling to open set recognition scenarios based posterior fitting of point distances to classifier decision boundaries, while \cite{rudd2015extreme} incorporated EVT calibration into a generative model, which performs loose density estimation as a mixture of EVT distributions.
Importantly, the EVT distributions employed by Rudd et al.~\cite{rudd2015extreme}, unlike Gaussian kernels in an SVM or KDE application, are variable bandwidth functions of the data. They are also are directly derived from EVT and incorporate higher-order statistics which Gaussian distributions cannot (e.g., skew, curtosis). Finally, they provably manage open space risk.

\ifCLASSOPTIONcaptionsoff
  \newpage
\fi

\section{Conclusions and Open Issues}
\label{sec:conclusion}

Stealth malwares are a growing threat because they exploit many system features that have legitimate uses and they can propagate undetected for long periods of time.
We, therefore, felt the need to provide the first academic survey specifically focused on malicious stealth technologies and mitigation measures.
We hope that security professionals in both academic and industrial environments can draw on this work in their research and development efforts.
Our work also highlights the need to combine countermeasures that aim to protect the integrity of system components with more generic machine learning solutions.
We have identified flawed assumptions behind many machine learning algorithms and proposed steps to improve them based on research from other recognition domains.
We encourage the security community to consider these suggestions in future development of intrusion recognition algorithms.

While we are the first to propose a mathematically formalized \textit{open world} approach to intrusion recognition, there are open issues that must be addressed through experimentation and implementation, including how tightly to bound open space risk, and more generally how to determine the openness of the problem in operational scenarios. An overly aggressive bound may actually degrade performance for problems that are predominantly closed-set, prioritizing the minimization of open space risk over the minimization of empirical risk. Another important consideration is the cost of misclassifying an unknown sample as belonging to a known class, which depends in part on the operational resources available to label novel classes, and in part on the degree of threat expected of novel classes. These tradeoffs are important subjects for experimental analysis and operational deployment of open world anti-malware systems.

For benchmarking and experimentation, good datasets that support open world protocols are vital for future research. While some effort has been made, e.g., \cite{creech2014semantic,creech2013generation,creech2014developing}, there are few modern publicly available datasets for intrusion detection, specifically of stealth malwares. We believe that the collection and distribution of modernized and realistic publicly available datasets containing stealth malware samples are vital to the furtherance of academic research in the field. While many corporate security companies have good reasons for keeping their datasets private, a guarded increase in collusion with academia to allow extended -- yet still restricted -- sharing of data is in the best interest of all parties in developing better stealth malware countermeasures.

We have proven that a number of existing algorithms currently used in the intrusion recognition domain already satisfy the requirements of an open set framework, and we believe that they should be leveraged and extended both in theory and in practice to address the flawed assumptions behind many existing algorithms that we detailed in Sec.~\ref{sec:flawed_assumptions}.
Adopting an open world mathematical framework obviates the assumptions that \textit{intrusions are closed set}, \textit{anomalies imply class labels}, and that \textit{static models are sufficient}.
How to appropriately address the other assumptions requires further research.
Although some progress has been made in open world algorithms, the question, how to obtain a nicely discriminable feature space while accommodating a readily interpretable model, merits future research.
Finally, how to model class distributions without Gaussian assumptions demands further mathematical treatment -- statistical extreme value theory is a good start, but it has yet to be gracefully defined how select distributional tail boundaries.
Also, with the exception of special cases it is still not well formalized, how to model the remainder of the distribution (the non-extreme values) of non-Gaussian data.

\begin{table*}[htp]
  \scriptsize
  \centering
  \section*{Appendix: API Calls, Data Structures, and Registry Keys}
  \vspace{1em}
  \renewcommand\arraystretch{1.5}
  \setlength\tabcolsep{1em}
  \begin{tabularx}{.9\textwidth}{c|X}
    \bf \Name{Windows} API entry  & \bf Documentation \\\hline
    \Windows{CallNextHookEx} & Passes the hook information to the next hook procedure in the current hook chain. A hook procedure can call this function either before or after processing the hook information.\\
    \Windows{CreateRemoteThread} & Creates a thread that runs in the virtual address space of another process.\\
    \Windows{DeviceIoControl} & Sends a control code directly to a specified device driver, causing the corresponding device to perform the corresponding operation.\\
    \Windows{DllMain} & An optional entry point into a dynamic-link library (DLL). When the system starts or terminates a process or thread, it calls the entry-point function for each loaded DLL using the first thread of the process. The system also calls the entry-point function for a DLL when it is loaded or unloaded using the \Windows{LoadLibrary} and \Windows{FreeLibrary} functions.\\
    \Windows{FindFirstFile} & Searches a directory for a file or subdirectory with a name that matches a specific name (or partial name if wildcards are used).\\
    \Windows{FindNextFile} & Continues a file search from a previous call to the \Windows{FindFirstFile}, \Windows{FindFirstFileEx}, or \Windows{FindFirstFileTransacted} functions.\\
     \Windows{GetProcAddress} & Retrieves the address of an exported function or variable from the specified dynamic-link library (DLL).\\
     \Windows{LoadLibrary} & Loads the specified module into the address space of the calling process. The specified module may cause other modules to be loaded.\\
     \Windows{NtDelayExecution} & \emph{Undocumented export of \Windows{ntdll.dll}.}\\
     \Windows{NTQuerySystemInformation} & Retrieves the specified system information.\\
     \Windows{OpenProcess} & Opens an existing local process object.\\
     \Windows{PsSetLoadImageNotifyRoutine} & The \Windows{PsSetLoadImageNotifyRoutine} routine registers a driver-supplied callback that is subsequently notified whenever an image is loaded (or mapped into memory).\\
     \Windows{PspCreateProcessNotify} & \emph{Completely undocumented function.}\\
     \Windows{SetWindowsHookEx} & Installs an application-defined hook procedure into a hook chain. You would install a hook procedure to monitor the system for certain types of events. These events are associated either with a specific thread or with all threads in the same desktop as the calling thread.\\
     \Windows{SleepEx} & Suspends the current thread until the specified condition is met.\\
     \Windows{WriteProcessMemory} & Writes data to an area of memory in a specified process. The entire area to be written to must be accessible or the operation fails.\\
     \Windows{ZwQuerySystemInformation} & Retrieves the specified system information.\\\hline
     \Windows{EAT} & The \Windows{Export Address Table} is a table where functions exported by a module are placed so that they can be used by other modules.\\ 
     \Windows{EPROCESS} & The \Windows{EPROCESS} structure is an opaque executive-layer structure that serves as the process object for a process.\\
     \Windows{IAT} & The \Windows{Import Address Table} is where the dynamic linker writes addresses of loaded modules such that each entry points to the memory locations of library functions.\\
     \Windows{IDT} & The \Windows{Interrupt Descriptor Table} is a kernel-level table of function pointers to callbacks that are called upon interrupts / exceptions.\\
     \Windows{IRP} & \Windows{I/O Request Packets} are used to communicate between device drivers and other areas of the kernel.\\
     \Windows{KTHREAD} & The \Windows{KTHREAD} structure is an opaque kernel-layer structure that serves as the thread object for a thread.\\
     \Windows{KeServiceDescriptorTable} & Contains pointers to the \Windows{SSDT} and \Windows{SSPT}. It is an undocumented export of \Windows{ntoskrnl.exe}.\\
     \Windows{SSDT} & The \Windows{System Service Descriptor Table} is a kernel-level dispatch table of callbacks for system calls.\\
     \Windows{SSPT} & The \Windows{System Service Parameter Table} is a kernel-level table containing sizes (in bytes) of arguments for \Windows{SSDT} callbacks.\\\hline
     \Windows{AppInit\_DLLs} & Space or comma delimited list of DLLs to load.\\
     \Windows{LoadAppInit\_DLLs} & Globally enables or disables \Windows{AppInit\_DLLs}.\\\hline
     \Windows{kernel32.dll} & Exposes to applications most of the \Name{Win32} base APIs, such as memory management, input/output (I/O) operations, process and thread creation, and synchronization functions.\\
     \Windows{ntdll.dll} & Exports the \Name{Windows Native API}. The \Name{Native API} is the interface used by user-mode components of the operating system that must run without support from \Name{Win32} or other API subsystems.\\
     \Windows{ntoskrnl.exe} & Provides the kernel and executive layers of the \Name{Windows NT} kernel space, and is responsible for various system services such as hardware virtualization, process and memory management, thus making it a fundamental part of the system.\\
     \Windows{user32.dll} & Implements the \Name{Windows} user component that creates and manipulates the standard elements of the Windows user interface, such as the desktop, windows, and menus.
  \end{tabularx}
  \cap{tab:SysCalls}{System Calls}{This table explains the \Name{Windows} system calls, data structures, registry keys and system files (in this order) that are used in the malware described in this section, in alphabetical order. Many of the entries are copied directly from the \Name{Microsoft Developer Network} (MSDN) documentation\cite{MSDN} or \Name{Wikipedia} for file descriptions\cite{Wikipedia}. Others are summaries of descriptions from later in the text with their own respective citations.}
\end{table*}
\newpage \clearpage

\bibliographystyle{ieee}
\bibliography{ids_survey}

\begin{thebibliography}{100}
\providecommand{\url}[1]{#1}
\csname url@samestyle\endcsname
\providecommand{\newblock}{\relax}
\providecommand{\bibinfo}[2]{#2}
\providecommand{\BIBentrySTDinterwordspacing}{\spaceskip=0pt\relax}
\providecommand{\BIBentryALTinterwordstretchfactor}{4}
\providecommand{\BIBentryALTinterwordspacing}{\spaceskip=\fontdimen2\font plus
\BIBentryALTinterwordstretchfactor\fontdimen3\font minus
  \fontdimen4\font\relax}
\providecommand{\BIBforeignlanguage}[2]{{%
\expandafter\ifx\csname l@#1\endcsname\relax
\typeout{** WARNING: IEEEtran.bst: No hyphenation pattern has been}%
\typeout{** loaded for the language `#1'. Using the pattern for}%
\typeout{** the default language instead.}%
\else
\language=\csname l@#1\endcsname
\fi
#2}}
\providecommand{\BIBdecl}{\relax}
\BIBdecl

\bibitem{sommer2010outside}
R.~Sommer and V.~Paxson, ``{Outside the Closed World: On Using Machine Learning
  for Network Intrusion Detection},'' in \emph{Proceedings of the 2010 IEEE
  Symposium on Security and Privacy}.\hskip 1em plus 0.5em minus 0.4em\relax
  IEEE Computer Society, 2010, pp. 305--316.

\bibitem{li2011survey}
W.-X. Li, J.-B. Wang, D.-J. Mu, and Y.~Yuan, ``{Survey on Android Rootkit},''
  \emph{CNKI - Microprocessors}, vol.~2, 2011.

\bibitem{kim2012brief}
S.~Kim, J.~Park, K.~Lee, I.~You, and K.~Yim, ``{A Brief Survey on Rootkit
  Techniques in Malicious Codes},'' \emph{Journal of Internet Services and
  Information Security}, vol.~2, no. 3/4, pp. 134--147, 2012.

\bibitem{shields2008survey}
\BIBentryALTinterwordspacing
T.~Shields, ``{Survey of Rootkit Technologies and Their Impact on Digital
  Forensics},'' 2008. [Online]. Available:
  \url{http://www.donkeyonawaffle.org/misc/txs-rootkits_and_digital_forensics.pdf}
\BIBentrySTDinterwordspacing

\bibitem{axelsson2000intrusion}
S.~Axelsson, ``{Intrusion Detection Systems: A Survey and Taxonomy},''
  Technical report Chalmers University of Technology, Goteborg, Sweden, Tech.
  Rep., 2000.

\bibitem{vasilomanolakis2015taxonomy}
E.~Vasilomanolakis, S.~Karuppayah, M.~M\"{u}hlh\"{a}user, and M.~Fischer,
  ``{Taxonomy and Survey of Collaborative Intrusion Detection},'' \emph{ACM
  Computing Surveys}, vol.~47, no.~4, pp. 55:1--55:33, 2015.

\bibitem{zuech2015intrusion}
R.~Zuech, T.~M. Khoshgoftaar, and R.~Wald, ``{Intrusion detection and Big
  Heterogeneous Data: a Survey},'' \emph{Journal of Big Data}, vol.~2, no.~1,
  pp. 1--41, 2015.

\bibitem{tsai2009intrusion}
C.-F. Tsai, Y.-F. Hsu, C.-Y. Lin, and W.-Y. Lin, ``{Intrusion Detection by
  Machine Learning: A Review},'' \emph{Expert Systems with Applications},
  vol.~36, no.~10, pp. 11\,994--12\,000, 2009.

\bibitem{garcia2009anomaly}
P.~Garcia-Teodoro, J.~Diaz-Verdejo, G.~Maci{\'a}-Fern{\'a}ndez, and
  E.~V{\'a}zquez, ``{Anomaly-Based Network Intrusion Detection: Techniques,
  Systems and Challenges},'' \emph{Computers \& Security}, vol.~28, no. 1-2,
  pp. 18--28, 2009.

\bibitem{lee1999data}
W.~Lee, S.~J. Stolfo, and K.~W. Mok, ``{A Data Mining Framework for Building
  Intrusion Detection Models},'' in \emph{Proceedings of the 1999 IEEE
  Symposium on Security and Privacy}.\hskip 1em plus 0.5em minus 0.4em\relax
  IEEE, 1999, pp. 120--132.

\bibitem{kaspersky2015}
K.~Lab, ``{Kaspersky Security Bulletin},'' Kaspersky Lab, Tech. Rep., 2015.

\bibitem{hpe2015}
H.~P. Enterprise, ``{HPE Security Research Cyber Risk Report},'' Hewlett
  Packard Enterprise, Tech. Rep., 2015.

\bibitem{hpe2016}
------, ``{HPE Security Research Cyber Risk Report},'' Hewlett Packard
  Enterprise, Tech. Rep., 2016.

\bibitem{ibm2016}
{IBM X-Force}, ``{IBM X-Force Threat Intelligence Report 2016},'' IBM, Tech.
  Rep., 2016.

\bibitem{symantecistr2015}
Symantec, ``{Internet Security Threat Report},'' Semantec, Tech. Rep., 04 2015.

\bibitem{symantecistr2016}
------, ``{Internet Security Threat Report},'' Symantec, Tech. Rep., 04 2016.

\bibitem{mcaffee2016}
{McAfee Labs}, ``{McAfee Labs Threats Report},'' Intel Security, Tech. Rep., 03
  2016.

\bibitem{microsoft_security}
Microsoft, ``{Microsoft Security Intelligence Report},'' Microsoft, Tech. Rep.,
  12 2015.

\bibitem{mandiant_consulting}
{Mandiant Consulting}, ``{M-Trends},'' Mandiant Consulting, Tech. Rep., 02
  2016.

\bibitem{faruki2015android}
P.~Faruki, A.~Bharmal, V.~Laxmi, V.~Ganmoor, M.~S. Gaur, M.~Conti, and
  M.~Rajarajan, ``{Android Security: A Survey of Issues, Malware Penetration,
  and Defenses},'' \emph{IEEE Communications Surveys \& Tutorials}, vol.~17,
  no.~2, pp. 998--1022, 2015.

\bibitem{chuvakin2003ups}
\BIBentryALTinterwordspacing
A.~Chuvakin, ``{Ups and Downs of {UNIX}/Linux Host-Based Security Solutions},''
  \emph{{;login:} The magazine of USENIX \& SAGE}, vol.~28, no.~2, 2003.
  [Online]. Available:
  \url{https://www.usenix.org/publications/login/april-2003-volume-28-number-2/ups-and-downs-unixlinux-host-based-security}
\BIBentrySTDinterwordspacing

\bibitem{petroni2004copilot}
\BIBentryALTinterwordspacing
N.~L. Petroni~Jr., T.~Fraser, J.~Molina, and W.~A. Arbaugh, ``{Copilot - A
  Coprocessor-Based Kernel Runtime Integrity Monitor},'' in \emph{Proceedings
  of the 13th Conference on USENIX Security Symposium}, ser. SSYM'04, 2004.
  [Online]. Available:
  \url{http://www.jesusmolina.com/publications/2004NPTF.pdf}
\BIBentrySTDinterwordspacing

\bibitem{butler2005windows1}
\BIBentryALTinterwordspacing
J.~Butler and S.~Sparks, ``{Windows Rootkits of 2005, Part One},'' Symantec
  Connect Community, Tech. Rep., 2005. [Online]. Available:
  \url{http://www.symantec.com/connect/articles/windows-rootkits-2005-part-one}
\BIBentrySTDinterwordspacing

\bibitem{kim1994tripwire}
G.~H. Kim and E.~H. Spafford, ``{The Design and Implementation of Tripwire: A
  File System Integrity Checker},'' in \emph{Proceedings of the 2nd ACM
  Conference on Computer and Communications Security}.\hskip 1em plus 0.5em
  minus 0.4em\relax ACM, 1994, pp. 18--29.

\bibitem{hoglund2005rootkits}
G.~Hoglund and J.~Butler, \emph{Rootkits: Subverting the Windows Kernel}.\hskip
  1em plus 0.5em minus 0.4em\relax Addison-Wesley Professional, 2005.

\bibitem{szor2005theart}
P.~Szor, \emph{The Art of Computer Virus Research and Defense}.\hskip 1em plus
  0.5em minus 0.4em\relax Addison-Wesley Professional, 2005.

\bibitem{Gaptz}
\BIBentryALTinterwordspacing
E.~Rodionov and A.~Matrosov, ``{Mind the Gapz: The Most Complex Rootkit Ever
  Analyzed?}'' [Online]. Available:
  \url{http://www.welivesecurity.com/wp-content/uploads/2013/04/gapz-bootkit-whitepaper.pdf}
\BIBentrySTDinterwordspacing

\bibitem{Olmasco}
\BIBentryALTinterwordspacing
A.~Matrosov, ``{Olmasco Bootkit: Next Circle of {TDL4} Evolution (or Not?)},''
  Oct 2012. [Online]. Available:
  \url{http://www.welivesecurity.com/2012/10/18/olmasco-bootkit-next-circle-of-tdl4-evolution-or-not-2/}
\BIBentrySTDinterwordspacing

\bibitem{butler2004vice}
\BIBentryALTinterwordspacing
J.~Butler and G.~Hoglund, ``{{VICE} -- Catch the Hookers!}'' 2004, presented at
  Black Hat {USA}. [Online]. Available:
  \url{http://www.infosecinstitute.com/blog/butler.pdf}
\BIBentrySTDinterwordspacing

\bibitem{microsoft2007what}
\BIBentryALTinterwordspacing
{Microsoft Support}, ``What is a {DLL}?'' 2007, article ID: 815065. [Online].
  Available: \url{http://support.microsoft.com/kb/815065}
\BIBentrySTDinterwordspacing

\bibitem{leitch2011iat}
J.~Leitch, ``{IAT Hooking Revisited},'' 2011.

\bibitem{x64_introduction}
\BIBentryALTinterwordspacing
``{Introduction to x64 {Assembly} {Intel} {Software}}.'' [Online]. Available:
  \url{https://software.intel.com/en-us/articles/introduction-to-x64-assembly}
\BIBentrySTDinterwordspacing

\bibitem{microsoft1999detours}
\BIBentryALTinterwordspacing
{Microsoft Research}, ``Detours,'' 1999. [Online]. Available:
  \url{http://research.microsoft.com/en-us/projects/detours}
\BIBentrySTDinterwordspacing

\bibitem{hunt1999detours}
G.~Hunt and D.~Brubacher, ``{Detours: Binary Interception of {Win32}
  Functions},'' in \emph{Proceedings of the 3rd Conference on USENIX Windows NT
  Symposium}, vol.~3.\hskip 1em plus 0.5em minus 0.4em\relax USENIX
  Association, 1999.

\bibitem{protean2013api}
\BIBentryALTinterwordspacing
{Protean Security}, ``{{API} Hooking and {DLL} Injection on {W}indows},'' 2013.
  [Online]. Available:
  \url{http://resources.infosecinstitute.com/api-hooking-and-dll-injection-on-windows}
\BIBentrySTDinterwordspacing

\bibitem{protean2013hookex}
\BIBentryALTinterwordspacing
------, ``{Using {SetWindowsHookEx} for {DLL} Injection on {W}indows},'' 2013.
  [Online]. Available:
  \url{http://resources.infosecinstitute.com/using-setwindowshookex-for-dll-injection-on-windows}
\BIBentrySTDinterwordspacing

\bibitem{protean2013remotethread}
\BIBentryALTinterwordspacing
------, ``{Using {CreateRemoteThread} for {DLL} Injection on {W}indows},''
  2013. [Online]. Available:
  \url{http://resources.infosecinstitute.com/using-createremotethread-for-dll-injection-on-windows}
\BIBentrySTDinterwordspacing

\bibitem{msdnSet}
\BIBentryALTinterwordspacing
{Microsoft Developer Network}, ``{SetWindowsHookEx} function ({W}indows).''
  [Online]. Available:
  \url{http://msdn.microsoft.com/en-us/library/windows/desktop/ms644990(v=vs.85).aspx}
\BIBentrySTDinterwordspacing

\bibitem{msdnHooks}
\BIBentryALTinterwordspacing
------, ``{Hooks Overview ({W}indows)}.'' [Online]. Available:
  \url{http://msdn.microsoft.com/en-us/library/windows/desktop/ms644959(v=vs.85).aspx}
\BIBentrySTDinterwordspacing

\bibitem{TDL4-reboot}
\BIBentryALTinterwordspacing
D.~Harley, ``{{TDL4} Rebooted},'' Oct 2011. [Online]. Available:
  \url{http://www.welivesecurity.com/2011/10/18/tdl4-rebooted/}
\BIBentrySTDinterwordspacing

\bibitem{sonydrm}
\BIBentryALTinterwordspacing
``{Sony's {DRM} {R}ootkit: The Real Story},'' 2005. [Online]. Available:
  \url{http://www.schneier.com/blog/archives/2005/11/sonys\_drm\_rootk.html.}
\BIBentrySTDinterwordspacing

\bibitem{necurs}
\BIBentryALTinterwordspacing
P.~Ferrie, ``{The curse of {N}ecurs - {P}art 1}.'' [Online]. Available:
  \url{http://www.virusbtn.com/pdf/magazine/2014/201404.pdf}
\BIBentrySTDinterwordspacing

\bibitem{rtf-necurs}
\BIBentryALTinterwordspacing
H.~Li, S.~Zhu, and J.~Xie, ``{{RTF} Attack Takes Advantage of Multiple
  Exploits}.'' [Online]. Available:
  \url{https://blogs.mcafee.com/mcafee-labs/rtf-attack-takes-advantage-of-multiple-exploits/}
\BIBentrySTDinterwordspacing

\bibitem{ssdt_hooking}
\BIBentryALTinterwordspacing
``Hooking the {System} {Service} {Dispatch} {Table} ({SSDT}) - {InfoSec}
  {Resources}.'' [Online]. Available:
  \url{http://resources.infosecinstitute.com/hooking-system-service-dispatch-table-ssdt/}
\BIBentrySTDinterwordspacing

\bibitem{idt_hooking}
\BIBentryALTinterwordspacing
``Hooking {IDT} - {InfoSec} {Resources}.'' [Online]. Available:
  \url{http://resources.infosecinstitute.com/hooking-idt/}
\BIBentrySTDinterwordspacing

\bibitem{tanenbaum2007modern}
A.~S. Tanenbaum, \emph{Modern Operating Systems}, 3rd~ed.\hskip 1em plus 0.5em
  minus 0.4em\relax Prentice Hall Press, 2007.

\bibitem{msdnPs}
\BIBentryALTinterwordspacing
{Microsoft Developer Network}, ``{{PsSetLoadImageNotifyRoutine} Routine
  ({W}indows Drivers)}.'' [Online]. Available:
  \url{http://msdn.microsoft.com/en-us/library/windows/hardware/ff559957(v=vs.85).aspx}
\BIBentrySTDinterwordspacing

\bibitem{jack2005step}
B.~Jack, ``{Remote {W}indows Kernel Exploitation: Step into Ring 0},'' {eEye}
  Digital Security, Tech. Rep., 2005.

\bibitem{srivastava2011operating}
A.~Srivastava, A.~Lanzi, J.~Giffin, and D.~Balzarotti, ``{Operating System
  Interface Obfuscation and the Revealing of Hidden Operations},'' in
  \emph{Proceedings of the 8th International Conference on Detection of
  Intrusions and Malware, and Vulnerability Assessment}.\hskip 1em plus 0.5em
  minus 0.4em\relax Springer, 2011, pp. 214--233.

\bibitem{beck2005detecting}
D.~Beck, B.~Vo, and C.~Verbowski, ``{Detecting Stealth Software with {Strider}
  {Ghostbuster}},'' in \emph{Proceedings of the 2005 International Conference
  on Dependable Systems and Networks}.\hskip 1em plus 0.5em minus 0.4em\relax
  IEEE Computer Society, 2005.

\bibitem{baliga2011data}
A.~Baliga, V.~Ganapathy, and L.~Iftode, ``{Detecting Kernel-Level Rootkits
  Using Data Structure Invariants},'' \emph{IEEE Transactions on Dependable and
  Secure Computing}, vol.~8, no.~5, pp. 670--684, 2011.

\bibitem{seifried2008fourth}
\BIBentryALTinterwordspacing
K.~Seifried, ``{Fourth-Generation Rootkits},'' \emph{Linux Magazine}, no.~97,
  2008. [Online]. Available:
  \url{http://www.linux-magazine.com/Issues/2008/97/Security-Lessons}
\BIBentrySTDinterwordspacing

\bibitem{szor2001hunting}
P.~Sz{\"o}r and P.~Ferrie, ``{Hunting for Metamorphic},'' in \emph{Proceedings
  of the 2001 Virus Bulletin Conference}, 2001, pp. 123--144.

\bibitem{beaucamps2007advanced}
P.~Beaucamps, ``{Advanced Polymorphic Techniques},'' \emph{International
  Journal of Computer Science}, vol.~2, no.~3, pp. 194--205, 2007.

\bibitem{sridhara2013metamorphic}
S.~Madenur~Sridhara and M.~Stamp, ``Metamorphic worm that carries its own
  morphing engine,'' \emph{Journal of Computer Virology and Hacking
  Techniques}, vol.~9, no.~2, pp. 49--58, 2013.

\bibitem{filiol2007metamorphism}
{\'E}.~Filiol, ``{Metamorphism, Formal Grammars and Undecidable Code
  Mutation},'' \emph{International Journal in Computer Science}, vol.~2, no.~1,
  pp. 70--75, 2007.

\bibitem{zbitskiy2009code}
P.~V. Zbitskiy, ``{Code Mutation Techniques by Means of Formal Grammars and
  Automatons},'' \emph{Journal in Computer Virology}, vol.~5, no.~3, pp.
  199--207, 2009.

\bibitem{faruki2014evaluation}
P.~Faruki, A.~Bharmal, V.~Laxmi, M.~S. Gaur, M.~Conti, and M.~Rajarajan,
  ``{Evaluation of Android Anti-Malware Techniques Against Dalvik Bytecode
  Obfuscation},'' in \emph{2014 IEEE 13th International Conference on Trust,
  Security and Privacy in Computing and Communications}.\hskip 1em plus 0.5em
  minus 0.4em\relax IEEE, 2014, pp. 414--421.

\bibitem{faruki2016droidanalyst}
P.~Faruki, S.~Bhandari, V.~Laxmi, M.~Gaur, and M.~Conti, ``{DroidAnalyst:
  Synergic App Framework for Static and Dynamic App Analysis},'' in
  \emph{Recent Advances in Computational Intelligence in Defense and
  Security}.\hskip 1em plus 0.5em minus 0.4em\relax Springer, 2016, pp.
  519--552.

\bibitem{bencsath2012cousins}
B.~Bencs{\'a}th, G.~P{\'e}k, L.~Butty{\'a}n, and M.~F{\'e}legyh{\'a}zi, ``{The
  Cousins of {Stuxnet}: {Duqu}, {Flame}, and {Gauss}},'' \emph{Future
  Internet}, vol.~4, no.~4, pp. 971--1003, 2012.

\bibitem{kelihos}
\BIBentryALTinterwordspacing
``{Apple {IDs} Targeted by {Kelihos} Botnet Phishing Campaign}.'' [Online].
  Available:
  \url{http://www.symantec.com/connect/blogs/apple-ids-targeted-kelihos-botnet-phishing-campaign}
\BIBentrySTDinterwordspacing

\bibitem{nap}
\BIBentryALTinterwordspacing
``An {Encounter} with {Trojan} {Nap},'' \emph{FireEye Threat Research Blog}.
  [Online]. Available:
  \url{https://www.fireeye.com/blog/threat-research/2013/02/an-encounter-with-trojan-nap.html}
\BIBentrySTDinterwordspacing

\bibitem{poisonivy}
\BIBentryALTinterwordspacing
J.~T. Bennett, N.~Moran, and N.~Villeneuve, ``{Poison Ivy: Assessing Damage and
  Extracting Intelligence},'' \emph{FireEye Threat Research Blog}, 2013.
  [Online]. Available:
  \url{http://www.FireEye.com/resources/pdfs/FireEye-poison-ivy-report.pdf}
\BIBentrySTDinterwordspacing

\bibitem{upclicker}
\BIBentryALTinterwordspacing
``{Trojan {Upclicker} Ties Malware to the Mouse},'' \emph{InfoSecurity
  Magazine}. [Online]. Available:
  \url{http://www.infosecurity-magazine.com/news/trojan-upclicker-ties-malware-to-the-mouse/}
\BIBentrySTDinterwordspacing

\bibitem{singh2013hot}
\BIBentryALTinterwordspacing
A.~Singh and Z.~Bu, ``{Hot Knives Through Butter: Evading File-based
  Sandboxes},'' \emph{Threat Research Blog}, 2013. [Online]. Available:
  \url{https://www.fireeye.com/blog/threat-research/2013/08/hot-knives-through-butter-bypassing-file-based-sandboxes.html}
\BIBentrySTDinterwordspacing

\bibitem{istr2016symantec}
Symantec, ``{Internet Security Threat Report: Trends for 2016},'' vol.~21,
  p.~81, 2016.

\bibitem{pos}
T.~Micro, ``Point-of-sale system breaches: threats to the retail and
  hospitality industries,'' Tech. rep., Trend Micro Inc. http://​ www.​
  trendmicro.​ com/​ cloud-content/​ us/​ pdfs/​
  security-intelligence/​ white-papers/​ wp-pos-system-breaches.​ pdf,
  Tech. Rep., 2014.

\bibitem{darkhotel}
\BIBentryALTinterwordspacing
``Darkhotel's attacks in 2015 - {Securelist}.'' [Online]. Available:
  \url{https://securelist.com/blog/research/71713/darkhotels-attacks-in-2015/}
\BIBentrySTDinterwordspacing

\bibitem{falliere2011w32}
N.~Falliere, L.~O. Murchu, and E.~Chien, ``{W32.Stuxnet Dossier},'' \emph{White
  paper, Symantec Corp., Security Response}, vol.~5, p.~6, 2011.

\bibitem{gauss_abnormal}
{Kaspersky Lab Global Research and Analysis Team}, ``{Gauss: Abnormal
  Distribution},'' Kaspersky Lab, Tech. Rep., 2012.

\bibitem{bencsath2012skywiper}
B.~Bencs{\'a}th, G.~P{\'e}k, L.~Butty{\'a}n, and M.~Felegyhazi, ``{sKyWIper
  (aka Flame aka Flamer): A Complex Malware for Targeted attacks},''
  \emph{CrySyS Lab Technical Report, No. CTR-2012-05-31}, 2012.

\bibitem{garfinkel2003virtual}
T.~Garfinkel and M.~Rosenblum, ``{A Virtual Machine Introspection Based
  Architecture for Intrusion Detection},'' in \emph{Proceedings of the Network
  and Distributed System Security Symposium}, 2003.

\bibitem{rutkowska2005system}
J.~Rutkowska, ``{System Virginity Verifier -- Defining the Roadmap for Malware
  Detection on Windows Systems},'' in \emph{Proceedings of the 5th Hack in the
  Box security Conference}, 2005.

\bibitem{rutkowska2004detecting}
\BIBentryALTinterwordspacing
------, ``{Detecting Windows Server Compromises with {Patchfinder 2}},'' 2004.
  [Online]. Available:
  \url{http://repo.hackerzvoice.net/depot_madchat/vxdevl/avtech/Detecting\%20Windows\%20Server\%20Compromises\%20with\%20Patchfinder\%202.pdf}
\BIBentrySTDinterwordspacing

\bibitem{rutkowska2005thoughts}
\BIBentryALTinterwordspacing
------, ``{Thoughts About Cross-View Based Rootkit Detection},'' 2005.
  [Online]. Available:
  \url{https://vxheaven.org/lib/pdf/Thoughts\%20about\%20Cross-View\%20based\%20Rootkit\%20Detection.pdf}
\BIBentrySTDinterwordspacing

\bibitem{butler2005windows3}
\BIBentryALTinterwordspacing
J.~Butler and S.~Sparks, ``{Windows Rootkits of 2005, Part Three},'' Symantec
  Connect Community, Tech. Rep., 2005. [Online]. Available:
  \url{http://www.symantec.com/connect/articles/windows-rootkits-2005-part-three}
\BIBentrySTDinterwordspacing

\bibitem{cogswell2006rootkitrevealer}
B.~Cogswell and M.~Russinovich, ``{RootkitRevealer},'' \emph{Rootkit detection
  tool by Microsoft}, vol.~1, 2006.

\bibitem{butler2006raide}
J.~Butler and P.~Silberman, ``{RAIDE: Rootkit Analysis Identification
  Elimination},'' \emph{Black Hat USA}, vol.~47, 2006.

\bibitem{petroni2006architecture}
N.~L. Petroni~Jr., T.~Fraser, A.~Walters, and W.~A. Arbaugh, ``{An Architecture
  for Specification-Based Detection of Semantic Integrity Violations in Kernel
  Dynamic Data},'' in \emph{Proceedings of the 15th Conference on USENIX
  Security Symposium}, vol.~2, 2006.

\bibitem{siddiqui2008survey}
M.~Siddiqui, M.~C. Wang, and J.~Lee, ``{A Survey of Data Mining Techniques for
  Malware Detection Using File Features},'' in \emph{Proceedings of the 46th
  Annual Southeast Regional Conference}.\hskip 1em plus 0.5em minus 0.4em\relax
  ACM, 2008, pp. 509--510.

\bibitem{dolangavitt2009robust}
B.~Dolan-Gavitt, A.~Srivastava, P.~Traynor, and J.~Giffin, ``{Robust Signatures
  for Kernel Data Structures},'' in \emph{Proceedings of the 16th ACM
  Conference on Computer and Communications Security}.\hskip 1em plus 0.5em
  minus 0.4em\relax ACM, 2009, pp. 566--577.

\bibitem{seshadri2007secvisor}
A.~Seshadri, M.~Luk, N.~Qu, and A.~Perrig, ``{{SecVisor}: Tiny Hypervisor to
  Provide Lifetime Kernel Code Integrity for Commodity {OSes}},'' \emph{SIGOPS
  Operating Systems Review}, vol.~41, no.~6, pp. 335--350, 2007.

\bibitem{riley2008guest}
R.~Riley, X.~Jiang, and D.~Xu, ``{Guest-Transparent Prevention of Kernel
  Rootkits with VMM-Based Memory Shadowing},'' in \emph{Proceedings of the 11th
  International Symposium on Recent Advances in Intrusion Detection}.\hskip 1em
  plus 0.5em minus 0.4em\relax Springer, 2008, pp. 1--20.

\bibitem{nethercote2007valgrind}
N.~Nethercote and J.~Seward, ``{Valgrind: A Framework for Heavyweight Dynamic
  Binary Instrumentation},'' \emph{ACM SIGPLAN Notices}, vol.~42, pp. 89--100,
  2007.

\bibitem{riley2009multi}
R.~Riley, X.~Jiang, and D.~Xu, ``{Multi-Aspect Profiling of Kernel Rootkit
  Behavior},'' in \emph{Proceedings of the 4th ACM European Conference on
  Computer Systems}.\hskip 1em plus 0.5em minus 0.4em\relax ACM, 2009, pp.
  47--60.

\bibitem{ganapathy2005automatic}
V.~Ganapathy, T.~Jaeger, and S.~Jha, ``{Automatic Placement of Authorization
  Hooks in the {Linux} Security Modules Framework},'' in \emph{Proceedings of
  the 12th ACM Conference on Computer and Communications Security}.\hskip 1em
  plus 0.5em minus 0.4em\relax ACM, 2005, pp. 330--339.

\bibitem{jang2016detecting}
J.-w. Jang, J.~Yun, A.~Mohaisen, J.~Woo, and H.~K. Kim, ``{Detecting and
  Classifying Method Based on Similarity Matching of Android Malware Behavior
  with Profile},'' \emph{SpringerPlus}, vol.~5, no.~1, p.~1, 2016.

\bibitem{abouassaleh2004ngram}
T.~Abou-Assaleh, N.~Cercone, V.~Keselj, and R.~Sweidan, ``{N-Gram-Based
  Detection of New Malicious Code},'' in \emph{Proceedings of the 28th Annual
  International Computer Software and Applications Conference}, vol.~2.\hskip
  1em plus 0.5em minus 0.4em\relax IEEE, 2004, pp. 41--42.

\bibitem{reddy2005new}
D.~K.~S. Reddy, S.~K. Dash, and A.~K. Pujari, ``{New Malicious Code Detection
  Using Variable Length N-Grams},'' in \emph{Proceedings of the Second
  International Conference on Information Systems Security}.\hskip 1em plus
  0.5em minus 0.4em\relax Springer, 2006, pp. 276--288.

\bibitem{chow2005shredding}
J.~Chow, B.~Pfaff, T.~Garfinkel, and M.~Rosenblum, ``{Shredding Your Garbage:
  Reducing Data Lifetime Through Secure Deallocation},'' in \emph{Proceedings
  of the 14th Conference on USENIX Security Symposium}.\hskip 1em plus 0.5em
  minus 0.4em\relax USENIX Association, 2005, p.~22.

\bibitem{schuster2006searching}
A.~Schuster, ``{Searching for Processes and Threads in {M}icrosoft {W}indows
  Memory Dumps},'' \emph{Digital Investigation}, vol.~3, no.~1, pp. 10--16,
  2006.

\bibitem{feng2003anomaly}
H.~H. Feng, O.~M. Kolesnikov, P.~Fogla, W.~Lee, and W.~Gong, ``{Anomaly
  Detection Using Call Stack Information},'' in \emph{Proceedings of the 2003
  IEEE Symposium on Security and Privacy}.\hskip 1em plus 0.5em minus
  0.4em\relax IEEE Computer Society, 2003, pp. 62--75.

\bibitem{forrest1996sense}
S.~Forrest, S.~A. Hofmeyr, A.~Somayaji, and T.~A. Longstaff, ``{A Sense of Self
  for {UNIX} Processes},'' in \emph{Proceedings of the 1996 IEEE Symposium on
  Security and Privacy}.\hskip 1em plus 0.5em minus 0.4em\relax IEEE Computer
  Society, 1996, pp. 120--128.

\bibitem{giffin2002detecting}
J.~T. Giffin, S.~Jha, and B.~P. Miller, ``{Detecting Manipulated Remote Call
  Streams},'' in \emph{Proceedings of the 11th USENIX Security Symposium},
  2002, pp. 61--79.

\bibitem{hofmeyr1998intrusion}
S.~A. Hofmeyr, S.~Forrest, and A.~Somayaji, ``{Intrusion Detection Using
  Sequences of System Calls},'' \emph{Journal of Computer Security}, vol.~6,
  no.~3, pp. 151--180, 1998.

\bibitem{krohn2007information}
M.~Krohn, A.~Yip, M.~Brodsky, N.~Cliffer, M.~F. Kaashoek, E.~Kohler, and
  R.~Morris, ``{Information Flow Control for Standard {OS} Abstractions},'' in
  \emph{SIGOPS Operating Systems Review}, vol.~41, no.~6.\hskip 1em plus 0.5em
  minus 0.4em\relax ACM, 2007, pp. 321--334.

\bibitem{mutz2007exploiting}
D.~Mutz, W.~K. Robertson, G.~Vigna, and R.~A. Kemmerer, ``{Exploiting Execution
  Context for the Detection of Anomalous System Calls},'' in \emph{Proceedings
  of the 10th International Symposium on Recent Advances in Intrusion
  Detection}.\hskip 1em plus 0.5em minus 0.4em\relax Springer, 2007, pp. 1--20.

\bibitem{sekar2001fast}
R.~Sekar, M.~Bendre, D.~Dhurjati, and P.~Bollineni, ``{A Fast Automaton-Based
  Method for Detecting Anomalous Program Behaviors},'' in \emph{Symposium on
  Security and Privacy}.\hskip 1em plus 0.5em minus 0.4em\relax IEEE, 2001, pp.
  144--155.

\bibitem{yu2015fool}
S.~Yu, S.~Guo, and I.~Stojmenovic, ``{Fool Me if You Can: Mimicking Attacks and
  Anti-Attacks in Cyberspace},'' \emph{IEEE Transactions on Computers},
  vol.~64, no.~1, pp. 139--151, 2015.

\bibitem{yu2015modeling}
S.~Yu, G.~Wang, and W.~Zhou, ``{Modeling Malicious Activities in Cyber
  Space},'' \emph{IEEE Network}, vol.~29, no.~6, pp. 83--87, 2015.

\bibitem{fink2005visual}
G.~Fink, P.~Muessig, and C.~North, ``{Visual Correlation of Host Processes and
  Network Traffic},'' in \emph{IEEE Workshop on Visualization for Computer
  Security}, 2005, pp. 11--19.

\bibitem{bishop2006pattern}
C.~M. Bishop, \emph{Pattern Recognition and Machine Learning (Information
  Science and Statistics)}.\hskip 1em plus 0.5em minus 0.4em\relax Springer,
  2006.

\bibitem{ogorman2012ransomware}
\BIBentryALTinterwordspacing
G.~O'Gorman and G.~McDonald, ``{Ransomware: A Growing Menace},'' 2012.
  [Online]. Available:
  \url{http://www.symantec.com/content/en/us/enterprise/media/security_response/whitepapers/ransomware-a-growing-menace.pdf}
\BIBentrySTDinterwordspacing

\bibitem{venkatesan2008code}
A.~Venkatesan, ``{Code Obfuscation and Virus Detection},'' Master's thesis, San
  Jos{\'e} State University, 2008.

\bibitem{venkatachalam2010detecting}
S.~Venkatachalam, ``{Detecting Undetectable Computer Viruses},'' Master's
  thesis, San Jos{\'e} State University, 2010.

\bibitem{runwal2012opcode}
N.~Runwal, R.~M. Low, and M.~Stamp, ``{Opcode Graph Similarity and Metamorphic
  Detection},'' \emph{Journal in Computer Virology}, vol.~8, no.~1, pp. 37--52,
  2012.

\bibitem{lin2011hunting}
D.~Lin and M.~Stamp, ``{Hunting for Undetectable Metamorphic Viruses},''
  \emph{Journal in Computer Virology}, vol.~7, no.~3, pp. 201--214, 2011.

\bibitem{desai2008towards}
P.~Desai, ``{Towards an Undetectable Computer Virus},'' Master's thesis, San
  Jos{\'e} State University, 2008.

\bibitem{wong2006analysis}
W.~Wong, ``{Analysis and Detection of Metamorphic Computer Viruses},'' Master's
  thesis, San Jos{\'e} State University, 2006.

\bibitem{wong2006hunting}
W.~Wong and M.~Stamp, ``{Hunting for Metamorphic Engines},'' \emph{Journal in
  Computer Virology}, vol.~2, no.~3, pp. 211--229, 2006.

\bibitem{attaluri2009profile}
S.~Attaluri, S.~McGhee, and M.~Stamp, ``{Profile Hidden Markov Models and
  Metamorphic Virus Detection},'' \emph{Journal in computer virology}, vol.~5,
  no.~2, pp. 151--169, 2009.

\bibitem{zanero2004unsupervised}
S.~Zanero and S.~M. Savaresi, ``{Unsupervised Learning Techniques for an
  Intrusion Detection System},'' in \emph{Proceedings of the 2004 ACM symposium
  on Applied computing}, 2004, pp. 412--419.

\bibitem{scheirer2013towards}
W.~J. Scheirer, A.~Rocha, A.~Sapkota, and T.~E. Boult, ``{Towards Open Set
  Recognition},'' \emph{IEEE Transactions on Pattern Analysis and Machine
  Intelligence}, vol.~36, no.~7, 2013.

\bibitem{masud2011classification}
M.~M. Masud, J.~Gao, L.~Khan, J.~Han, and B.~Thuraisingham, ``{Classification
  and Novel Class Detection in Concept-Drifting Data Streams Under Time
  Constraints},'' \emph{IEEE Transactions on Knowledge and Data Engineering},
  vol.~23, no.~6, pp. 859--874, 2011.

\bibitem{mukkamala2002intrusion}
S.~Mukkamala, G.~Janoski, and A.~Sung, ``{Intrusion Detection Using Neural
  Networks and Support Vector Machines},'' in \emph{Proceedings of the 2002
  International Joint Conference on Neural Networks}, vol.~2.\hskip 1em plus
  0.5em minus 0.4em\relax IEEE, 2002, pp. 1702--1707.

\bibitem{catania2012automatic}
C.~A. Catania and C.~G. Garino, ``{Automatic Network Intrusion Detection:
  Current Techniques and Open Issues},'' \emph{Computers \& Electrical
  Engineering}, vol.~38, no.~5, pp. 1062--1072, 2012.

\bibitem{lee1998data}
W.~Lee and S.~J. Stolfo, ``{Data Mining Approaches for Intrusion Detection},''
  in \emph{Proceedings of the 7th Conference on USENIX Security
  Symposium}.\hskip 1em plus 0.5em minus 0.4em\relax USENIX Association, 1998,
  p.~6.

\bibitem{portnoy2001intrusion}
L.~Portnoy, E.~Eskin, and S.~Stolfo, ``{Intrusion Detection with Unlabeled Data
  Using Clustering},'' in \emph{Proceedings of ACM CSS Workshop on Data Mining
  Applied to Security}, 2001.

\bibitem{lazarevic2003comparative}
A.~Lazarevic, L.~Ert{\"o}z, V.~Kumar, A.~Ozgur, and J.~Srivastava, ``{A
  Comparative Study of Anomaly Detection Schemes in Network Intrusion
  Detection},'' in \emph{Proceedings of the Third {SIAM} International
  Conference on Data Mining}, 2003, pp. 25--6.

\bibitem{ertoz2004minds}
L.~Ert{\"o}z, E.~Eilertson, A.~Lazarevic, P.-N. Tan, V.~Kumar, J.~Srivastava,
  and P.~Dokas, ``{The {MINDS} -- Minnesota Intrusion Detection System},'' in
  \emph{Next Generation Data Mining}.\hskip 1em plus 0.5em minus 0.4em\relax
  MIT Press, 2004, ch.~3.

\bibitem{rehak2009adaptive}
M.~Rehak, M.~Pechoucek, M.~Grill, J.~Stiborek, K.~Barto{\v{s}}, and P.~Celeda,
  ``{Adaptive Multiagent System for Network Traffic Monitoring},'' \emph{IEEE
  Intelligent Systems}, vol.~24, no.~3, pp. 16--25, 2009.

\bibitem{aggarwal2007data}
C.~C. Aggarwal, Ed., \emph{Data Streams - Models and Algorithms}, ser. Advances
  in Database Systems.\hskip 1em plus 0.5em minus 0.4em\relax Springer, 2007,
  vol.~31.

\bibitem{helmer1998intelligent}
G.~G. Helmer, J.~S.~K. Wong, V.~Honavar, and L.~Miller, ``{Intelligent Agents
  for Intrusion Detection},'' in \emph{IEEE Information Technology
  Conference}.\hskip 1em plus 0.5em minus 0.4em\relax IEEE, 1998, pp. 121--124.

\bibitem{ahmad2011feature}
I.~Ahmad, A.~B. Abdulah, A.~S. Alghamdi, K.~Alnfajan, and M.~Hussain,
  ``{Feature Subset Selection for Network Intrusion Detection Mechanism Using
  Genetic Eigenvectors},'' in \emph{Proceedings of 2011 International
  Conference on Telecommunication Technology and Applications}.\hskip 1em plus
  0.5em minus 0.4em\relax LACSIT Press, 2011, pp. 75--79.

\bibitem{nguyen2010improving}
H.~Nguyen, K.~Franke, and S.~Petrovi{\'c}, ``{Improving Effectiveness of
  Intrusion Detection by Correlation Feature Selection},'' in \emph{Fifth
  International Conference on Availability, Reliability and Security}.\hskip
  1em plus 0.5em minus 0.4em\relax IEEE, 2010, pp. 17--24.

\bibitem{lakhina2010feature}
S.~Lakhina, S.~Joseph, and B.~Verma, ``{Feature Reduction Using Principal
  Component Analysis for Effective Anomaly-Based Intrusion Detection on
  {NSL-KDD}},'' \emph{International Journal of Engineering Science and
  Technology}, pp. 1790--1799, 2010.

\bibitem{middlemiss2003feature}
M.~Middlemiss and G.~Dick, ``{Feature Selection of Intrusion Detection Data
  Using a Hybrid Genetic Algorithm/{KNN} Approach},'' in \emph{Design and
  Application of Hybrid Intelligent Systems}, A.~Abraham, M.~K\"{o}ppen, and
  K.~Franke, Eds.\hskip 1em plus 0.5em minus 0.4em\relax IOS Press, 2003, pp.
  519--527.

\bibitem{yu2010feature}
H.-F. Yu \emph{et~al.}, ``{Feature Engineering and Classifier Ensemble for
  {KDD} Cup 2010},'' in \emph{Proceedings of the KDD Cup 2010 Workshop}, 2010,
  pp. 1--16.

\bibitem{mukkamala2003feature}
S.~Mukkamala and A.~Sung, ``{Feature Selection for Intrusion Detection with
  Neural Networks and Support Vector Machines},'' \emph{Transportation Research
  Record}, vol. 1822, no.~1, pp. 33--39, 2003.

\bibitem{stein2005decision}
G.~Stein, B.~Chen, A.~S. Wu, and K.~A. Hua, ``{Decision Tree Classifier for
  Network Intrusion Detection with {GA}-Based Feature Selection},'' in
  \emph{Proceedings of the 43rd Annual Southeast Regional Conference},
  vol.~2.\hskip 1em plus 0.5em minus 0.4em\relax ACM, 2005, pp. 136--141.

\bibitem{sung2003identifying}
A.~H. Sung and S.~Mukkamala, ``{Identifying Important Features for Intrusion
  Detection Using Support Vector Machines and Neural Networks},'' in
  \emph{Proceedings of the 2003 Symposium on Applications and the
  Internet}.\hskip 1em plus 0.5em minus 0.4em\relax IEEE Computer Society,
  2003, pp. 209--216.

\bibitem{wang2010new}
G.~Wang, J.~Hao, J.~Ma, and L.~Huang, ``{A New Approach to Intrusion Detection
  Using Artificial Neural Networks and Fuzzy Clustering},'' \emph{Expert
  Systems with Applications}, vol.~37, no.~9, pp. 6225--6232, 2010.

\bibitem{amini2006rt}
M.~Amini, R.~Jalili, and H.~R. Shahriari, ``{{RT-UNNID}: A Practical Solution
  to Real-time Network-based Intrusion Detection Using Unsupervised Neural
  Networks},'' \emph{Computers \& Security}, vol.~25, no.~6, pp. 459--468,
  2006.

\bibitem{liu2007letters}
G.~Liu, Z.~Yi, and S.~Yang, ``{Letters: A Hierarchical Intrusion Detection
  Model Based on the {PCA} Neural Networks},'' \emph{Neurocomputing}, vol.~70,
  no. 7-9, pp. 1561--1568, 2007.

\bibitem{srinivasan2006self}
T.~Srinivasan, V.~Vijaykumar, and R.~Chandrasekar, ``{A Self-Organized
  Agent-Based Architecture for Power-Aware Intrusion Detection in Wireless
  Ad-Hoc Networks},'' in \emph{International Conference on Computing \&
  Informatics}.\hskip 1em plus 0.5em minus 0.4em\relax IEEE, 2006, pp. 1--6.

\bibitem{shun2008network}
J.~Shun and H.~A. Malki, ``{Network Intrusion Detection System Using Neural
  Networks},'' in \emph{Proceedings of the 2008 Fourth International Conference
  on Natural Computation}, vol.~5.\hskip 1em plus 0.5em minus 0.4em\relax IEEE
  Computer Society, 2008, pp. 242--246.

\bibitem{nguyen2015deep}
A.~Nguyen, J.~Yosinski, and J.~Clune, ``{Deep Neural Networks Are Easily
  Fooled: High Confidence Predictions for Unrecognizable Images},'' \emph{The
  IEEE Conference on Computer Vision and Pattern Recognition}, 2015.

\bibitem{szegedy2014intriguing}
C.~Szegedy, W.~Zaremba, I.~Sutskever, J.~Bruna, D.~Erhan, I.~J. Goodfellow, and
  R.~Fergus, ``{Intriguing Properties of Neural Networks},'' in
  \emph{International Conference on Learning Representations}, 2014.

\bibitem{goodfellow2015explaining}
I.~J. Goodfellow, J.~Shlens, and C.~Szegedy, ``{Explaining and Harnessing
  Adversarial Examples},'' in \emph{International Conference on Learning
  Representations}, 2015.

\bibitem{jensen2007bayesian}
F.~V. Jensen and T.~D. Nielsen, \emph{Bayesian Networks and Decision Graphs},
  2nd~ed.\hskip 1em plus 0.5em minus 0.4em\relax Springer Publishing Company,
  2007.

\bibitem{kotz2000extreme}
S.~Kotz and S.~Nadarajah, \emph{Extreme Value Distributions: Theory and
  Applications}.\hskip 1em plus 0.5em minus 0.4em\relax Imperial College Press,
  2000.

\bibitem{markou2003novelty}
M.~Markou and S.~Singh, ``{Novelty Detection: A Review -- Part 1: Statistical
  Approaches},'' \emph{Signal Processing}, vol.~83, no.~12, pp. 2481--2497,
  2003.

\bibitem{scheirer2014probability}
W.~J. Scheirer, L.~P. Jain, and T.~E. Boult, ``{Probability Models for Open Set
  Recognition},'' \emph{IEEE Transactions on Pattern Analysis and Machine
  Intelligence}, vol.~36, pp. 2317--2324, 2014.

\bibitem{jain2014multiclass}
L.~P. Jain, W.~J. Scheirer, and T.~E. Boult, ``{Multi-Class Open Set
  Recognition Using Probability of Inclusion},'' in \emph{Proceedings of the
  European Conference on Computer Vision}, 2014.

\bibitem{bendale2015towards}
A.~Bendale and T.~E. Boult, ``{Towards Open World Recognition},'' in \emph{The
  IEEE Conference on Computer Vision and Pattern Recognition}, 2015.

\bibitem{yeung2002parzen}
D.-Y. Yeung and C.~Chow, ``{Parzen-Window Network Intrusion Detectors},'' in
  \emph{Proceedings of the 16th International Conference on Pattern
  Recognition}, vol.~4.\hskip 1em plus 0.5em minus 0.4em\relax IEEE, 2002, pp.
  385--388.

\bibitem{alizadeh2015traffic}
H.~Alizadeh, A.~Khoshrou, and A.~Z{\'u}quete, ``{Traffic Classification and
  Verification using Unsupervised Learning of {G}aussian Mixture Models},'' in
  \emph{International Workshop on Measurements \& Networking}.\hskip 1em plus
  0.5em minus 0.4em\relax IEEE, 2015.

\bibitem{fan2013anomaly}
W.~Fan, N.~Bouguila, and H.~Sallay, ``{Anomaly Intrusion Detection Using
  Incremental Learning of an Infinite Mixture Model with Feature Selection},''
  in \emph{Proceedings of the 8th International Conference on Rough Sets and
  Knowledge Technology}.\hskip 1em plus 0.5em minus 0.4em\relax Springer, 2013,
  pp. 364--373.

\bibitem{gruhl2015building}
C.~Gruhl, B.~Sick, A.~Wacker, S.~Tomforde, and J.~H{\"a}hner, ``{A Building
  Block for Awareness in Technical Systems: Online Novelty Detection and
  Reaction With an Application in Intrusion Detection},'' in \emph{Proceedings
  of the 7th International Conference on Awareness Science and
  Technology}.\hskip 1em plus 0.5em minus 0.4em\relax IEEE, 2015.

\bibitem{lam2015outlier}
\BIBentryALTinterwordspacing
P.~Lam, L.-L. Wang, H.~Y.~T. Ngan, N.~H.~C. Yung, and A.~G.-O. Yeh, ``{Outlier
  Detection In Large-Scale Traffic Data By Na{\"i}ve {B}ayes Method and
  Gaussian Mixture Model Method},'' \emph{arXiv preprint}, 2015. [Online].
  Available: \url{http://arxiv.org/abs/1512.08413}
\BIBentrySTDinterwordspacing

\bibitem{yamanishi2004line}
K.~Yamanishi, J.-I. Takeuchi, G.~Williams, and P.~Milne, ``{Online Unsupervised
  Outlier Detection Using Finite Mixtures with Discounting Learning
  Algorithms},'' \emph{Data Mining and Knowledge Discovery}, vol.~8, no.~3, pp.
  275--300, 2004.

\bibitem{yeung2003host}
D.-Y. Yeung and Y.~Ding, ``{Host-Based Intrusion Detection Using Dynamic and
  Static Behavioral Models},'' \emph{Pattern Recognition}, vol.~36, no.~1, pp.
  229--243, 2003.

\bibitem{amer2013enhancing}
M.~Amer, M.~Goldstein, and S.~Abdennadher, ``{Enhancing One-Class Support
  Vector Machines for Unsupervised Anomaly Detection},'' in \emph{Proceedings
  of the ACM SIGKDD Workshop on Outlier Detection and Description}.\hskip 1em
  plus 0.5em minus 0.4em\relax ACM, 2013, pp. 8--15.

\bibitem{yang2015adaptive}
J.~Yang, T.~Deng, and R.~Sui, ``{An Adaptive Weighted One-Class SVM for Robust
  Outlier Detection},'' in \emph{Proceedings of the 2015 Chinese Intelligent
  Systems Conference}.\hskip 1em plus 0.5em minus 0.4em\relax Springer, 2015,
  pp. 475--484.

\bibitem{heller2003one}
K.~A. Heller, K.~M. Svore, A.~D. Keromytis, and S.~J. Stolfo, ``{One Class
  Support Vector Machines for Detecting Anomalous Windows Registry Accesses},''
  in \emph{Proceedings of the Workshop on Data Mining for Computer Security},
  2003.

\bibitem{wang2004anomaly}
Y.~Wang, J.~Wong, and A.~Miner, ``{Anomaly Intrusion Detection Using One-Class
  {SVM}},'' in \emph{Proceedings from the Fifth Annual IEEE SMC Information
  Assurance Workshop}, 2004, pp. 358--364.

\bibitem{li2003improving}
K.-L. Li, H.-K. Huang, S.-F. Tian, and W.~Xu, ``{Improving One-Class {SVM} for
  Anomaly Detection},'' in \emph{International Conference on Machine Learning
  and Cybernetics}, vol.~5.\hskip 1em plus 0.5em minus 0.4em\relax IEEE, 2003,
  pp. 3077--3081.

\bibitem{perdisci2006using}
R.~Perdisci, G.~Gu, and W.~Lee, ``{Using an Ensemble of One-Class {SVM}
  Classifiers to Harden Payload-Based Anomaly Detection Systems},'' in
  \emph{Sixth International Conference on Data Mining}.\hskip 1em plus 0.5em
  minus 0.4em\relax IEEE, 2006, pp. 488--498.

\bibitem{scholkopf2001estimating}
B.~Sch{\"o}lkopf, J.~C. Platt, J.~C. Shawe-Taylor, A.~J. Smola, and R.~C.
  Williamson, ``{Estimating the Support of a High-Dimensional Distribution},''
  \emph{Neural Computation}, vol.~13, no.~7, pp. 1443--1471, 2001.

\bibitem{rudd2015extreme}
\BIBentryALTinterwordspacing
E.~M. Rudd, L.~P. Jain, W.~J. Scheirer, and T.~E. Boult, ``{The Extreme Value
  Machine},'' \emph{arXiv preprint}, 2015. [Online]. Available:
  \url{http://arxiv.org/abs/1506.06112}
\BIBentrySTDinterwordspacing

\bibitem{lane1998approaches}
T.~Lane and C.~E. Brodley, ``{Approaches to Online Learning and Concept Drift
  for User Identification in Computer Security},'' in \emph{Proceedings of the
  Fourth International Conference on Knowledge Discovery and Data Mining},
  1998, pp. 259--263.

\bibitem{karthick2012adaptive}
R.~R. Karthick, V.~P. Hattiwale, and B.~Ravindran, ``{Adaptive Network
  Intrusion Detection System Using a Hybrid Approach},'' in \emph{Fourth
  International Conference on Communication Systems and Networks}.\hskip 1em
  plus 0.5em minus 0.4em\relax IEEE, 2012, pp. 1--7.

\bibitem{wang2004anomalous}
K.~Wang and S.~J. Stolfo, ``{Anomalous Payload-Based Network Intrusion
  Detection},'' in \emph{Proceedings of the 7th International Symposium on
  Recent Advances in Intrusion Detection}.\hskip 1em plus 0.5em minus
  0.4em\relax Springer, 2004, pp. 203--222.

\bibitem{zhong2007clustering}
S.~Zhong, T.~M. Khoshgoftaar, and N.~Seliya, ``{Clustering-Based Network
  Intrusion Detection},'' \emph{International Journal of Reliability, Quality
  and Safety Engineering}, vol.~14, no.~2, pp. 169--187, 2007.

\bibitem{cannady2000applying}
J.~Cannady, ``{Applying {CMAC}-Based Online Learning to Intrusion Detection},''
  in \emph{Proceedings of the IEEE-INNS-ENNS International Joint Conference on
  Neural Networks}, vol.~5, 2000, pp. 405--410.

\bibitem{hu2014online}
W.~Hu, J.~Gao, Y.~Wang, O.~Wu, and S.~Maybank, ``{Online {A}daboost-Based
  Parameterized Methods for Dynamic Distributed Network Intrusion Detection},''
  \emph{IEEE Transactions on Cybernetics}, vol.~44, no.~1, pp. 66--82, 2014.

\bibitem{wang2013concept}
S.~Wang, L.~L. Minku, D.~Ghezzi, D.~Caltabiano, P.~Tino, and X.~Yao, ``{Concept
  Drift Detection for Online Class Imbalance Learning},'' in \emph{IEEE
  International Joint Conference on Neural Networks}, 2013, pp. 1--10.

\bibitem{scheirer2010robust}
W.~J. Scheirer, A.~Rocha, R.~Michaels, and T.~E. Boult, ``{Robust Fusion:
  Extreme Value Theory for Recognition Score Normalization},'' in
  \emph{Proceedings of the 11th European Conference on Computer Vision}, 2010.

\bibitem{scheirer2012multi}
W.~J. Scheirer, N.~Kumar, P.~N. Belhumeur, and T.~E. Boult, ``{Multi-Attribute
  Spaces: Calibration for Attribute Fusion and Similarity Search},'' in
  \emph{Proceedings of the 25th IEEE Conference on Computer Vision and Pattern
  Recognition}, 2012, pp. 2933--2940.

\bibitem{creech2014semantic}
G.~Creech and J.~Hu, ``{A Semantic Approach to Host-Based Intrusion Detection
  Systems Using Contiguous and Discontiguous System Call Patterns},''
  \emph{IEEE Transactions on Computers}, vol.~63, no.~4, pp. 807--819, 2014.

\bibitem{creech2013generation}
------, ``{Generation of a New IDS Test Dataset: Time to Retire the KDD
  Collection},'' in \emph{2013 IEEE Wireless Communications and Networking
  Conference (WCNC)}.\hskip 1em plus 0.5em minus 0.4em\relax IEEE, 2013, pp.
  4487--4492.

\bibitem{creech2014developing}
G.~Creech, ``{Developing a High-Accuracy Cross Platform Host-Based Intrusion
  Detection System Capable of Reliably Detecting Zero-Day Attacks},'' Ph.D.
  dissertation, University of New South Wales, 2014.

\bibitem{MSDN}
\BIBentryALTinterwordspacing
{Microsoft Developer Network}, ``{API index}.'' [Online]. Available:
  \url{https://msdn.microsoft.com/en-us/library/windows/desktop/hh920508(v=vs.85).aspx}
\BIBentrySTDinterwordspacing

\bibitem{Wikipedia}
\BIBentryALTinterwordspacing
{Wikipedia}, ``{Microsoft Windows library files}.'' [Online]. Available:
  \url{https://en.wikipedia.org/wiki/Microsoft_Windows_library_files}
\BIBentrySTDinterwordspacing

\end{thebibliography}

%
%
%
%

\begin{IEEEbiographynophoto}{Ethan Rudd}
Ethan Rudd received his B.S. degree in Physics from Trinity University in 2012.
He received his M.S. degree in Computer Science from the University of Colorado at Colorado Springs in 2014.
He is currently pursuing his Ph.D. in Computer Science.
He works in the Vision and Security Technology (VAST) lab, conducting research in computer vision, biometrics, and applied machine learning.
\end{IEEEbiographynophoto}

\begin{IEEEbiographynophoto}{Andras Rozsa}
Andras Rozsa received his M.Eng. degree in Information Technology from University of Veszprem, Hungary, in 2005. He received his M.S. degree in Computer Science from the University of Colorado at Colorado Springs (UCCS) in 2014. He is currently pursuing his Ph.D. in Engineering with a specialty in Security. He works in the Vision and Security Technology (VAST) lab at UCCS, researching machine learning and deep neural network applications to computer vision.
\end{IEEEbiographynophoto}


\begin{IEEEbiographynophoto}{Manuel G\"unther}
Dr. Manuel G\"unther received his diploma in Computer Science with a major subject of machine learning from the Technical University of Ilmenau, Germany, in 2004.
His doctoral thesis was written between 2004 and 2011 at the Ruhr University of Bochum, Germany, about statistical extensions of Gabor graph based face detection, recognition and classification techniques.
Between 2012 and 2015, Dr. G\"unther was a postdoctoral researcher in the Biometrics Group at the Idiap Research Institute in Martigny, Switzerland.
There, he was actively participating in the implementation of the open source signal processing and machine learning library Bob, particularly he was the leading developer of the bob.bio packages (\url{http://pythonhosted.org/bob.bio.base}).
Since 2015, Dr. G\"unther has been a research associate at the Vision and Security Technology lab at the University of Colorado at Colorado Springs.
Currently, he is part of the IARPA Janus research team and is occupied with incorporating facial attributes to build more reliable face recognition algorithms in uncontrolled imaging environments.

\end{IEEEbiographynophoto}
\begin{IEEEbiographynophoto}{Terrance Boult}
Dr. Terry Boult, El Pomar Professor of Innovation and Security at University of
Colorado Colorado Springs (UCCS),  does research in computer vision, machine
learning, biometrics and security. Prior to joining UCCS in 2003, he was an
endowed professor and founding chairman of Lehigh University's CSE Department
and from 1986-1992 was a faculty member at Columbia University.  He is an innovator with
a passion for combining teaching, research and business. He has won multiple
teaching, research, innovation, and entrepreneurial awards.  At University of
Colorado at Colorado Springs he was the architect of the awarding winning
Bachelor of Innovation\texttrademark family of degrees and a key member in founding the
UCCS Ph.D. in Engineering Security.   Dr. Boult has been involved with multiple
start up companies in the security space.   
\end{IEEEbiographynophoto}



\end{document}